\documentclass[
final
]{dmtcs-episciences}

\usepackage[utf8]{inputenc}
\usepackage{subfigure}

\usepackage{amsmath,amssymb,stmaryrd,dsfont,paralist,amsthm,eurosym}
\usepackage[format=default,labelfont=bf]{caption}
\usepackage{microtype}
\usepackage[colorinlistoftodos]{todonotes}
\usepackage{hyperref}
\usepackage{multicol} 
\usepackage{stmaryrd,dsfont}
\usepackage{todonotes}
\usepackage{makeidx}         
\usepackage{wasysym}         

\newtheorem{theorem}{Theorem}[section]
\newtheorem{lemma}[theorem]{Lemma}
\newtheorem{proposition}[theorem]{Proposition}
\newtheorem{corollary}[theorem]{Corollary}
\newtheorem{observation}[theorem]{Observation}
\newtheorem{example}[theorem]{Example}

\newcommand{\cA}{\mathcal A}
\newcommand{\cB}{\mathcal B}

\newcommand{\cD}{\mathcal D}

\newcommand{\cG}{\mathcal G}

\newcommand{\cL}{\mathcal L}

\newcommand{\cP}{\mathcal P}

\newcommand{\cV}{\mathcal V}

\newcommand{\rc}{\mathrm{c}}
\newcommand{\rf}{\fin}

\newcommand{\rP}{\mathrm{P}}

\DeclareMathOperator{\n}{-}

\newcommand{\LL}{\mathrm{L}}
\newcommand{\GLL}{\mathrm{GL}}

\renewcommand{\phi}{\varphi}
\renewcommand{\epsilon}{\varepsilon}

\newcommand{\sem}[1]{[\![#1]\!]}

\DeclareMathOperator{\m}{-}

\newcommand{\nat}{\mathbb{N}}
\newcommand{\init}{\mathrm{in}}
\newcommand{\fin}{\mathrm{fin}}
\newcommand{\Ae}{\mathit{Ae}}
\newcommand{\con}{\mathrm{conv}}

\newcommand{\tup}[1]{\langle #1 \rangle}

\newcommand{\MSOL}{\mathrm{MSOL}}

\newcommand{\lab}{\mathrm{lab}}
\newcommand{\nlab}{\mathrm{\ell}}
\newcommand{\edge}{\mathrm{edge}}
\newcommand{\eq}{\mathrm{eq}_\Ae}
\newcommand{\ec}{\mathrm{ec}}
\newcommand{\unique}{\mathrm{uniquebar}}
\newcommand{\first}{\mathrm{first}}
\newcommand{\last}{\mathrm{last}}
\newcommand{\nnext}{\mathrm{next}}
\newcommand{\Free}{\mathrm{Free}}
\newcommand{\closed}{\mathrm{closed}}
\newcommand{\ppath}{\mathrm{path}}
\newcommand{\sstring}{\mathrm{string}}
\newcommand{\sstringlike}{\mathrm{string}\mathord{\m}\mathrm{like}}

\newcommand{\ssucc}{\mathrm{succ}}
\newcommand{\ppred}{\mathrm{pred}}

\newcommand{\edgr}{\mathrm{ed}\mathord{\m}\mathrm{gr}}
\newcommand{\ndgr}{\mathrm{nd}\mathord{\m}\mathrm{gr}}
\newcommand{\eeuro}{\;\text{\euro}\;}

\newcommand{\exclusive}{\mathrm{exclusive}}
\newcommand{\com}{\mathrm{com}}
\newcommand{\tr}{\mathrm{tr}}
\newcommand{\pair}{\mathrm{pair}}
\newcommand{\rel}{\mathrm{rel}}
\newcommand{\beh}{\mathrm{beh}}
\newcommand{\out}{\mathrm{out}}
\newcommand{\inn}{\mathrm{in}}

\newcommand{\true}{\mathrm{true}}
\newcommand{\false}{\mathrm{false}}

\newcommand{\moveup}{\mathrm{up}}
\newcommand{\movedown}{\mathrm{down}}

\newcommand{\TRIV}{\mathrm{TRIV}}
\newcommand{\Triv}{\mathrm{Triv}}
\newcommand{\PDop}{\mathrm{P}}

\newcommand{\ttop}{\mathrm{top}}

\newcommand{\pop}{\mathrm{pop}}
\newcommand{\push}{\mathrm{push}}

\title{A B\"uchi-Elgot-Trakhtenbrot theorem for automata with MSO graph storage}

\author{Joost Engelfriet\affiliationmark{1} \and Heiko Vogler \affiliationmark{2}}
\affiliation{
  LIACS, Leiden University, Leiden, The Netherlands\\
  Technische Universit\"at Dresden, Dresden, Germany}


\received{2019-05-03}
\revised{2020-07-06}
\accepted{2020-08-03}

\begin{document}
\publicationdetails{22}{2020}{4}{3}{5424}
\maketitle

\begin{abstract} We introduce MSO graph storage types, and call a storage type MSO-expressible if it is isomorphic to some MSO graph storage type. An MSO graph storage type has MSO-definable sets of graphs as storage configurations and as storage transformations. We consider sequential automata with MSO graph storage and associate with each such automaton a string language (in the usual way) and a graph language; a graph is accepted by the automaton if it represents a correct sequence of storage configurations for a given input string. For each MSO graph storage type, we define an MSO logic which is a subset of the usual MSO logic on graphs. We prove a  B\"uchi-Elgot-Trakhtenbrot theorem, 
both for the string case and the graph case.
Moreover, we prove that (i)~each  MSO graph transduction can be used as storage transformation in an MSO graph storage type, (ii)~every automatic storage type is MSO-expressible, and (iii)~the pushdown operator on storage types preserves the property of MSO-expressibility. Thus, the iterated pushdown storage types are MSO-expressible.
\end{abstract}

\newpage
\tableofcontents
\newpage

\section{Introduction}
\label{sec:introduction}

Starting in the 60's of the previous century, a number of different types of  nondeterministic one-way string automata with additional storage were introduced in order to model different aspects of programming languages or natural languages. Examples of such storages are pushdowns \cite{cho62},  stacks \cite{gingrehar67}, checking-stacks \cite{gre69,eng79},  checking-stack pushdowns \cite{vle76},  nested stacks \cite{aho69}, iterated pushdowns \cite{gre70,mas76,eng86,damgoe86}, queues, and monoids or groups \cite{kam09,zet17}. 
Several general frameworks were considered in which the concept of storage has different names: machines \cite{sco67}, AFA-schemas \cite{gin75} (where AFA stands for abstract family of acceptors), data stores \cite{gol77,gol79}, and storage types  \cite{eng86,engvog86}. 

Intuitively, a storage type $S$ consists of a set $C$ of (storage) configurations, an initial configuration in $C$, a finite set $\Theta$ of instructions, and a meaning function~$m$. The meaning function assigns to each instruction a storage transformation, which is a binary relation on $C$. An automaton $\cA$ with storage type~$S$, for short: $S$-automaton, has a finite set of states with designated initial and final states, and a finite number of transitions of the form  $(q,\alpha,\theta,q')$ where $q,q'$ are states, $\alpha$ is an input symbol or the empty string, and $\theta$ is an instruction. During a computation on an input string, $\cA$ changes its state and reads input symbols consecutively (as for finite-state automata without storage); 
additionally, $\cA$~maintains a configuration in its storage, starting in the initial configuration of~$S$. If the current configuration of the storage is $c$ and $\cA$~executes a transition with instruction $\theta$, then $c$ is replaced by some configuration $c'$ such that $(c,c') \in m(\theta)$; if such a $c'$ does not exist, then $\cA$ cannot execute this transition. 
It is easy to see that pushdown automata, stack automata, nested-stack automata etc.~are particular $S$-automata (cf. \cite{eng86,engvog86} for examples). A~string language is \mbox{\emph{$S$-recognizable}} if there is an $S$-automaton that accepts this language. 
Since we only consider ``finitely encoded'' storage types (which means that $\Theta$ is finite),
there is one $S$-recognizable language of particular interest: the language $\cB(S)\subseteq \Theta^*$
that consists of all \emph{behaviours of~$S$}, i.e., all strings of instructions $\theta_1\cdots\theta_n$ for which there are configurations $c_1,\dots,c_{n+1}$ such that $c_1$ is the initial configuration and $(c_i,c_{i+1})\in m(\theta_i)$ for every $i\in\{1,\dots,n\}$. Intuitively, $\cB(S)$ represents the expressive power of $S$.

A major contribution to the theory of automata with storage is the following result \cite{gingre69,gingrehop69,gingre70,gin75}:  a class $\cL$ of string languages is a 
full principal AFL (abstract family of languages) if and only if there is a finitely encoded AFA-schema (i.e., storage type) $S$ such that $\cL$ is  the class  of all \mbox{$S$-recognizable} string languages.
In fact, $\cL$ is generated by the language $\cB(S)$. 
In \cite{eng86}, recursive \mbox{$S$-automata} and alternating $S$-automata  were investigated, and two characterizations of recursive \mbox{$S$-automata} were proved: (i)~in terms of sequential $\PDop(S)$-automata (where $\PDop$~is the pushdown operator on storage types \cite{gre70,eng86,engvog86,eng91c}) and (ii)~in terms of deterministic (sequential) $S$-automata. Based on the concept of weighted automata \cite{sch61,eil74,salsoi78,kuisal86,berreu88,sak09,drokuivog09}, recently also weighted $S$-automata have been investigated \cite{hervog15,hervog16,vogdroher16,fulhervog18,fulvog19,hervogdro19}.

A fundamental theorem for the class of recognizable (or regular) string languages is the B\"uchi-Elgot-Trakhtenbrot theorem \cite{buc60,buc62,elg61,tra61} (for short: BET-theorem). It states that a string language is recognizable by a finite-state automaton if and only if it is MSO-definable, i.e., definable by a closed formula of monadic second-order logic (MSO logic). 
This theorem has been generalized in several directions: (i)~for structures different from strings, such as, e.g., trees \cite{thawri68,don70}, nested words \cite{alumad09,bol08}, traces \cite{tho90,chogue93}, and pictures \cite{giaresseitho96}, and (ii)~for weighted automata \cite{drogas07,drogas09,gasmon18}. Moreover, (iii)~the BET-theorem was extended to classes of languages which go beyond recognizability by finite-state automata. In \cite{lauschthe94} context-free languages were characterized by an extension of MSO logic in which formulas have the form $\exists M. \varphi$, where $M$ is a matching (of the positions of the given string) and $\varphi$ is a formula of MSO logic (or even first-order logic). A similar result was obtained in \cite{fravou15} for realtime indexed languages.
Inspired by this third direction, in \cite{vogdroher16}, for each storage type $S$ an extended weighted MSO logic was introduced and a BET-theorem for weighted $S$-automata was proved; in that logic formulas have the form  $\exists B. \varphi$ where $B$ is a behaviour of~$S$ (of the same length as the input string) and $\varphi$ is a formula of weighted MSO logic.

The  BET-theorems in (iii) above can be captured by the following scheme. 
Let us consider a class of ``$X$-recognizable'' languages (where, e.g., $X$ is a storage type), and suppose that we have defined 
for every input alphabet $A$ a set of graphs $\cG[X,A]$ and a mapping $\pi\colon \cG[X,A]\to A^*$.
For every string $w\in A^*$, let $\cG[X,w]$ be the set of all graphs $g\in\cG[X,A]$
such that $\pi(g)=w$; intuitively, the graphs in $\cG[X,w]$ are ``extensions'' of the string $w$. 
In this situation, the BET-theorem says that
a language $L\subseteq A^*$ is $X$-recognizable if and only if there is a closed formula $\phi$ of MSO logic on graphs such that 
\[
L=\{w\in A^*\mid \exists g\in \cG[X,w]: \,g\models \phi\}
\] 
where $g\models \phi$ means as usual that $g$ satisfies $\phi$. Or in other words, 
$L$ is \mbox{$X$-recognizable} if and only if $L=\pi(\cG[X,A]\cap G)$ 
for some MSO-definable set of graphs $G$. 
In \cite{lauschthe94} (with $X$-recognizable = context-free), each string (viewed as a graph in the obvious way) is extended with edges between its positions that form a matching (and $\pi$ removes those edges). 
In \cite{vogdroher16} (with $X=S$), each string is extended with an additional labeling of its positions, which forms a behaviour of $S$. Whereas in \cite{lauschthe94} the class of graphs $\cG[X,A]$ is itself MSO-definable (because matchings can be defined by an MSO formula), that is not the case in \cite{vogdroher16}, 
because the language $\cB(S)$ is, in general, not regular. Thus, in \cite{vogdroher16}, an $S$-recognizable language $L$ is expressed by a combination of a formula of MSO logic and the non-recognizable language $\cB(S)$.
Our aim in this paper is to define storage types $S$ for which we can find a BET-theorem for $S$-recognizable languages that satisfies the above scheme, such that every set of graphs $\cG[S,A]$ is MSO-definable.
In that case the BET-theorem is equivalent to saying that $L$ is $S$-recognizable if and only if $L=\pi(G)$ for some MSO-definable subset $G$ of $\cG[S,A]$. 
As a final remark, we observe that according to the above scheme the BET-theorem for trees (see (i) above) 
also leads to 
a BET-theorem for the context-free languages (closely related to the one in~\cite{lauschthe94}). In this case $\cG[X,A]$ is the MSO-definable set of all binary trees~$t$ of which the yield is in $A^*$ (and the internal nodes are labeled by some fixed symbol), and $\pi(t)$ is the yield of $t$. Thus, each string $w$ is extended into trees with yield $w$. Since the context-free languages are the yields of the recognizable tree languages $G\subseteq \cG[X,A]$ (see~\cite[Chapter~III, Theorem~3.4]{gecste84}), they are indeed the yields of the MSO-definable tree languages $G\subseteq \cG[X,A]$. It should be noted that the trees in $\cG[X,A]$ can be viewed as the skeletons of derivation trees of a context-free grammar (in Chomsky normal form).  
Similarly, for a storage type $S$ we will define the set of graphs $\cG[S,A]$ such that its elements can be viewed as skeletons of the computations of $S$-automata.\footnote{For the storage type $S$ of pushdowns this will lead to yet another BET-theorem for the context-free languages.} Roughly speaking, such a skeleton is the sequence $c_1,\dots,c_{n+1}$ of configurations that witnesses a behaviour $\theta_1\cdots\theta_n$ of $S$. Thus, the configurations of $S$ have to be represented by graphs. Moreover, in order to be able to express in MSO logic the relationship between $c_i$ and $c_{i+1}$ caused by the instruction $\theta_i$, the storage transformation $m(\theta)$ of each instruction~$\theta$ also has to be represented by a set of graphs.

For pushdown-like storage types (as, e.g., the first six above-mentioned ones), the configurations and instructions are often explained and illustrated by means of pictures. For example, Figures \ref{fig:pushdown-informal}(a) and (b) show illustrations  of a pushdown configuration and of instances of a push- and a pop-instruction, respectively (cf. \cite[p.~344f]{engvog86} for an example concerning nested stacks over some storage type~$S$). Indeed, such pictures can be formalized as graphs (with pushdown cells as nodes and neighbourhood as edges), and hence, storage transformations can be understood as graph transductions.

\begin{figure}[t]
  \begin{center}
    \includegraphics[scale=0.3,trim={0 2.8cm 0 0.5cm},clip]{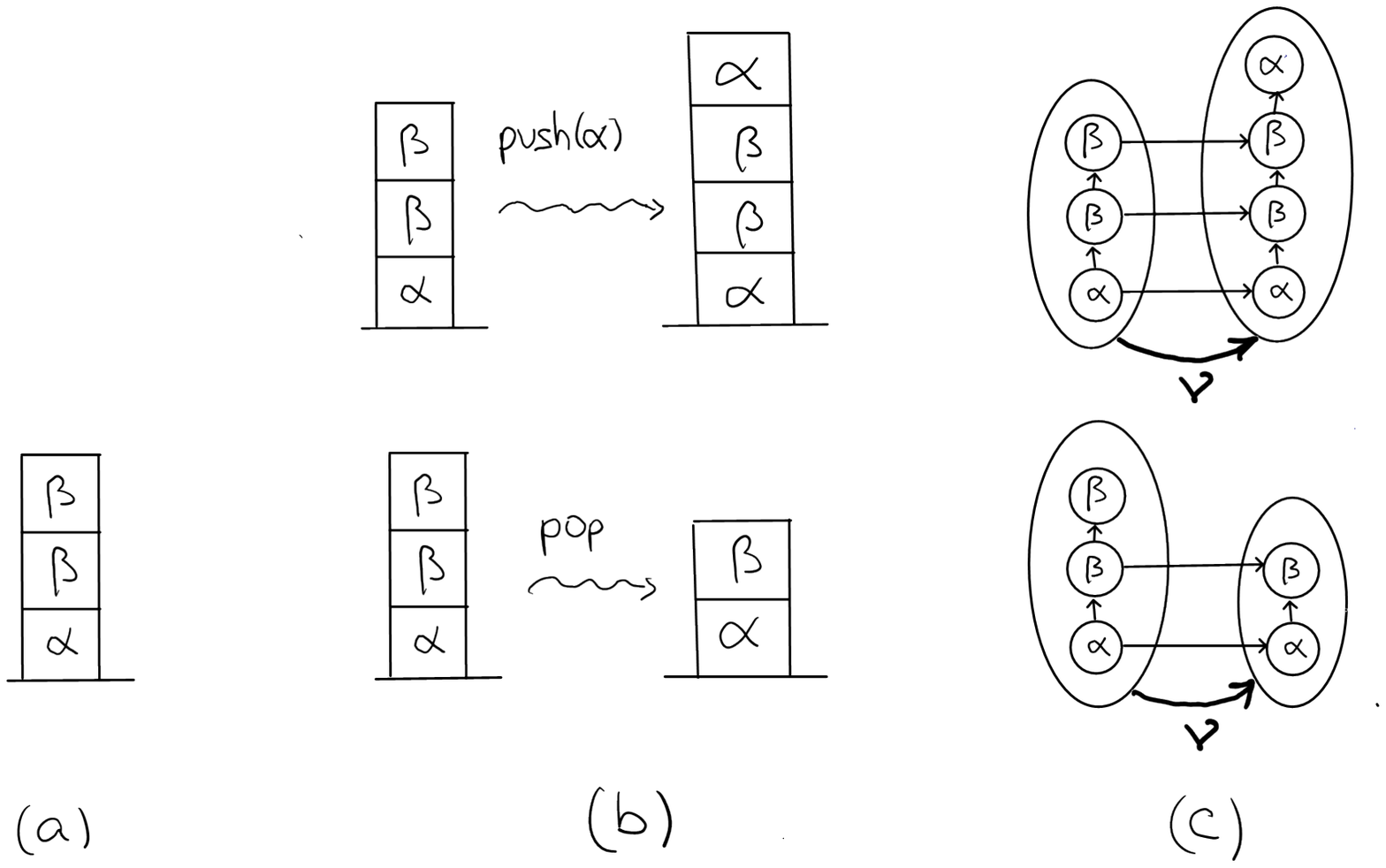}
    \end{center}
  \caption{\label{fig:pushdown-informal} (a) Illustration of a pushdown confguration, (b) illustration of instances of the instructions $\push(\alpha)$ and $\pop$, (c) two pair graphs corresponding to the instances of the instructions shown in (b). In (c), the two components $g_1$ and $g_2$ of a pair graph are surrounded by ovals, and the $\nu$-labeled edge between the ovals represents the $\nu$-labeled edges from each node of $g_1$ to each node of~$g_2$.}
  \end{figure}

  In this paper,  we define particular storage types for which the set of configurations is an MSO-definable set of graphs.  Moreover, each instruction $\theta$ is a closed MSO formula 
that defines a set of so-called \mbox{\emph{pair graphs}}. 
The storage transformation $m(\theta)$ is specified by the formula $\theta$ as follows. Intuitively, 
a \mbox{pair graph} is a graph that is partitioned into two component graphs $g_1$ and $g_2$, which are two configurations, one before the execution of the instruction and one after execution; to express this, there are $\nu$-labeled edges from each node of $g_1$ to each node of $g_2$ (where $\nu$ stands for `next');
moreover, there can be additional edges between $g_1$ and $g_2$ (intermediate edges) which model the similarity of the two configurations (cf. Figure~\ref{fig:pushdown-informal}(c) for examples of pair graphs which represent instances of the instructions $\push(\alpha)$ and $\pop$, respectively; in this case, the intermediate edges between $g_1$ and $g_2$ show which nodes of $g_2$ can be viewed as copies of nodes of $g_1$, unchanged by the instruction). By dropping the $\nu$-labeled edges and the intermediate edges we obtain the ordered pair $(g_1,g_2)$ which is an element of the graph transduction specified by the MSO formula $\theta$, i.e., the storage transformation $m(\theta)$. We call such a storage type an \mbox{\emph{MSO graph storage type}}.
We  say that a storage type is  \emph{MSO-expressible} if it is isomorphic to some MSO graph storage type.

We study $S$-automata $\cA$ where $S$ is an MSO graph storage type. To simplify the discussion in this Introduction, we will assume that $\cA$ has no $\epsilon$-transitions, i.e., $\alpha\neq\epsilon$ in every transition $(q,\alpha,\theta,q')$.
We also assume that the graphs defined by the MSO formulas of $S$ do not have $A$-labeled edges.

The $S$-automaton $\cA$ accepts a string language $\LL(\cA)$ over some input alphabet~$A$ and a graph language $\GLL(\cA)$.
The string language $\LL(\cA)$ is defined in the usual way as for automata with arbitrary storage, i.e., the configurations are kept in a private memory. But we can also view $\cA$ as graph acceptor. Then the sequence of configurations, assumed by the string acceptor $\cA$ while accepting a string $w\in A^*$, is made public and, together with the string $w$, forms the input for the graph acceptor $\cA$. So to speak, the graph acceptor $\cA$ accepts the storage protocols of the string acceptor $\cA$. In order to describe such storage protocols, we define \emph{string-like graphs}. Intuitively, each string-like graph $g$ is a graph that consists of a sequence of component graphs; their order is provided by $A$-labeled edges (similarly to the $\nu$-edges in pair graphs) and the sequence of labels of these edges is called the trace of $g$ (which corresponds to the input string $w$ above). Each component is a configuration of the MSO graph storage type $S$, and the first component is the initial configuration of $S$. Moreover, between consecutive components intermediate edges may occur that model the similarity of the respective configurations (as in pair graphs); cf. Figure~\ref{fig:comput-intro} for an example.
We denote the set of all such string-like graphs by $\cG[S,A]$. It should be intuitively clear that $\cG[S,A]$ is MSO-definable. Note that every string-like graph $g$ with trace $w\in A^*$ can be viewed as an ``extension'' of the string $w$; thus, `trace' is the mapping $\pi: \cG[S,A]\to A^*$ in the scheme of BET-theorems sketched above.
  \begin{figure}[t]
  \begin{center}
    \includegraphics[scale=0.37,trim={0 6.2cm 0 4cm},clip]{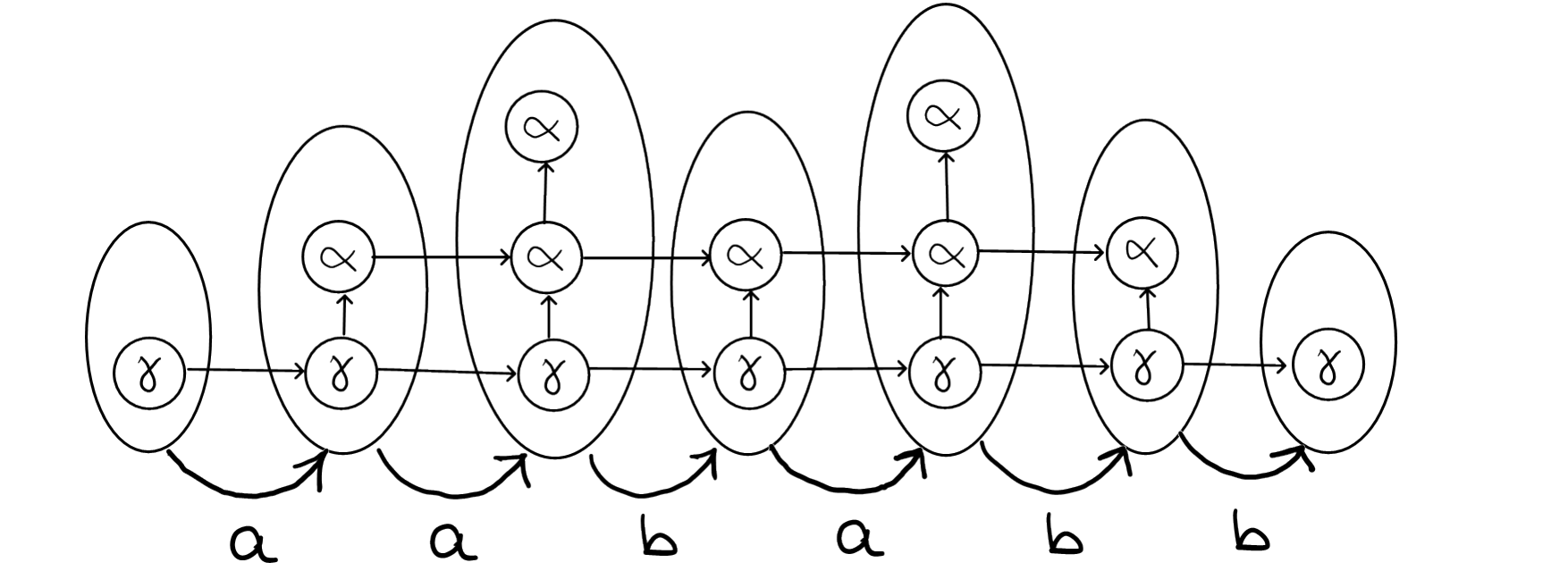}
    \end{center}
\caption{\label{fig:comput-intro} A string-like graph $g$ with seven components (surrounded by ovals). Each component represents a pushdown configuration (formalized as graph). Starting from the initial configuration $\gamma$, the sequence of components results from the execution of the instructions $\push(\alpha)$, $\push(\alpha)$, $\pop$, $\push(\alpha)$, $\pop$, and $\pop$. The trace of $g$ is $aababb$, where $a,b \in A$ are input symbols. An \mbox{$a$-labeled} edge from one oval to another represents $a$-labeled edges from each node of the one component to each node of the other component, and similarly for $b$-labeled edges.} 
\end{figure}
The graph acceptor $\cA$ accepts a string-like graph $g\in\cG[S,A]$ with $n+1$ components ($n\geq 0$), 
if there is a sequence
\[
(q_1,\alpha_1,\theta_1,q_2) \cdots (q_n,\alpha_n,\theta_n,q_{n+1})
\]
of transitions of $\cA$ such that (i) the state sequence $q_1\cdots q_{n+1}$ obeys the usual conditions, 
(ii)~$\alpha_1\cdots\alpha_n$ is the trace of $g$, and (iii) for each $i \in \{1,\ldots,n\}$, the restriction of $g$ to its $i$th and $(i+1)$st components (including the intermediate edges)  is a pair graph that satisfies the formula $\theta_i$ (after having replaced each $A$-label by~$\nu$). In view of (iii) the sequence $\theta_1\cdots\theta_n$ is a behaviour of $S$ (i.e., an element of $\cB(S)\subseteq \Theta^*$), which we will call a \emph{behaviour of $S$ on $g$}. The graph language $\GLL(\cA)$ accepted by $\cA$ is the set of all string-like graphs that are accepted by $\cA$.
A graph language $L\subseteq \cG[S,A]$ is \emph{$S$-recognizable} if there exists an $S$-automaton $\cA$ such that $L=\GLL(\cA)$.

Our first two main results are BET-theorems, one for sets of  string-like graphs and one for string languages, accepted by $S$-automata over the input alphabet $A$. 
Unfortunately we cannot exactly follow the scheme of BET-theorems sketched above. 
Instead of using arbitrary MSO formulas on the graphs of $\cG[S,A]$, as in that scheme, 
we have to restrict ourselves to a specific subset of that logic, tailored to the storage type $S$. 

Thus, we introduce the special logic $\MSOL(S,A)$ which can be viewed as a subset of the usual MSO logic for graphs.
Each formula $\phi$ of $\MSOL(S,A)$ has two levels: an outer level that refers to the fact that 
a string-like graph $g$ is a sequence of graph components connected by $A$-labeled edges (string aspect of~$g$), and an inner level that deals with the consecutive graph components of $g$ as configurations of $S$ (storage behaviour aspect of $g$). To express the storage behaviour aspect,
the inner level of $\phi$ consists of subformulas of $\phi$ of the form $\nnext(\theta,x,y)$ where $\theta \in \Theta$. (Recall that $\Theta$ is a set of closed MSO formulas.) The formula $\nnext(\theta,x,y)$ holds for $g$ if the nodes $x$ and $y$ are in the $i$th and $(i+1)$st component of $g$ (respectively) for some $i$, and the restriction of $g$ to these components (including the intermediate edges, and with the $A$-labels replaced by~$\nu$) satisfies the formula~$\theta$.\footnote{Formulas of the form $\nnext(\theta,x,y)$ play a similar role as formulas of the form $B(x)=(p,f)$ in the logic presented in~\cite{vogdroher16}.}
The outer level of $\phi$, which is the remainder of~$\phi$, is built up as usual (with negation, disjunction, and first-order and second-order existential quantification) from the above subformulas $\nnext(\theta,x,y)$ and the following atomic subformulas. To express the string aspect, there is no need for atomic formulas that can test the label of a node, but there are atomic formulas $\edge_\alpha(x,y)$ that can test whether there is an edge from $x$ to $y$ with label $\alpha$, for $\alpha\in A$. Moreover, the atomic formula $x \in X$ is replaced by the atomic formula $x \eeuro X$, which holds for~$g$ if $x \in X$ or there is a node $y\in X$ in the same component of $g$ as $x$. It should be intuitively clear that the logic $\MSOL(S,A)$ can be viewed as a subset of the usual MSO logic for graphs  (cf. Observation~\ref{obs:msolsubset}). 
A set of string-like graphs $L\subseteq \cG[S,A]$ is \emph{$\MSOL(S,A)$-definable} if 
there exists a closed formula $\phi \in \MSOL(S,A)$ such that 
\[
L = \{g\in\cG[S,A] \mid g \models \beh \wedge \phi\} 
\]
where the formula $\beh \in \MSOL(S,A)$ guarantees the existence of an $S$-behaviour on $g$.
Similarly, a string language $L \subseteq A^*$ is \emph{$\MSOL(S,A)$-definable} 
if there exists a closed formula $\phi \in \MSOL(S,A)$ such that 
\[
L = \{w \in A^* \mid \exists g\in \cG[S,w]: g  \models \beh \wedge \phi\}
\]
where $\cG[S,w]$ is the set of all $g\in\cG[S,A]$ that have trace $w$.
Then our first two main results state, for every MSO graph storage type $S$ and alphabet $A$, that
\begin{compactitem}
\item for every graph language $L \subseteq \cG[S,A]$, $L$ is $S$-recognizable if and only if
it is $\MSOL(S,A)$-definable (cf. Theorem~\ref{thm:main}), and
\item for every string language $L \subseteq A^*$, $L$ is $S$-recognizable if and only if it is $\MSOL(S,A)$-definable (cf. Theorem~\ref{thm:main-string}).
\end{compactitem}

The remaining three main results concern the question: which storage types are MSO-expressible?
We call a binary relation $R$ on graphs \emph{MSO-expressible} if there is a closed formula $\theta$ of MSO logic for graphs such that $\theta$ defines a set $\LL(\theta)$ of pair graphs and, roughly speaking, $R$ is obtained from $\LL(\theta)$ by dropping all the $\nu$-labeled edges and the intermediate edges. We prove that
\begin{compactitem}
\item every MSO graph transduction is MSO-expressible (cf. Theorem \ref{thm:MSO-graph-transduction-MSO-expressible}) 
\end{compactitem}
where an MSO graph transduction is induced by a (nondeterministic) MSO graph transducer \cite{bloeng00,coueng12}. Thus, if the storage transformations of a  storage type~$S$ are MSO graph transductions, then $S$ is MSO-expressible, i.e., isomorphic to an MSO graph storage type.
We also prove that
\begin{compactitem}
\item every automatic storage type is MSO-expressible (cf. Theorem \ref{thm:automatic}) 
\end{compactitem}
where a storage type $S$ is called automatic if the set $C$, together with the binary relations 
$m(\theta)$ for every $\theta\in\Theta$, is an automatic structure~\cite{khomin,rub08}.

Finally, we consider the above-mentioned pushdown operator $\PDop$ on storage types and prove that
\begin{compactitem}
\item  for every storage type $S$, if $S$ is MSO-expressible, then so is $\PDop(S)$ (cf.  Theorem \ref{thm:P}).
\end{compactitem}
Consequently, the $n$-iterated pushdown storage $\PDop^n$ is MSO-expressible (cf. Corollary \ref{cor:Pn}). We denote the class of all string languages that are accepted by $\PDop^n$-automata by $\PDop^n\mathord{\n}\mathrm{REC}$. 
The family $(\PDop^n\mathord{\n}\mathrm{REC} \mid n \in \mathbb{N})$ was investigated intensively \cite{wan74,engsch77,engsch78,dam82,damgoe86,eng91c}.  It is an infinite hierarchy of classes of string languages which starts with the classes of regular languages ($n=0)$, context-free languages ($n=1$), and indexed languages ($n=2$).

\section{Preliminaries}
\label{sec:preliminaries}

\subsection{Mathematical Notions}

We denote the set $\{0,1,2,\ldots\}$ of natural numbers by $\mathbb{N}$. 
For each $n \in \mathbb{N}$ we denote the set $\{i \in \mathbb{N} \mid 1 \le i \le n\}$ by $[n]$. Thus, in particular, $[0] = \emptyset$. For sets $A$ and $B$, we denote a total function (or: mapping) $f$ from $A$ to $B$ by $f: A \to B$. For a nonempty set $A$, a partition of $A$ is a set $\{A_1,\dots,A_n\}$ of mutually disjoint nonempty subsets of $A$ such that $\bigcup_{i\in[n]}A_i=A$. An \emph{ordered partition} of~$A$ is a sequence $(A_1,\dots,A_n)$ of distinct sets such that $\{A_1,\dots,A_n\}$ is a partition of~$A$. 

For a set $A$, we denote by $A^*$ the set of all sequences $(a_1,\dots,a_n)$ with $n\in\nat$ and $a_i\in A$ for every $i\in[n]$. The empty sequence (with $n=0$) is denoted by~$\epsilon$, and $A^+$ denotes the set of nonempty sequences. A sequence $(a_1,\dots,a_n)$ is also called a string over $A$, and it is then written as $a_1\cdots a_n$. The length of $w=a_1\cdots a_n$ is $n$, also denoted by $|w|$. 
An \emph{alphabet} is a finite and nonempty set. For an alphabet $A$,
a subset of $A^*$ is called a language over~$A$, or (when necessary) a string language over $A$. 

\begin{quote}\emph{In the rest of the paper, we let $\Sigma$ and $\Gamma$ denote arbitrary alphabets if not specified otherwise. } \end{quote}

\subsection{Graphs and Monadic Second-Order Logic} 
\label{sec:graphs-MSO}

We use $\Sigma$ and $\Gamma$ as alphabets of node labels and edge labels, respectively. A~\emph{graph over $(\Sigma,\Gamma)$} is a tuple $g = (V,E,\nlab)$
where $V$ is a nonempty  finite set (of \emph{nodes}), $E \subseteq V \times \Gamma \times V$ (set of \emph{edges}) such that $u\neq v$ for every $(u,\gamma,v)\in E$, and $\nlab: V \to \Sigma$ (\emph{node-labeling function}). Note that we only consider graphs that are nonempty and do not have loops; 
moreover, multiple edges must have distinct labels. 
For a graph $g$ we denote its sets of nodes and edges by $V_g$ and $E_g$, respectively, and its node-labeling function by $\nlab_g$.  For $\Delta\subseteq \Gamma$, an edge $(u,\gamma,v)$ is called a $\Delta$-edge if $\gamma\in\Delta$; for $\gamma\in\Gamma$ we write $\gamma$-edge for $\{\gamma\}$-edge. The set of all graphs over $(\Sigma,\Gamma)$ is denoted by $\cG_{\Sigma,\Gamma}$. A subset of $\cG_{\Sigma,\Gamma}$ is also called a graph language over $(\Sigma,\Gamma)$.

We will view isomorphic graphs to be the same. Thus, we consider abstract graphs. As usual, 
we use a concrete graph to define the corresponding abstract graph. 

Let $g=(V,E,\nlab)$ be a graph over $(\Sigma,\Gamma)$.
For a node $u\in V$ and an edge label $\gamma\in\Gamma$ we define its incoming and outgoing \emph{neighbours} (with respect to $\gamma$-edges)
by $\inn_\gamma(u)=\{v\in V \mid (v,\gamma,u)\in E\}$ and 
$\out_\gamma(u)=\{v\in V \mid (u,\gamma,v)\in E\}$, respectively.
Now let $\Delta\subseteq \Gamma$. We define $u,u'\in V$ to be \emph{$\Delta$-equivalent}, 
denoted by $u\equiv_\Delta u'$, if 
$\inn_\delta(u)=\inn_\delta(u')$ and $\out_\delta(u)=\out_\delta(u')$ for every $\delta\in\Delta$.
Since $g$ has no loops, there are no $\Delta$-edges between $\Delta$-equivalent nodes. 
It is also easy to see that, for every $\delta\in\Delta$, the equivalence relation $\equiv_\Delta$ 
is a congruence with respect to the $\delta$-edges, i.e., for every $u,u',v,v'\in V$, 
if $(u,\delta,v)\in E$, $u \equiv_\Delta u'$, and $v \equiv_\Delta v'$, then $(u',\delta,v')\in E$.

Let $g=(V,E,\nlab)$ be a graph over $(\Sigma,\Gamma)$. For a nonempty set $V'\subseteq V$, 
the \emph{subgraph of $g$ induced by~$V'$} is the graph $g[V']=(V',E',\nlab')$ where 
$E'=\{(u,\gamma,v)\in E\mid u,v\in V'\}$ and $\nlab'$ is the restriction of $\nlab$ to~$V'$.
For every $\Delta\subseteq\Gamma$ and $\gamma\in\Gamma$, 
we denote by $\lambda_{\Delta,\gamma}(g)$ the graph that is obtained from $g$ 
by changing every edge label in $\Delta$ into $\gamma$; 
thus, e.g., if $(u,\gamma_1,v)$ and $(u,\gamma_2,v)$ are distinct edges of $g$ 
and $\gamma_1, \gamma_2 \in \Delta$, then these two edges collapse to one edge $(u,\gamma,v)$.

Let $w=\gamma_1\cdots \gamma_n$ be a string over $\Gamma$, 
for some $n\in\nat$ and $\gamma_i\in \Gamma$ for each $i\in[n]$.
The graph $g=(V,E,\nlab)$ is a \emph{string graph for $w$} if
$V =[n+1]$ and  $E = \{(i,\gamma_i,i+1) \mid i \in [n]\}$.
Thus, string graphs for $w$ only differ in their node-labeling functions.
A graph is a \emph{string graph} if it is a string graph for some $w \in \Gamma^*$.

We use monadic second-order logic to describe properties of graphs. 
This logic has \emph{node variables} (first-order variables), like $x,x_1,x_2,\ldots, y,z$ and \emph{node-set variables} (second-order variables), like $X,X_1,X_2, \ldots,Y,Z$. A \emph{variable} is a node variable or a node-set variable. For a given graph $g$ over $(\Sigma,\Gamma)$, each node variable ranges over $V_g$, and each node-set variable ranges over the set of subsets of $V_g$.  

The \emph{set of MSO-logic formulas over $\Sigma$ and $\Gamma$}, denoted by  $\MSOL(\Sigma,\Gamma)$, is the smallest set $M$ of expressions such that
\begin{compactenum}
\item[$(1)$] for every $\sigma\in\Sigma$ and $\gamma \in \Gamma$, the set $M$ contains the expressions $\lab_\sigma(x)$,  $\edge_\gamma(x,y)$, and $(x\in X)$, which are called atomic formulas, and
\item[$(2)$] if $\varphi,\psi \in M$, then $M$ contains the expressions $(\neg \varphi)$, $(\varphi \vee \psi)$, $(\exists x. \varphi)$, and $(\exists X. \varphi)$.
\end{compactenum}
We will drop parentheses around subformulas if they could be reintroduced without ambiguity. 
We will use macros like $x=y$, $X\subseteq Y$,  $\varphi \rightarrow \psi$,  $\varphi \leftrightarrow \psi$, $\varphi \wedge \psi$, $\forall x. \varphi$, $\forall X.\varphi$, $\true$, and $\false$,  with their obvious definitions. We abbreviate $\forall x. \forall y. \varphi$ by $\forall x,y. \varphi$ and similarly for more than two variables and for  existential quantification. 
Moreover, for every $\Delta\subseteq \Gamma$, we use the macros
\begin{align*}
\edge_\Delta(x,y) &= \bigvee\limits_{\gamma\in\Delta}\edge_\gamma(x,y), \\
\closed_\Delta(X) &= \forall x,y. ((\edge_\Delta(x,y) \wedge x \in X) \rightarrow y \in X), \text{ and}\\
\ppath_\Delta(x,y) &= \forall X. ((\closed_\Delta(X) \wedge x \in X) \rightarrow y \in X)\enspace,
\end{align*}
where the formula $\ppath_\Delta(x,y)$ means that there is a directed path from $x$ to~$y$ consisting of $\Delta$-edges.

In the usual way, we define the set $\Free(\varphi)$ of free variables of a formula $\varphi$.
If, say, $\{x,Y,z\}\subseteq \Free(\varphi)$, then we also write $\phi$ as $\phi(x,Y,z)$. 
For a set $\cV$ of variables, we denote the set 
$\{\varphi \in \MSOL(\Sigma,\Gamma)\mid \Free(\varphi)\subseteq\cV\}$ 
by $\MSOL(\Sigma,\Gamma,\cV)$. 
Each $\varphi \in \MSOL(\Sigma,\Gamma,\emptyset)$ is called \emph{closed}.

Let $g$ be a graph over $(\Sigma,\Gamma)$. Moreover, let $\cV$ be a set of variables and let
$\varphi \in \MSOL(\Sigma,\Gamma,\cV)$.
A~\emph{$\cV$-valuation on $g$} is a mapping $\rho$ that assigns to each node variable of $\cV$ an element of $V_g$ and to each node-set variable of $\cV$ a subset of $V_g$. 
In the usual way, we define the \emph{models relationship} $(g,\rho) \models \varphi$ to mean that~$g$, with the values of its free variables provided by $\rho$, \emph{satisfies}~$\phi$. 
Note that $(g,\rho) \models \lab_\sigma(x)$ if and only if $\ell_g(\rho(x))=\sigma$, and 
$(g,\rho) \models \edge_\gamma(x,y)$ if and only if $(\rho(x),\gamma,\rho(y))\in E_g$.
If, say, $\{x,Y,z\} \subseteq \cV$, 
then we also write $(g,\rho',\rho(x),\rho(Y),\rho(z)) \models \varphi(x,Y,z)$
instead of $(g,\rho) \models \varphi$, where $\rho'$ is the restriction of~$\rho$ to $\cV\setminus \{x,Y,z\}$.
If $\phi$ is closed, then we write $g \models \varphi$ instead of $(g,\emptyset) \models \varphi$, and 
we define $\LL(\phi)=\{g\in \cG_{\Sigma,\Gamma}\mid g \models \varphi\}$.
A graph language  $L \subseteq \cG_{\Sigma,\Gamma}$ is \emph{$\MSOL(\Sigma,\Gamma)$-definable}
(or just MSO-definable, when $\Sigma$ and $\Gamma$ are clear from the context) 
if there is a closed formula $\varphi \in \MSOL(\Sigma,\Gamma)$ such that $L = \LL(\varphi)$.

A set of closed formulas $\Phi\subseteq \MSOL(\Sigma,\Gamma)$ is \emph{exclusive} if its elements are mutually exclusive, i.e., $\LL(\varphi)\cap \LL(\psi)=\emptyset$ for all distinct $\varphi,\psi\in\Phi$.

For a formula $\phi\in\MSOL(\Sigma,\Gamma,\cV)$ and a node-set variable $Y\notin\cV$, the \emph{relativization of $\phi$ to $Y$} is the formula $\phi|_Y\in \MSOL(\Sigma,\Gamma,\cV\cup\{Y\})$ that is obtained from $\phi$ by restricting all quantifications of $\phi$ to $Y$. Formally, $\phi|_Y = \phi$ for every atomic formula, and
\[
\begin{array}{ll}
(\neg\phi)|_Y = \neg(\phi|_Y), & 
(\exists x.\phi)|_Y = \exists x.(x\in Y \wedge \phi|_Y),\\[1mm]
(\phi\vee \psi)|_Y=\phi|_Y\vee\psi|_Y, &
(\exists X.\phi)|_Y = \exists X.(X\subseteq Y \wedge \phi|_Y). 
\end{array}
\]
Let $g=(V,E,\nlab)$ be a graph over $(\Sigma,\Gamma)$, let $V'$ be a nonempty subset of $V$, and 
let $\rho$ be a $\cV$-valuation on the induced subgraph $g[V']$. 
Then, $(g[V'],\rho)\models \phi$ if and only if $(g,\rho,V')\models \psi(Y)$,
where $\psi=\phi|_Y$. 

\begin{example}\rm \label{ex:strings}
We show that the set of string graphs over $(\Sigma,\Gamma)$ is MSO-definable. For this, we define a closed MSO-logic formula $\sstring_{\Gamma}$ in $\MSOL(\Sigma,\Gamma)$  such that 
for each $g \in \cG_{\Sigma,\Gamma}$ we have
\[
g \models \sstring_{\Gamma} \ \text{ if and only if } \ g \text{ is a string graph over $(\Sigma,\Gamma)$}.
\]
Each string graph has a unique first node and a unique last node:
\begin{align*}
\first(x) &= (\neg \exists y. \edge_\Gamma(y,x)) \wedge \forall z. ((\neg \exists y. \edge_\Gamma(y,z)) \rightarrow z=x)\\
\last(x) &= (\neg \exists y. \edge_\Gamma(x,y)) \wedge \forall z. ((\neg \exists y. \edge_\Gamma(z,y)) \rightarrow z=x)\enspace.
\end{align*}
Moreover, each node has at most one successor and at most one predecessor: 
\begin{align*}
\ssucc_{\le 1}(x) &= \forall y,z. (\edge_\Gamma(x,y) \wedge \edge_\Gamma(x,z) \rightarrow y=z)\\
\ppred_{\le 1}(x) &= \forall y,z. (\edge_\Gamma(y,x) \wedge \edge_\Gamma(z,x) \rightarrow y=z) \enspace.
\end{align*}
In a string graph, there is at most one edge between two nodes:
\begin{align*}
\exclusive(x,y) &= \bigwedge\limits_{\gamma\in\Gamma} 
(\edge_\gamma(x,y) \rightarrow \neg \bigvee\limits_{\delta\in \Gamma\setminus\{\gamma\}} \edge_\delta(x,y)) \enspace.
\end{align*}
Since a string graph is connected, we eventually let 
\begin{align*}
\sstring_{\Gamma} = &\ \exists x. \first(x) \wedge \exists x. \last(x) \\
& \wedge \forall x. (\ssucc_{\le 1}(x) \wedge \ppred_{\le 1}(x)) \\
& \wedge \forall x,y. \,\exclusive(x,y) \\
& \wedge \forall x,y,z. (\first(x) \wedge \last(z) \rightarrow \ppath_\Gamma(x,y) \wedge \ppath_\Gamma(y,z))\enspace. 
\end{align*}
\qed
\end{example}

\subsection{Regular Languages}\label{sec:reg-lang}

Let $A$ be an alphabet. 
A (nondeterministic) \emph{finite-state automaton} over $A$ is a tuple $\cA=(Q,Q_\init,Q_\rf,T)$
where $Q$ is a finite set of states, $Q_\init\subseteq Q$ 
is the set of initial states, 
$Q_\rf\subseteq Q$ is the set of final states, and $T$ is a finite set of transitions. 
Each transition is of the form $(q,a,q')$ with $q,q'\in Q$ and $a\in A$. 
Let $w=a_1\cdots a_n$ be a string over $A$, with $n \in \nat$ and $a_i \in A$ for each $i \in [n]$. 
The string $w$ is \emph{accepted by} $\cA$ 
if there exist $q_1,\dots,q_{n+1}\in Q$
such that $q_1\in Q_\init$, $q_{n+1}\in Q_\rf$, and
$(q_i,a_i,q_{i+1})\in T$ for every $i\in[n]$.
The language $\LL(\cA)$ accepted by $\cA$ consists of all strings over $A$ that 
are accepted by $\cA$. 
A language $L\subseteq A^*$ is \emph{regular} if $L=\LL(\cA)$ for some finite-state automaton $\cA$ over $A$. 

Instead of defining an MSO logic for strings, we follow the equivalent approach of representing every string 
by a string graph (as defined in Section~\ref{sec:graphs-MSO}) and using the MSO logic for graphs.  
For $w\in A^*$ we define $\edgr(w)$ to be the unique string graph for $w$ in $\cG_{\{*\},A}$.
Each node of $\edgr(w)$ is labeled by $*$, and the edges of $\edgr(w)$ are labeled by the symbols 
that occur in $w$. Obviously, $\edgr(w)$ is a unique graph representation of the string $w$, 
cf.~\cite[p.~232]{enghoo01}. So, as a logic for strings over $A$ we will use $\MSOL(\{*\},A)$,
and we view a language $L \subseteq A^*$ to be MSO-definable if the graph language 
$\edgr(L)=\{\edgr(w)\mid \text{\mbox{$w\in L$}}\}$ is $\MSOL(\{*\},A)$-definable. 

The classical BET-theorem for strings can now be formulated as follows, 
see, e.g., \cite[Proposition~9]{enghoo01}.

\begin{proposition}\rm \label{pro:MSO-string-graph} 
A language $L \subseteq A^*$ is regular if and only if 
$\edgr(L)$ is $\MSOL(\{*\},A)$-definable. 
\end{proposition}

Intuitively, the nodes of $\edgr(w)$ can be viewed as the ``positions'' of the string $w=a_1\cdots a_n$,
where there is a position between each pair $(a_i,a_{i+1})$ of symbols of $w$, plus one position at the beginning of $w$ and one position at its end.
A finite-state automaton visits these $n+1$ positions from left to right. 
The atomic formula $\edge_a(x,y)$ of $\MSOL(\{*\},A)$ means that the symbol $a$ is 
between positions $x$ and~$y$ (and the atomic formula $\lab_*(x)$ is always true).
There is another unique graph representation of strings that corresponds more closely to the classical proof of the BET-theorem for strings: $\text{nd-gr}(w)$ is the string graph $(V,E,\nlab)\in\cG_{A,\{*\}}$ 
with $V=[n]$, $E=\{(i,*,i+1)\mid i\in[n-1]\}$, and $\nlab(i)=a_i$ for every $i\in[n]$. In this representation 
the string $w$ has a ``position'' at each symbol $a_i$ (so the nodes of $\text{nd-gr}(w)$ are again the 
positions of $w$), a finite-state automaton visits these $n$ positions from left to right (and falls off the 
end of~$w$ in a final state), and the atomic formula $\lab_a(x)$ of $\MSOL(A,\{*\})$ means that 
the symbol $a$ is at position $x$ (and the atomic formula $\edge_*(x,y)$ is true whenever $x$ and $y$ are neighbouring positions). 
Now the BET-theorem says that $L$ is regular if and only if 
$\text{nd-gr}(L)$ is $\MSOL(A,\{*\})$-definable. It is shown in~\cite[Proposition~9]{enghoo01} 
that these two variants of the BET-theorem for strings are equivalent, because the transformations 
from $\edgr(w)$ to $\text{nd-gr}(w)$ and back, are simple MSO graph transductions 
(in the sense of~\cite[Chapter~7]{coueng12},  cf. Section~\ref{sec:msogratra}).

\subsection{Storage Types and $S$-Automata}
\label{sec:storagetypes}

In the literature, 
automata that make use of an auxiliary storage
can test the current storage configuration by means of a predicate, 
and transform it by means of a deterministic instruction. 
General frameworks to define automata with a particular type of storage 
were considered, e.g., in \cite{gin75,sco67,eng86,engvog86}.
We will consider nondeterministic automata only, and hence predicates are not needed:  
they can be viewed as special instructions (see below). 
For more generality, we also allow our instructions to be nondeterministic (as in \cite{gol77,gol79}).
On the other hand, we only consider finitely encoded storage types \cite{gin75}, i.e., storage types that have only finitely many instructions. For pushdown-like storage types it means that the pushdown alphabet must be fixed
(which, as is well known, is not a restriction).

A \emph{storage type} is a tuple $S=(C,c_\init,\Theta,m)$ such that 
$C$ is a set (of \emph{storage configurations}), 
$c_\init\in C$ (the \emph{initial storage configuration}),
$\Theta$ is a finite set (of \emph{instructions}),
and $m$ is the \emph{meaning function} that associates a binary relation 
$m(\theta)\subseteq C\times C$ with every $\theta\in \Theta$. 

For every automaton $\cA$ with storage type $S$, the storage configuration at the start 
of $\cA$'s computations should be $c_\init$.
Every instruction $\theta\in\Theta$ executes the \emph{storage transformation} $m(\theta)$; 
if $(c,c')\in m(\theta)$, then, intuitively, $c$ and $c'$ are the storage configurations before and after 
execution of the instruction $\theta$, respectively.
Note that a test on the storage configuration, i.e., a Boolean function $\tau: C\to \{0,1\}$, can be modeled (as usual) by two ``partial identity'' instructions $\theta_0$ and $\theta_1$ such that $m(\theta_i)=\{(c,c)\mid \tau(c)=i\}$. 

Two storage types $S=(C,c_\init,\Theta,m)$ and $S_*=(C_*,(c_\init)_*,\Theta_*,m_*)$ 
are \emph{isomorphic} if there are bijections between $C$ and $C_*$ 
and between $\Theta$ and $\Theta_*$, such that
$m_*(\theta_*)=\{(c_*,c'_*)\mid (c,c')\in m(\theta)\}$ for every $\theta\in\Theta$, 
where $x_*$ denotes the bijective image of $x$ 
(and thus, in particular, $(c_\init)_*$ is the bijective image of~$c_\init$).

We now turn to the automata that use the storage type $S$.
Let $A$ be an alphabet (of input symbols). For technical reasons we will use a special symbol~$e$ 
(not in $A$) to represent the empty string $\epsilon$. For simplicity, we will denote the set 
$A\cup\{e\}$ by $\Ae$.
Moreover, we denote by $\mu_e$ the string homomorphism from $\Ae$ to $A$
that erases $e$, i.e., $\mu_e(e)=\varepsilon$ and $\mu_e(a)=a$ for every $a\in A$. 

For a storage type $S=(C,c_\init,\Theta,m)$ and an alphabet $A$, an \emph{$S$-automaton over $A$} is a tuple 
$\cA=(Q,Q_\init,Q_\rf,T)$ where $Q$ is a finite set of states,
$Q_\init\subseteq Q$ is the set of initial states, $Q_\rf\subseteq Q$ is the set of final states,
and $T$ is a finite set of transitions. Each transition is of the form $(q,\alpha,\theta,q')$ with $q,q'\in Q$,
$\alpha\in \Ae$, and $\theta\in\Theta$. 

A transition $(q,\alpha,\theta,q')$ will be called an $\alpha$-transition. Intuitively, for $a\in A$, 
an $a$-transition consumes the input symbol $a$, whereas an $e$-transition does not consume input
(and is usually called an $\epsilon$-transition).

An \emph{instantaneous description} of $\cA$ is a triple $(q,w,c)$ such that 
$q\in Q$, \mbox{$w\in A^*$}, and $c\in C$. 
It is \emph{initial} if $q\in Q_\init$ and $c=c_\init$, 
and it is \emph{final} if $q\in Q_\rf$.
For every transition $\tau=(q,\alpha,\theta,q')$ in $T$ we define the binary relation $\vdash^\tau$ 
on the set of instantaneous descriptions: for all $w\in A^*$ and $c,c'\in C$, 
we let $(q,\mu_e(\alpha)w,c)\vdash^\tau (q',w,c')$ if $(c,c')\in m(\theta)$. 
The \emph{computation step relation} of $\cA$ is the binary relation $\vdash\,=\bigcup_{\tau\in T} \vdash^\tau$.
A string $w\in A^*$ is \emph{accepted by} $\cA$ if there exist 
an initial instantaneous description $(q_\init,w,c_\init)$ 
and a final instantaneous description $(q_\rf,\varepsilon,c)$
such that $(q_\init,w,c_\init)\vdash^* (q_\rf,\varepsilon,c)$.
Such a sequence of computation steps is called a \emph{run} of $\cA$ on $w$.
The language $\LL(\cA)$ accepted by~$\cA$ consists of all strings over $A$ 
that are accepted by $\cA$.  
A language $L\subseteq A^*$ is \emph{$S$-recognizable (over $A$)} if $L=\LL(\cA)$ for some $S$-automaton $\cA$ over $A$. The class of \mbox{$S$-recognizable} languages over any alphabet will be denoted by~$S$-REC.
Two storage types $S$ and~$S'$ are \emph{language equivalent} if $S$-REC = $S'$-REC.
Obviously, isomorphic storage types are language equivalent.

\begin{example}\rm \label{ex:stackstring}
We consider the stacks introduced in~\cite{gingrehar67}, in a slight but equivalent variation. Intuitively, a stack is a pushdown over some alphabet $\Omega$, i.e., a nonempty sequence of cells, with the additional ability  of inspecting the contents of all its cells. For this purpose, the stack maintains a ``stack pointer'', which points at the current cell.
In our variation the stack allows the instructions 
$\push(\omega)$, $\pop(\omega)$, $\movedown(\omega)$, and $\moveup(\omega)$ having the following 
meaning: $\push(\omega)$ pushes the symbol $\omega$ on top of the stack, $\pop(\omega)$ pops the top symbol $\omega$, $\movedown(\omega)$ moves the pointer from a cell with content $\omega$ down to the cell below, and $\moveup(\omega)$ moves it from a cell with content $\omega$ up to  the cell above.
As usual, the push- and pop-instructions can only be executed when the stack pointer is at the top of the stack, which remains true after the execution of these instructions.
Figure \ref{fig:new-original-stack-operations} shows examples of these instructions, where we use the stack alphabet $\Omega = \{\alpha,\beta,\gamma\}$.\footnote{Here we use the symbol $\alpha$ as stack symbol and not as arbitrary element of $Ae$.} 
\begin{figure}[t]
  \begin{center}
    \includegraphics[scale=0.4,trim={0 13cm 0cm 5cm},clip]{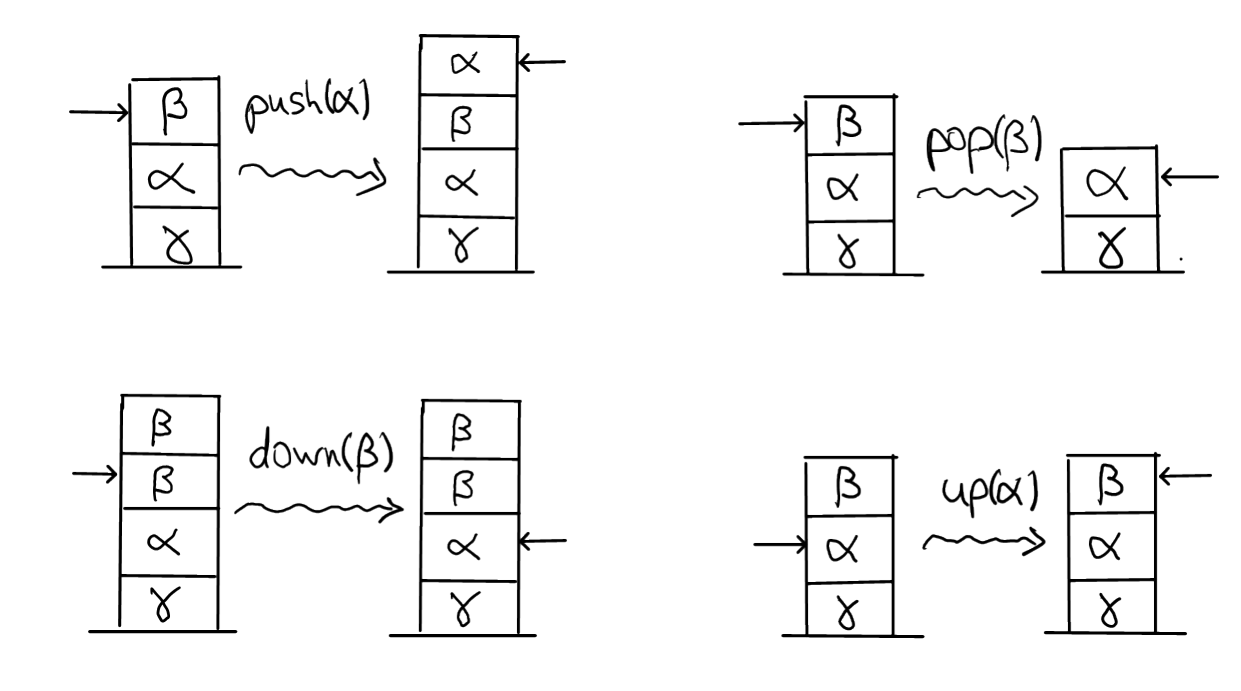}
    \end{center}
\caption{\label{fig:new-original-stack-operations} An illustration of instances of the stack instructions $\push(\alpha)$, $\pop(\beta)$, $\movedown(\beta)$, and $\moveup(\alpha)$.} 
\end{figure}

We will formalize this storage as the storage type $\mathrm{Stack}$. 
To this aim we define the alphabet $\overline{\Omega}= \{\overline{\alpha},\overline{\beta},\overline{\gamma}\}$. 
Then $\mathrm{Stack} = (C,c_\init,\Theta,m)$ 
is the storage defined as follows. 
First, we define stack configurations to be strings over $\Omega\cup\overline{\Omega}$ that contain exactly one occurrence of a symbol in $\overline{\Omega}$, i.e., $C=\Omega^*\,\overline{\Omega}\,\Omega^*$. The last symbol of such a string represents the top of the stack, 
and the unique occurrence of a  barred symbol indicates the position of the stack pointer.
Thus, in Figure~\ref{fig:new-original-stack-operations}, the instruction $\movedown(\beta)$ transforms 
the stack represented by the string $\gamma\alpha\overline{\beta}\beta$ into the stack represented by the string $\gamma\overline{\alpha}\beta\beta$.
Second, $c_\init=\overline{\gamma}$, i.e., the initial stack configuration consists of one cell 
that contains $\gamma$. 
Third, and finally, $\Theta$~consists of all instructions mentioned above, and
$m(\push(\alpha))=\{(w\,\overline{\omega},w\,\omega\,\overline{\alpha}) \mid w\in\Omega^*,\omega\in\Omega\}$, $m(\pop(\alpha))$ is the inverse of $m(\push(\alpha))$, $m(\moveup(\alpha))=\{(w\,\overline{\alpha}\,\omega \,w',w\,\alpha\,\overline{\omega}\,w')\mid w,w'\in\Omega^*,\omega\in\Omega\}$, 
$m(\movedown(\alpha))=\{(w\,\omega\,\overline{\alpha}\,w',w\,\overline{\omega}\,\alpha \,w')\mid w,w'\in\Omega^*,\omega\in\Omega\}$, and similarly for $\beta$ and $\gamma$. 
It is a straightforward exercise to show that the class $\mathrm{Stack}$-REC
of $\mathrm{Stack}$-recognizable languages equals the class of languages accepted by the (one-way, nondeterministic) stack automata of~\cite{gingrehar67}. 

Let $A=\{0,1\}$, and let us consider a Stack-automaton $\cA$ over $A$ that accepts the language $\{ww^\mathrm{R}w\mid w\in A^+\}$, where $w^\mathrm{R}$ is the reverse of the string~$w$. 
We define $\cA= (Q,Q_\init,Q_\rf,T)$ 
with $Q=\{q_1,q_2,q_3,q_4\}$, $Q_\init=\{q_1\}$, and $Q_\rf=\{q_4\}$. 
Let $\sigma:A\to\Omega$
such that $\sigma(0)=\alpha$ and $\sigma(1)=\beta$. 
The set $T$ contains the following transitions, for every $a\in A$. 
\begin{itemize}
\item \noindent
push-phase:

$(q_1,a,\push(\sigma(a)),q_1)$

\item \noindent
movedown-phase:

$(q_1,a,\movedown(\sigma(a)),q_2)$

$(q_2,a,\movedown(\sigma(a)),q_2)$

\item \noindent
moveup-phase:

$(q_2,e,\moveup(\gamma),q_3)$

$(q_3,a,\moveup(\sigma(a)),q_3)$

$(q_3,a,\pop(\sigma(a)),q_4)$
\end{itemize}
The automaton first reads $w$ from the input and pushes its $\sigma$-image symbol by symbol on the stack.
Second, it nondeterministically decides to move down the stack and read $w^\mathrm{R}$ from the input, until it arrives at the bottom symbol~$\gamma$. Third, it uses the $e$-transition to move one cell up, and then moves up the stack reading~$w$. Finally, it nondeterministically decides that it is at the top of the stack, and pops the top symbol while reading it. 
\qed
\end{example}

\begin{example}\rm \label{ex:triv}
The \emph{trivial storage type} (modulo isomorphism) is the storage type 
$\Triv=(C,c_\init,\Theta,m)$ such that $C=\{c\}$,
$c_\init=c$, and $\Theta=\{\theta\}$ with $m(\theta)=\{(c,c)\}$.
It should be clear that a $\Triv$-automaton can be viewed as a finite-state automaton 
that is also allowed to have \mbox{$e$-transitions},
and hence, as is well known, $\Triv$-REC is the class of regular languages.
\qed
\end{example}

Let us define $\cB(S)\subseteq \Theta^*$ to be the set of all
strings $\theta_1\cdots \theta_n$ (with $n\in\nat$ and $\theta_i\in\Theta$ for every $i\in[n]$), 
for which there exist $c_1,\dots,c_{n+1}\in C$ such that 
$c_1=c_\init$ and $(c_i,c_{i+1})\in m(\theta_i)$ for every $i\in[n]$ 
(cf. the  definition of $L_{\cD}$ in \cite[p.~148]{gin75}).
We call such sequences \emph{storage behaviours} or, in particular, $S$-behaviours. 
The next lemma characterizes the $S$-recognizable languages (cf. \cite[Lemma 5.2.3]{gin75}).

\begin{lemma}\rm \label{lm:decomposition2} 
A language $L\subseteq A^*$ 
is $S$-recognizable if and only if there exists a regular language 
\mbox{$R \subseteq (\Ae \times \Theta)^*$} such that
\begin{align*}
L = \{w \in A^* \mid \ & \text{there exist } n\in\nat, 
               \alpha_1,\ldots,\alpha_n\in \Ae, \text{ and } \theta_1,\ldots,\theta_n \in \Theta \\
        & \text{such that } \mu_e(\alpha_1 \cdots \alpha_n)=w, \,\theta_1 \cdots \theta_n \in \cB(S), 
                \text{ and }  \\
        & (\alpha_1,\theta_1) \cdots (\alpha_n,\theta_n) \in R \} \enspace.
\end{align*}
\end{lemma}
\begin{proof}
For every $S$-automaton $\cA = (Q,Q_\init,Q_\rf,T)$ over $A$
we construct the finite-state automaton $\cA' = (Q,Q_\init,Q_\rf,T')$ over $\Ae\times\Theta$ such that 
\[
T'=\{(q,(\alpha,\theta),q')\mid (q,\alpha,\theta,q') \in T\} \enspace. 
\]
It is straightforward to show, using the definitions of $\LL(\cA)$, $\cB(S)$, and $\LL(\cA')$, 
that $L=\LL(\cA)$ and $R=\LL(\cA')$ satisfy the requirements. 
Since the transformation 
of~$\cA$ into~$\cA'$ is a bijection between $S$-automata over $A$ and 
finite-state automata over $\Ae\times\Theta$, this proves the lemma. 
\end{proof}

If in Lemma~\ref{lm:decomposition2} we replace the regular language $R$ by 
a closed formula \mbox{$\phi\in\MSOL(\{*\},\Ae\times\Theta)$}, 
and the expression $(\alpha_1,\theta_1) \cdots (\alpha_n,\theta_n) \in R$
by the expression $\edgr((\alpha_1,\theta_1) \cdots (\alpha_n,\theta_n)) \models \phi$,
as we are allowed to do by Proposition~\ref{pro:MSO-string-graph},
then we essentially obtain the BET-theorem for the storage type $S$
as proved in~\cite{vogdroher16}, where it is generalized to weighted $S$-automata.

It is well known that, under appropriate additional conditions on $S$, the class $S$-REC of $S$-recognizable languages is closed under the full AFL operations \cite[p.~19]{gin75}.
As an example, 
we show, using Lemma~\ref{lm:decomposition2}, that if $S$ has a reset instruction (as in~\cite{gol79}), 
then $S$-REC is closed under concatenation and Kleene star (cf.~\cite[Theorem~3.4]{gol79}). 

Let $S=(C,c_\init,\Theta,m)$ be a storage type. 
A \emph{reset} is an instruction $\theta\in\Theta$ 
such that $m(\theta)=C\times \{c_\init\}$. 

\begin{lemma}\rm \label{lm:reset2}
If $S$ is a storage type that has a reset, then $S$-REC is closed under concatenation and Kleene star.
\end{lemma}

\begin{proof} 
By Lemma~\ref{lm:decomposition2}, every $S$-recognizable language $L$ 
can be ``defined'' by a regular language $R \subseteq (\Ae \times \Theta)^*$. 
For $i\in\{1,2\}$, let $L_i\subseteq A^*$ be defined by the regular language $R_i$.
Let $\chi$ be a reset.
Now let $L$ be the language defined by the regular language $R_1(e,\chi)R_2$.
We observe that, since $\chi$ is a reset, 
$\theta_1 \cdots \theta_n\chi\eta_1\cdots\eta_m$ 
is in $\cB(S)$ if and only if $\theta_1 \cdots \theta_n$ and 
$\eta_1\cdots\eta_m$ are in~$\cB(S)$.
By Lemma~\ref{lm:decomposition2} (applied to $L$), $w\in L$ if and only if there exist $\alpha_1,\ldots,\alpha_n,\beta_1,\dots,\beta_m\in \Ae$, and $\theta_1,\ldots,\theta_n,\eta_1,\dots,\eta_m \in \Theta$ such that 
$(\alpha_1,\theta_1) \cdots (\alpha_n,\theta_n) \in R_1$,
$(\beta_1,\eta_1) \cdots (\beta_m,\eta_n) \in R_2$,
$\mu_e(\alpha_1 \cdots \alpha_ne\beta_1\cdots\beta_m)=w$, and 
$\theta_1 \cdots \theta_n\chi\eta_1\cdots\eta_m\in\cB(S)$.
%
And, again by Lemma~\ref{lm:decomposition2} (applied to $L_1$ and $L_2$) and by the above observation,
that is equivalent to the existence of $w_1\in L_1$ and $w_2\in L_2$ such that $w=w_1w_2$.
Thus, $L$ is the concatenation $L_1 L_2$ of $L_1$ and $L_2$. 

Similarly, if $L\subseteq A^*$ is defined by $R$, 
then $L^*$ is defined by the regular language $(R(e,\chi))^*$.
\end{proof}

By standard techniques it can be shown that if $S$ has an identity, i.e., 
an instruction $\theta$ such that $m(\theta)=\{(c,c)\mid c\in C\}$, then $S$-REC is a full trio, 
i.e., closed under finite-state transductions. 
It is even a full principal trio, generated by the language $\cB(S)$ 
(cf. again \cite[Lemma 5.2.3]{gin75}).
We finally mention that $S$-REC is closed under union for every storage type $S$.

\section{MSO Graph Storage Types}

As stated in the Introduction, our aim in this paper is to define storage types~$S$
for which we can prove a BET-theorem for $S$-recognizable languages such that 
it satisfies the mentioned scheme and every set of graphs $\cG[S,A]$ is MSO-definable. For this, 
we will consider storage types $S=(C,c_\init,\Theta,m)$ such that $C$ is an MSO-definable set of graphs and,
moreover, $m(\theta)$ is represented by an MSO-definable set of graphs for every $\theta\in\Theta$. 
Since $m(\theta)\subseteq C\times C$, i.e., $m(\theta)$~is a set of ordered pairs of graphs, 
this raises the question how to represent a pair of graphs as one single graph, 
and how to define a graph transformation by an MSO-logic formula for such graphs.

\subsection{Pair Graphs}
\label{sec:pairgraphs}

Let $\Sigma$ and $\Gamma$ be alphabets of node labels and edge labels, respectively, 
as in Section~\ref{sec:graphs-MSO}. To model ordered pairs of graphs in $\cG_{\Sigma,\Gamma}$,
we use a special edge label $\nu$ that is not in $\Gamma$. 

A \emph{pair graph} over $(\Sigma,\Gamma)$ is a graph $h$ over $(\Sigma,\Gamma\cup\{\nu\})$ 
for which there is an ordered partition $(V_1,V_2)$ of $V_h$ such that, for every $u,v\in V_h$,
$(u,\nu,v) \in E_h$ if and only if $u\in V_1$ and $v\in V_2$.
The set of all pair graphs over $(\Sigma,\Gamma)$ is denoted by $\cP\cG_{\Sigma,\Gamma}$;
note that this notation does not mention $\nu$.

For a pair graph $h$ as above, 
we call $V_1$ and $V_2$ the \emph{components} of $h$.
Obviously, the above requirements uniquely determine the ordered partition $(V_1,V_2)$.
Thus, we define the ordered pair of graphs represented by $h$ as follows: 
\[
\pair(h)=(h[V_1],h[V_2]) \in \cG_{\Sigma,\Gamma}\times \cG_{\Sigma,\Gamma} \enspace,
\] 
and for a set $H$ of pair graphs we define
\[
\rel(H)=\pair(H)=\{\pair(h)\mid h\in H\}\subseteq \cG_{\Sigma,\Gamma}\times \cG_{\Sigma,\Gamma} \enspace.
\]

Clearly, for given graphs $g_1,g_2\in \cG_{\Sigma,\Gamma}$ 
there is at least one pair graph $h$ 
in~$\cP\cG_{\Sigma,\Gamma}$ such that $\pair(h)=(g_1,g_2)$,
but in general there are many such pair graphs, because there is no restriction on the $\Gamma$-edges 
between the components $V_1$ and $V_2$ of~$h$.
These ``intermediate'' edges can be used to model the (eventual) similarity between $g_1$ and $g_2$,
and allow the description of this similarity by means of an MSO-logic formula to be satisfied by $h$.

A relation $R\subseteq \cG_{\Sigma,\Gamma}\times \cG_{\Sigma,\Gamma}$ is \emph{MSO-expressible} 
if there are an alphabet $\Delta$ and an MSO-definable set of pair graphs 
$H\subseteq \cP\cG_{\Sigma,\Gamma\cup\Delta}$ such that $\rel(H)=R$. 
The alphabet $\Delta$ allows the intermediate edges to carry arbitrary finite information, 
whenever that is necessary. 
We will prove in Section~\ref{sec:msogratra} that all MSO graph transductions 
(in the sense of~\cite[Chapter~7]{coueng12}) are MSO-expressible, using the ``duplicate names'' of those transductions as elements of $\Delta$.
In fact, the notion of MSO-expressibility is inspired by the ``origin semantics'' of 
MSO graph transductions (see, e.g., \cite{boj14,bojdavguipen17,bosmuspenpup18}; 
pair graphs generalize the ``origin graphs'' of \cite{bojdavguipen17}).

\begin{example}\rm
\label{ex:pairgraphs}
As a very simple example, let $\Sigma=\{*\}$ and $\Gamma=\{\gamma\}$, and let 
$C=\{\edgr(\gamma^n)\mid n\in\nat\}$ be the set of all string graphs over $(\Sigma,\Gamma)$.
We show that the identity on $C$ is MSO-expressible by a formula~$\phi$
such that $\LL(\phi)\subseteq \cP\cG_{\Sigma,\Gamma}$ (thus, $\Delta=\emptyset$). 
The set $H=\LL(\varphi)$ consists of all graphs $h$ over $(\Sigma,\Gamma\cup\{\nu\})$
such that 
\begin{figure}[t]
  \begin{center}
    \includegraphics[scale=0.4,trim={2cm 4cm 10cm 2cm},clip]{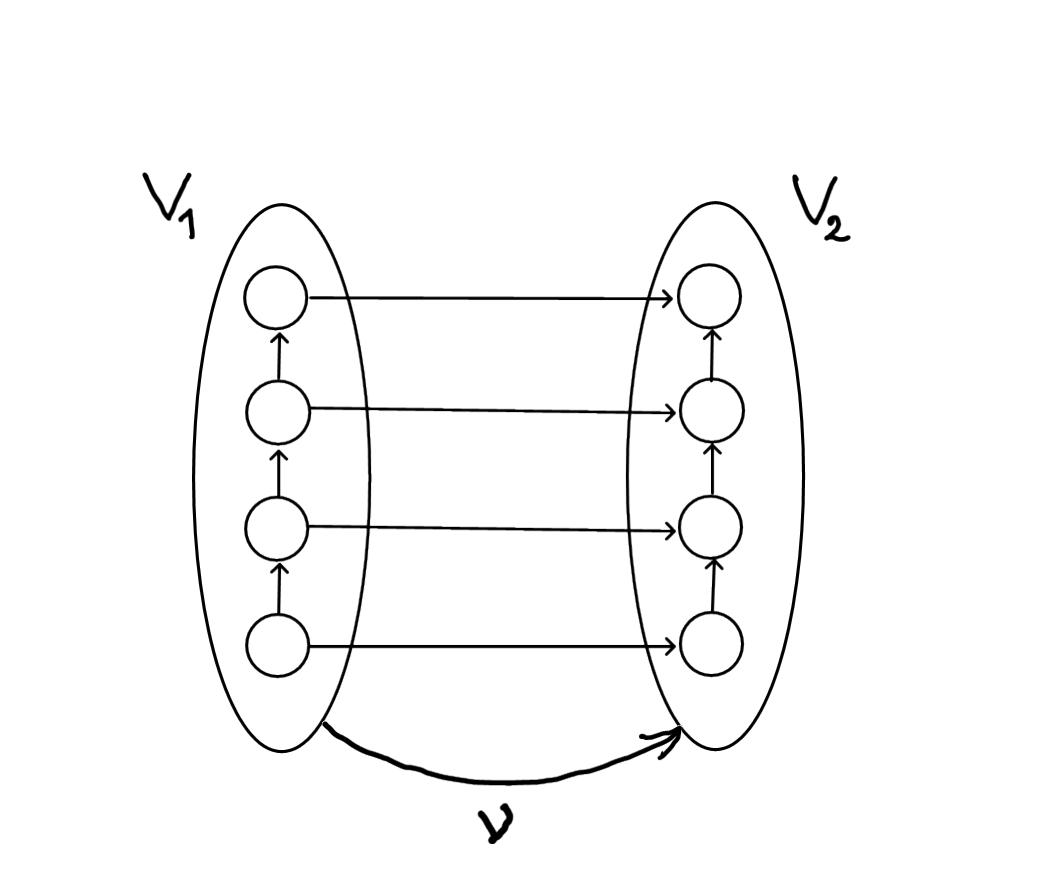}
    \end{center}
\caption{\label{fig:pair-graph} A pair graph $h\in \LL(\phi)$ such that 
$\pair(h)=(\edgr(\gamma^3),\edgr(\gamma^3))$.
All nodes have label $*$, and all straight edges have label $\gamma$.
The components $V_1$~and $V_2$ of $h$ are represented by ovals.
The $\nu$-edge from the first to the second oval represents all sixteen $\nu$-edges 
from the nodes of $V_1$ to the nodes of $V_2$.} 
\end{figure}
$V_h=V_1\cup V_2$ where $V_1=\{u_1,\dots,u_{n+1}\}$ 
and $V_2=\{v_1,\dots,v_{n+1}\}$ for some $n\in\nat$, 
and $E_h$ consists of 
\begin{compactitem}
\item the edges $(u_i,\gamma,u_{i+1})$ and $(v_i,\gamma,v_{i+1})$ for every $i\in[n]$,
which turn $V_1$ and $V_2$ into string graphs, 
\item the intermediate edges $(u_i,\gamma,v_i)$ for every $i\in[n+1]$, and 
\item the edges $(u_i,\nu,v_j)$ for every $i,j\in[n+1]$, 
which turn $h$ into a pair graph with the ordered partition $(V_1,V_2)$.
\end{compactitem}
It should be clear that $\pair(h)=(\edgr(\gamma^n),\edgr(\gamma^n))$, and hence 
$\rel(H)= \{(g,g)\mid g\in C\}$.
An example of a pair graph in $H$ is shown in Figure~\ref{fig:pair-graph}.

To show that $H$ is MSO-definable, we now describe the graphs $h\in H$ in such a way 
that the existence of~$\varphi$ should be clear to the reader. 
First, the set of nodes of $h$ is partitioned into two nonempty sets $X_1$ and~$X_2$ 
(node-set variables that correspond to $V_1$ and $V_2$ above), such that $h$ is a pair graph 
with ordered partition $(X_1,X_2)$. This part of $\phi$ can be obtained directly 
from the definition of pair graph. Second, for each $i\in\{1,2\}$, 
the subgraph $h[X_i]$ of $h$ induced by $X_i$
should satisfy the formula $\psi=\sstring_\Gamma$ of Example~\ref{ex:strings}; this can be expressed by
the relativization $\psi|_{X_i}$ of $\psi$ to $X_i$.
Third, the intermediate edges form a bijection between $X_1$ and $X_2$.
Moreover, that bijection should be a graph isomorphism  
between the induced subgraphs $h[X_1]$ and $h[X_2]$, i.e., 
for all $u,u'\in X_1$ and $v,v'\in X_2$, 
if $(u,\gamma,u'),(u,\gamma,v),(u',\gamma,v')\in E_h$, then $(v,\gamma,v')\in E_h$.
This ends the description of the graphs $h\in H$. 

We note that the intermediate edges $(u_i,\gamma,v_i)$ between the two components of $h$ are essential. 
If we drop them from each $h\in H$, then the resulting set of pair graphs is not MSO-definable. 
\qed
\end{example}

\subsection{Graph Storage Types}
\label{sec:graphstorage}

As observed at the beginning of this section,
we are interested in storage types $(C,c_\init,\Theta,m)$ such that $C$~is an MSO-definable set of graphs 
and, for every $\theta\in\Theta$, $m(\theta)$ is MSO-expressible, i.e., it is the binary relation on $C$ 
determined by an MSO-definable set of pair graphs. 

A storage type $S=(C,c_\init,\Theta,m)$ is an \emph{MSO graph storage type over $(\Sigma,\Gamma)$} 
if 
\begin{itemize}
\item $C=\LL(\phi_\rc)$ for some closed formula $\phi_\rc$ in $\MSOL(\Sigma,\Gamma)$, 

\item $\Theta$ is an exclusive set of closed formulas in $\MSOL(\Sigma,\Gamma\cup\{\nu\})$ such that $\LL(\theta)\subseteq \cP\cG_{\Sigma,\Gamma}$
for every $\theta\in\Theta$, and

\item $m(\theta) = \rel(\LL(\theta))$ for every $\theta\in\Theta$.
\end{itemize}
Note that $\Theta$ is required to be exclusive, which means that $\LL(\theta)$ and $\LL(\theta')$ are disjoint for distinct formulas $\theta$ and $\theta'$ in $\Theta$. 
Note also that for every formula $\theta\in\Theta$,
if $h\in \LL(\theta)\subseteq \cP\cG_{\Sigma,\Gamma}$ and $\pair(h)=(g_1,g_2)$, then 
intuitively, $g_1$ and $g_2$ are the storage configurations before and after 
execution of the instruction $\theta$.

From now on we will specify an MSO graph storage type $S=(C,c_\init,\Theta,m)$ as
$S=(\phi_\rc,g_\init,\Theta)$, such that $C=\LL(\phi_\rc)$, $c_\init = g_\init$, 
and $m$ is fixed by the above requirement.
An example of an MSO graph storage type will be given below in Example~\ref{ex:stack}.

By definition, the storage transformations of an MSO graph storage type~$S$ over $(\Sigma,\Gamma)$ are 
MSO-expressible with $\Delta=\emptyset$. However, it may well be that $C$ is an MSO-definable subset of 
$\cG_{\Sigma,\Gamma'}$ for some subset $\Gamma'$ of $\Gamma$, in which case the MSO-expressible storage transformations of $S$ may employ $\Delta=\Gamma\setminus\Gamma'$. Thus, arbitrary MSO-expressible relations can be used as storage transformations of an MSO graph storage type.
Similarly, the requirement that $\Theta$ is exclusive, is not restrictive (with respect to isomorphism of storage types).
If an instruction $\theta_0\in \Theta$ overlaps with another instruction $\theta_1\in \Theta$,
i.e., $\LL(\theta_0)\cap \LL(\theta_1)\neq\emptyset$, then we can take two edge labels $d_0$ and $d_1$,
add them to $\Gamma$, change every pair graph in $\LL(\theta_i)$ by adding $d_i$-edges 
from all nodes of its first component to all nodes of its second component, and change 
$\theta_i$ into  $\theta_i \wedge \forall x,y. ((\edge_{d_i}(x,y) \leftrightarrow \edge_\nu(x,y)) 
\wedge \neg \edge_{d_{1-i}}(x,y))$. Then we obtain an 
isomorphic MSO graph storage type over $(\Sigma,\Gamma\cup\{d_0,d_1\})$ that satisfies the requirement.
BET-theorems can be shown for MSO graph storage types without this requirement (cf. the first paragraph of Section~\ref{sec:MSO-expres}), but we have adopted it for technical convenience.
Altogether, the next observation can easily be shown.

\begin{observation}\label{ob:isom}\rm 
If $S=(C,c_\init,\Theta,m)$ is a storage type
such that $C$ is an MSO-definable set of graphs and,  
for every $\theta\in\Theta$, the storage transformation $m(\theta)$ is MSO-expressible, 
then $S$ is isomorphic to an MSO graph storage type.
\end{observation}

The closure properties of the class $S$-REC of $S$-recognizable languages, 
discussed in Section~\ref{sec:storagetypes}, also hold, of course, 
for every MSO graph storage type $S=(\phi_\rc,g_\init,\Theta)$ over $(\Sigma,\Gamma)$.
Note that we can always (if we so wish) enrich~$\Theta$ 
with a reset, as follows. 
For a graph $g\in \LL(\phi_\rc)$, let $h$ be the unique pair graph 
such that $\pair(h)=(g,g_\init)$ 
and there are no $\Gamma$-edges between the components of~$h$. Obviously, the set of all such graphs $h$ 
is MSO-definable by a formula~$\theta$, which is then a reset. In the case where 
$\Theta\cup\{\theta\}$ is not exclusive, we can add (dummy) $\Gamma$-edges between the components of $h$  
with a new label (which, possibly, has to be added to $\Gamma$). 
Similarly we can add an identity instruction to~$\Theta$, cf. Example~\ref{ex:pairgraphs}.

\begin{example}\label{ex:stack}
\rm
We define an MSO graph storage type $\mathrm{STACK} = (\phi_\rc,g_\init,\Psi)$ that is isomorphic to the storage type $\mathrm{Stack}=(C,c_\init,\Theta,m)$ of Example~\ref{ex:stackstring}.
Let $\Omega = \{\alpha,\beta,\gamma\}$ and $\overline{\Omega}= \{\overline{\alpha},\overline{\beta},\overline{\gamma}\}$, as in Example~\ref{ex:stackstring}.
To model stacks and stack transformations as graphs, 
we define the alphabet $\Sigma = \Omega \cup \overline{\Omega}$ of node labels, 
and the alphabet $\Gamma = \{*,d\}$ of edge labels.  
The symbol $d$ will be used to label the intermediate edges of pair graphs; 
it is not really needed, but will be useful later. 
\begin{figure}[t]
  \begin{center}
    \includegraphics[scale=0.4,trim={0 18cm 3cm 0},clip]{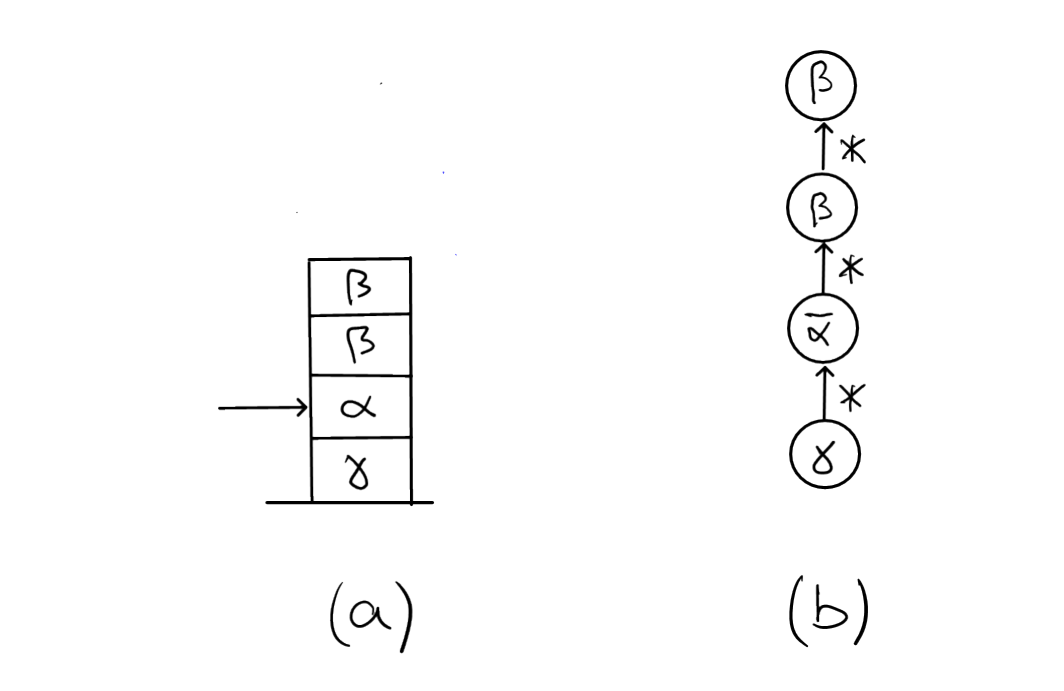}
    \end{center}
\caption{\label{fig:stack-graph} (a) A stack configuration and (b) its representation as a graph over $(\Sigma,\{*\})$.}  
\end{figure}
First, each stack $w\in C=\Omega^*\,\overline{\Omega}\,\Omega^*$ is represented by the string graph
$\text{nd-gr}(w)\in\cG_{\Sigma,\{*\}}$, as defined in Section~\ref{sec:reg-lang}.
Figure~\ref{fig:stack-graph} shows an example of a stack and its representation as a graph in $\cG_{\Sigma,\Gamma}$ (with $w=\gamma\,\overline{\alpha}\,\beta\,\beta$).

The closed formula $\varphi_\rc \in \MSOL(\Sigma,\{*,d\})$ 
such that $\LL(\varphi_\rc)$ is the set of all possible stack configurations, is defined by
\begin{align*}
\varphi_\rc &=  \sstring_{\Gamma} \wedge \forall x,y.(\neg\,\edge_d(x,y)) \wedge \unique \\
\unique &= (\exists x.\lab_{\,\overline{\Omega}\,}(x)) \wedge  \forall x,y.( \lab_{\,\overline{\Omega}\,}(x) \wedge \lab_{\,\overline{\Omega}\,}(y) \to (x=y))\\
\lab_{\,\overline{\Omega}\,}(x) &= \bigvee_{\omega \in \Omega}\lab_{\,\overline{\omega}\,}(x)
\end{align*}
where $\sstring_{\Gamma}$ 
is the formula of Example~\ref{ex:strings}. 

Second, $g_\init=\ndgr(\overline{\gamma})$. 
Third, and finally, the set $\Psi$ of STACK instructions consists of all formulas 
$\psi_\theta \in \MSOL(\Sigma,\{*,d,\nu\})$ that model a stack instruction $\theta\in\Theta$. 
We will show three examples for $\theta$:
$\push(\alpha)$, $\pop(\alpha)$, and $\moveup(\beta)$.
The formulas for the other stack instructions in $\Theta$ can be obtained in a similar way.

\underline{$\theta=\push(\alpha)$:} 
We describe the formula $\psi_\theta$ similarly to Example~\ref{ex:pairgraphs}.
The set $\LL(\psi_\theta)$ consists of all graphs $h=(V,E,\nlab)$
such that (see Figure~\ref{fig:push-alpha-beta}(b) for an example)
\begin{figure}[t]
  \begin{center}
    \includegraphics[scale=0.4,trim={0 10cm 1cm 6cm},clip]{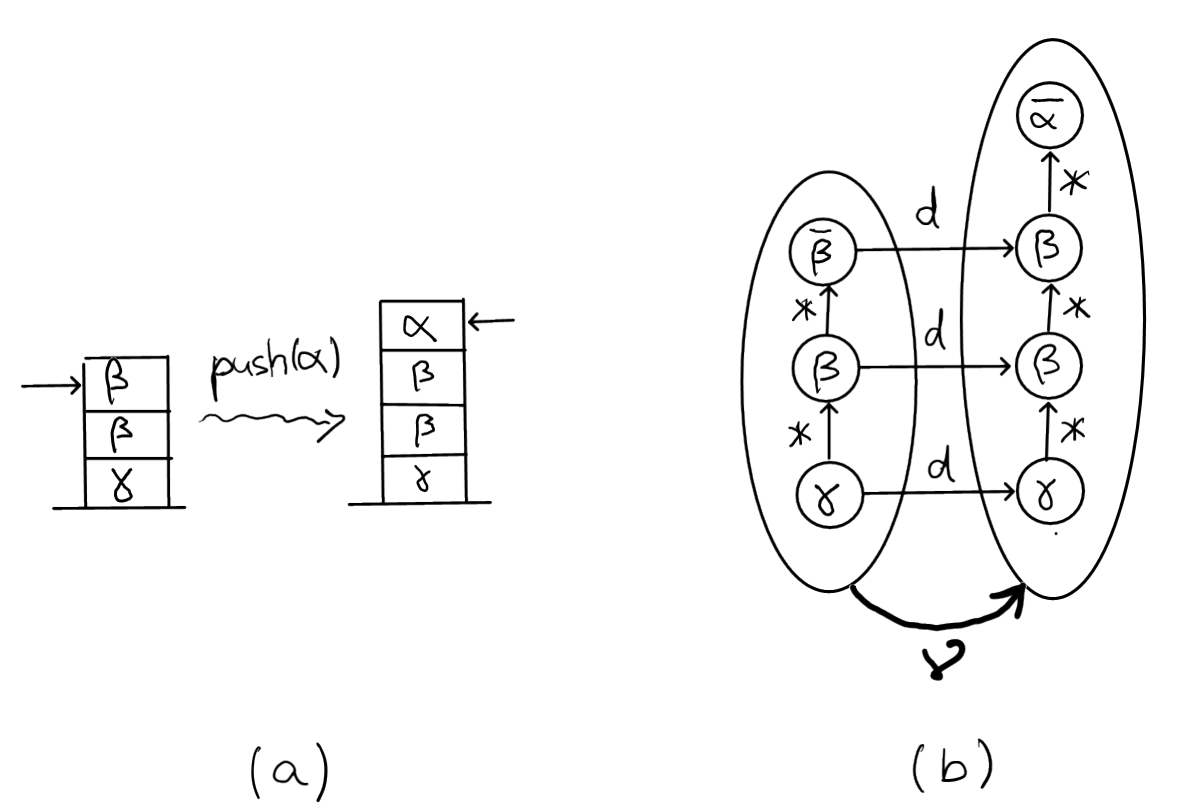}
    \end{center}
\caption{\label{fig:push-alpha-beta} (a) An instance of the execution of
the stack instruction $\theta=\push(\alpha)$.  (b) A pair graph $h$ in $\LL(\psi_\theta)$ that realizes (a).}
\end{figure}
\begin{itemize}
\item[$(1)$] $V=V_1\cup V_2$ where $V_1=\{u_1,\dots,u_n\}$ 
  and $V_2=\{v_1,\dots,v_n,v_{n+1}\}$ for some $n\geq 1$;

\item[$(2)$] $E$ consists of
\begin{itemize}
\item the edges $(u_i,*,u_{i+1})$ and $(v_j,*,v_{j+1})$ for every $i\in[n-1]$ and $j\in[n]$,
which turn $V_1$ and $V_2$ into string graphs, 
\item the intermediate edges $(u_i,d,v_i)$ for every $i\in[n]$, and 
\item the edges $(u_i,\nu,v_j)$ for every $i\in[n]$ and $j\in[n+1]$,
which turn $h$ into a pair graph with the ordered partition $(V_1,V_2)$;
\end{itemize}
\item[$(3)$] the node label function $\nlab$ satisfies
\begin{itemize}
\item $\nlab(v_i)\in\Omega$ for every $i\in[n]$, 
\item $\nlab(v_{n+1})=\overline{\alpha}$, 
\item $\nlab(u_i)=\nlab(v_i)$ for every $i\in[n-1]$, and
\item $\nlab(u_n)=\overline{\nlab(v_n)}$.
\end{itemize}
\end{itemize}
Intuitively, $h[V_1]$ and $h[V_2]$ are the stacks 
before and after execution of the push-instruction.
The $d$-edge from $u_i$ to $v_i$ indicates that $v_i$ is a copy (or duplicate) of $u_i$.

To show that this set of graphs is MSO-definable, we now describe the graphs $h\in \LL(\psi_\theta)$ 
in a suggestive way, as in Example~\ref{ex:pairgraphs}.
First, the set $V$ of nodes of $h$ is partitioned into two nonempty sets $X_1$ and $X_2$, 
such that $h$ is a pair graph with ordered partition $(X_1,X_2)$. 
Second, for each $i\in\{1,2\}$, the induced subgraph $h[X_i]$ should satisfy the formula $\varphi_\rc$,
i.e., $h$ satisfies $(\phi_\rc)|_{X_i}$.
Third, the $d$-edges form a bijection from $X_1$ to $X_2\setminus\{t_2\}$ 
where $t_2$ is the top of $X_2$, i.e., the unique element of $X_2$ that has no outgoing $*$-edge. 
Moreover, that bijection should be a graph isomorphism between $h[X_1]$ and $h[X_2\setminus\{t_2\}]$
(disregarding node labels), i.e., for all $u,u'\in X_1$ and $v,v'\in X_2$, 
if $(u,*,u'),(u,d,v),(u',d,v')\in E$, then $(v,*,v')\in E$. Fourth and finally,
the requirements in (3) above should be satisfied by $\nlab$. Let $t_1$ be the top of $X_1$.
If $(u,d,v)\in E$ and $u\neq t_1$, then $\nlab(u)=\nlab(v)\in\Omega$.
If $(t_1,d,v)\in E$, then $\nlab(t_1)=\overline{\nlab(v)}$.
And $\nlab(t_2)=\overline{\alpha}$.  This ends the description of the graphs $h\in \LL(\psi_\theta)$. 

\underline{$\theta=\pop(\alpha)$:}
The pair graphs in $\LL(\psi_\theta)$ are obtained from those in $\LL(\psi_{\push(\alpha)})$,
as described in the previous example, by inverting all $\nu$-edges and $d$-edges
(see Figure~\ref{fig:pop-alpha}(b) for an example). 
Thus, they have the ordered partition $(V_2,V_1)$. 
The construction of the formula $\psi_{\pop(\alpha)}$ is symmetric to 
the construction of the formula $\psi_{\push(\alpha)}$.

\begin{figure}[t]
  \begin{center}
     \includegraphics[scale=0.4,trim={0 10cm 0 8cm},clip]{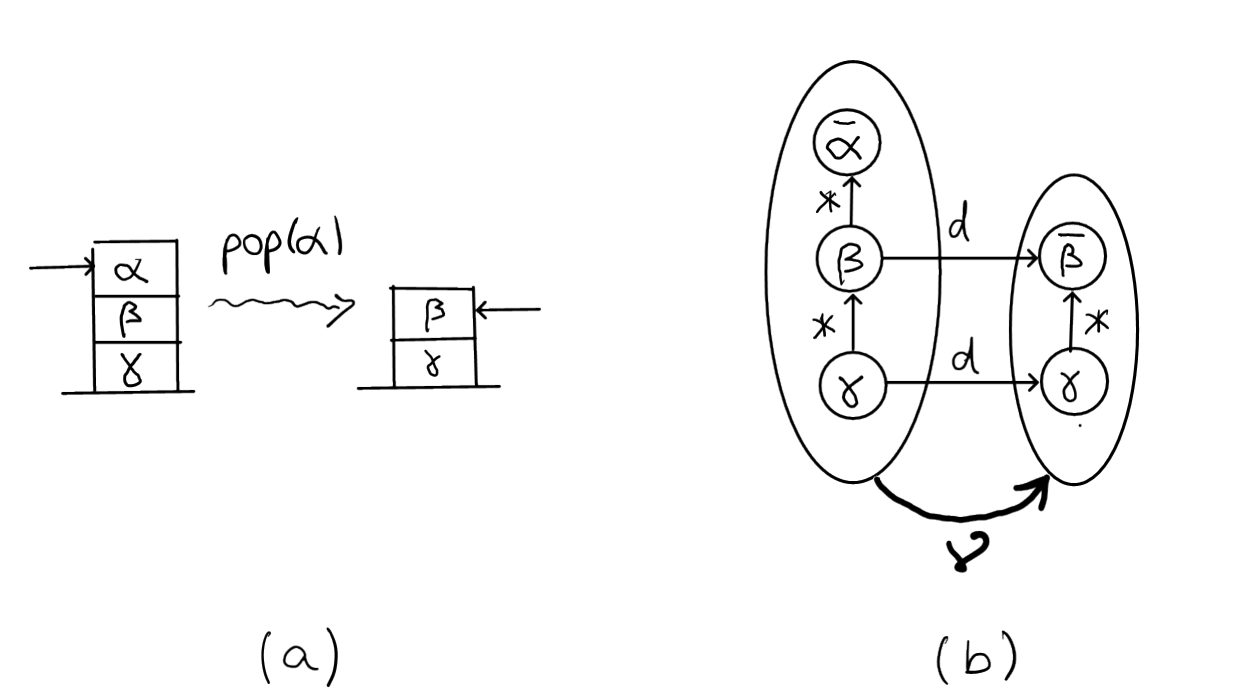} 
\end{center}
\caption{\label{fig:pop-alpha} (a) An instance of the execution of the stack instruction $\theta=\pop(\alpha)$. (b) A graph $h \in \LL(\psi_\theta)$ that realizes (a).} 
\end{figure}

\underline{$\theta=\moveup(\beta)$:} 
The set $\LL(\psi_\theta)$ consists of all graphs $h=(V,E,\nlab)$
such that (see Figure~\ref{fig:moveup-beta}(b) for an example)

\begin{figure}[t]
  \begin{center}
    \includegraphics[scale=0.4,trim={0 10cm 0 8cm},clip]{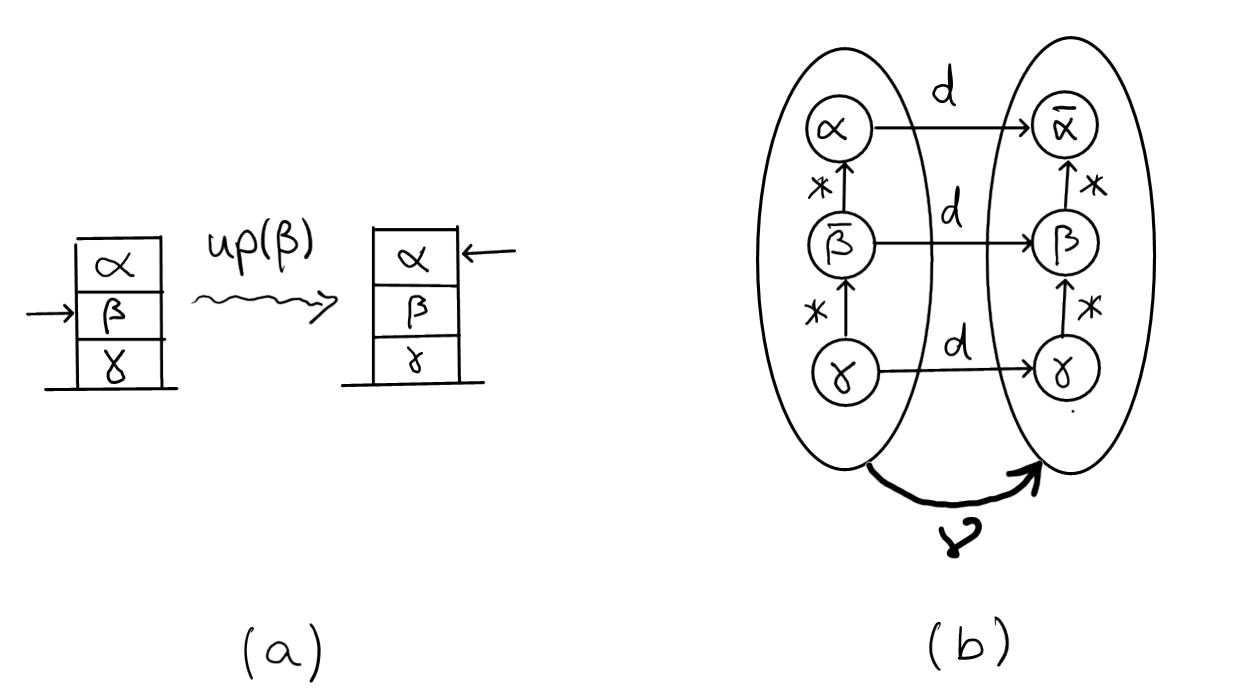}
    \end{center}
\caption{\label{fig:moveup-beta} 
(a) An instance of the execution of the stack instruction $\theta=\moveup(\beta)$. (b) A graph $h\in \LL(\psi_\theta)$ that realizes (a).} 
\end{figure}

\begin{itemize}
\item[$(1)$] $V=V_1\cup V_2$ where $V_1=\{u_1,\dots,u_n\}$ 
  and $V_2=\{v_1,\dots,v_n\}$ for some $n\geq 2$;
\item[$(2)$] $E$ consists of 
\begin{itemize} 
\item the edges $(u_i,*,u_{i+1})$ and $(v_i,*,v_{i+1})$ for every $i\in[n-1]$,
which turn $V_1$ and $V_2$ into string graphs, 
\item the intermediate edges $(u_i,d,v_i)$ for every $i\in[n]$, and 
\item the edges $(u_i,\nu,v_j)$ for every $i,j\in[n]$,
which turn $h$ into a pair graph with the ordered partition $(V_1,V_2)$;
\end{itemize}
\item[$(3)$] the node label function $\nlab$ satisfies 
the following requirements for some $i \in [n-1]$:
\begin{itemize}
\item $\nlab(u_i)=\overline{\beta}$,
\item $\nlab(u_j)\in\Omega$ for every $j \in [n]\setminus \{i\}$,
\item $\nlab(v_{i+1})=\overline{\nlab(u_{i+1})}$, 
\item $\nlab(v_i)=\beta$, and
\item $\nlab(v_j)=\nlab(u_j)$ for every $j \in [n]\setminus \{i,i+1\}$.
\end{itemize}
\end{itemize}
We now describe the graphs $h\in \LL(\psi_\theta)$ in a suggestive way. 
The first two steps are the same as for $\theta = \push(\beta)$.
Third, the $d$-edges form a bijection from~$X_1$ to~$X_2$. 
Moreover, that bijection should be a graph isomorphism between $h[X_1]$ and $h[X_2]$
(disregarding node labels). Finally,
the requirements in~(3) above should be satisfied by $\nlab$. There should exist an element $p_1$ of~$X_1$ with label~$\overline{\beta}$, and an element~$p_1'$ of $X_1$ such that $(p_1,*,p_1') \in E$. 
Let $(p_1,d,p_2) \in E$ and $(p_1',d,p_2') \in E$. 
Then $\nlab(p_2')= \overline{\nlab(p_1')}$ and $\nlab(p_2)=\beta$. 
\qed
\end{example}

\begin{example}\rm \label{ex:MSO-graph-storage-triv} 
The storage type $\Triv$ from Example~\ref{ex:triv} is isomorphic to
the MSO graph storage type 
$\TRIV= (\phi_\rc,g_\init,\{\theta\})$ over $(\{*\},\emptyset)$ such that 
$\LL(\phi_\rc)=\{g_\init\}$ where $g_\init$ is the graph with one $*$-labeled node (and no edges), and 
\mbox{$\LL(\theta)=\{h\}$} where $h$ is the (pair) graph 
with two $*$-labeled nodes and a $\nu$-labeled edge from one node to the other. 
\qed
\end{example}

\section{Graph Automata}

Let $S=(\phi_\rc,g_\init,\Theta)$ be an MSO graph storage type over $(\Sigma,\Gamma)$ 
and let $\cA=(Q,Q_\init,Q_\rf,T)$ be an $S$-automaton over the input alphabet $A$. 
Recall from Section~\ref{sec:storagetypes} that $\Ae=A\cup\{e\}$, 
where $e\notin A$ represents the empty string. 
Since the storage configurations of $\cA$, and its storage transformations, are specified by (MSO-definable) sets of graphs in $S$, we can imagine a different interpretation of~$\cA$, viz. as a finite-state automaton that accepts graphs. Rather than keeping track of its storage configurations in private memory, the automaton $\cA$ checks that its input graph represents, in addition to an input string $w\in A^*$, a correct sequence of storage configurations corresponding to a run of $\cA$ on $w$. Moreover, $\cA$ also checks that the input graph contains the intermediate edges (between the storage configurations) corresponding to the pair graphs of the instructions $\theta\in\Theta$ applied by $\cA$ in that run. 
A possible input graph of $\cA$ will be called a ``string-like'' graph, because it represents both a string over $A$, and a sequence of graphs with intermediate edges between consecutive graphs. More precisely, it represents a string over $\Ae$, taking into account the $e$-transitions of $\cA$.
Thus, the length of the sequence of graphs is the length of that string plus one. 
The sequence of graphs will be determined by  $\Ae$-edges (similar to the $\nu$-edges in pair graphs).

Since the input graphs will contain both $\Gamma$-edges and $\Ae$-edges,
we assume, without loss of generality, that $\Gamma\cap \Ae=\emptyset$.
Thus, we consider graphs over $(\Sigma,\Gamma\cup \Ae)$.
We first define ``string-like'' graphs, and then the way in which an $S$-automaton $\cA$ 
can be viewed as an acceptor of such graphs. 
An example of a string-like graph is shown in Figure \ref{fig:comput} 
(for $S=\mathrm{STACK}$ and $A=\{0,1\}$).

\subsection{String-like Graphs}
\label{sec:string-like-graphs}

A graph $g=(V,E,\nlab)\in\cG_{\Sigma,\Gamma\cup\Ae}$ is \emph{string-like} (over $S$ and $A$) if 
there are $n \in \nat$, $\alpha_1,\ldots,\alpha_n\in \Ae$, and an ordered partition 
$(V_1,\dots,V_{n+1})$ of $V$ such that
\begin{itemize}
\item[$(1)$] for every $\gamma\in \Gamma$ and $u,v \in V$, 
if $(u,\gamma,v)\in E$, then either there exists $i\in[n+1]$ 
such that $\{u,v\}\subseteq V_i$ or 
there exists $i\in[n]$ such that $\{u,v\}\subseteq V_i\cup V_{i+1}$;
\item[$(2)$] for every $\alpha\in \Ae$ and $u,v \in V$, 
$(u,\alpha,v) \in E$ if and only if there exists $i\in[n]$ such that 
$\alpha=\alpha_i$, $u\in V_i$, and $v\in V_{i+1}$; 
\item[$(3)$] $g[V_1]=g_\init$.
\end{itemize}
We call each set $V_i$ (with $i \in [n+1]$) a \emph{component of $g$}, and we call
the string $\alpha_1\cdots\alpha_n$ over $\Ae$ the \emph{trace of $g$}.  

Intuitively, $g$ can be viewed as a sequence of graphs $g_1,\dots,g_{n+1}$ over $(\Sigma,\Gamma)$
with additional $\Gamma$-edges between consecutive graphs $g_i$ and $g_{i+1}$; moreover, 
\mbox{$\alpha_i$-edges} are added from every node of $g_i$ to every node of $g_{i+1}$;
finally, we require $g_1$ to be the initial storage configuration of $S$. 
Clearly, the $\Ae$-edges uniquely determine the components $V_1,\dots,V_{n+1}$ and their order, and 
also uniquely determine the trace $\alpha_1\cdots\alpha_n$.
Thus, we define 
\[
\com(g)=(V_1,\dots,V_{n+1}) \text{ and }
\tr(g)=\alpha_1\cdots\alpha_n \in \Ae^* \enspace.
\]
Note that for every $i\in[n+1]$, $g[V_i]=g_i\in\cG_{\Sigma,\Gamma}$, and note
that for every $i\in[n]$, the graph $h=\lambda_{\Ae,\nu}(g[V_i\cup V_{i+1}])$ is a pair graph 
such that $\pair(h)=(g_i,g_{i+1})$, because the mapping $\lambda_{\Ae,\nu}$ 
changes every $\Ae$-edge into a $\nu$-edge.
Vice versa, if $A=\{\nu\}$, then a pair graph $g$ is a string-like graph such that $\tr(g)=\nu$.

We will denote the set of all string-like graphs over $S$ and $A$ by $\cG[S,A]$;
thus, in this notation $(\Sigma,\Gamma)$ and $e$ are implicit. 

If each of $g$'s components is a singleton, then the graph $g'$ that is obtained from $g$ by dropping 
the $\Gamma$-edges, is a string graph, as defined in Section~\ref{sec:graphs-MSO}. 
In particular, if $\Sigma=\{*\}$ and $\tr(g)=\tau \in \Ae^*$, then $g'$ is the string graph $\edgr(\tau)$
defined in Section~\ref{sec:reg-lang}, which is a unique graph representation of the string $\tau$.
Clearly, if $\Sigma=\{*\}$ and $g_\init$ is the graph with one $*$-labeled node (and no edges), then,
among all graphs in $\cG[S,A]$ with trace $\tau$, $\edgr(\tau)$ has
the minimal number of nodes and edges; note that in particular $\edgr(\tau) \in \cG[\TRIV,A]$,
for the MSO graph storage type $\TRIV$ defined in Example~\ref{ex:MSO-graph-storage-triv} in which 
$\Sigma=\{*\}$ and $\Gamma=\emptyset$.

We finally define ``$w$-like'' graphs, where $w$ is a string over the alphabet~$A$.
A~graph $g\in\cG[S,A]$ is \emph{$w$-like} if $\mu_e(\tr(g))=w$ 
(where $\mu_e$  is the string homomorphism from $\Ae$ to $A$
that erases $e$, cf. Section~\ref{sec:storagetypes}).
For instance, the graph in Figure~\ref{fig:comput} is $011001$-like.  
For every string $w\in A^*$, we denote by $\cG[S,w]$ the set of $w$-like graphs in $\cG[S,A]$.
According to the scheme of BET-theorems discussed in the Introduction, every $w$-like graph can be viewed as an ``extension'' of the string~$w$; the mapping $\tr\circ \mu_e: \cG[S,A]\to A^*$ (i.e., $\tr$ followed by $\mu_e$) corresponds to the mapping $\pi$ in that discussion. 

It should be noted that in a string-like graph $g\in\cG[S,A]$, 
two nodes $u$~and~$v$ of~$g$ are in the same component if and only if $u\equiv_\Ae v$,
which means that they have the same neighbours in $g$ (with respect to $\Ae$),
as defined in Section~\ref{sec:graphs-MSO}.
Since $\cG[S,A]\subseteq \cG_{\Sigma,\Gamma\cup\Ae}$,
the logic $\MSOL(\Sigma,\Gamma\cup\Ae)$ will be used to describe properties of string-like graphs.
In that logic we will use the formula 
\[
\eq(x,y) =  \forall z.\bigwedge_{\alpha \in \Ae}((\edge_\alpha(z,x) \leftrightarrow \edge_\alpha(z,y))\wedge
                      (\edge_\alpha(x,z) \leftrightarrow \edge_\alpha(y,z)))
\]
which expresses that the nodes $x$ and $y$ are $\Ae$-equivalent, i.e., 
for every $g\in\cG_{\Sigma,\Gamma\cup\Ae}$
and $u,v\in V_g$, $(g,u,v)\models \eq(x,y)$ if and only if $u\equiv_\Ae v$.

We now prove our intuitive requirement that the set of graphs $\cG[S,A]$ should be MSO-definable, cf. 
the discussion on the scheme of BET-theorems in the Introduction. 

\begin{observation}\label{ob:stringlike}\rm 
The set $\cG[S,A]$ of string-like graphs is MSO-definable.
\end{observation}

\begin{proof} We define a closed formula `$\sstringlike$' in $\MSOL(\Sigma,\Gamma\cup\Ae)$ such that $\LL(\sstringlike) = \cG[S,A] =
\{g \in \cG_{\Sigma,\Gamma\cup\Ae} \mid g \text{ is a string-like graph} \}$. 
We let 
\[
\sstringlike = \phi_2\wedge\phi_1\wedge\phi_3
\]
where $\phi_i$ expresses condition~($i$) of the definition of string-like graphs.

As observed above, the components of a string-like graph are the equivalence classes 
of the equivalence relation $\equiv_\Ae$. 
As observed in Section~\ref{sec:graphs-MSO} for an arbitrary graph, 
the equivalence relation $\equiv_\Ae$ 
is a congruence with respect to the $\Ae$-edges, and 
there are no $\Ae$-edges within an equivalence class. 
Hence, to express condition~(2), it suffices to require that
the equivalence classes of $\equiv_\Ae$ form a string, in the following sense:
the graph with the equivalence classes as nodes and an $\alpha$-edge from one equivalence class to another 
if there is an $\alpha$-edge from every element of the one to every element of the other, is a string graph.
Thus, the formula $\varphi_2$ is obtained from the formula $\sstring_{\Gamma}$ of Example~\ref{ex:strings}
by changing $\Gamma$~into $\Ae$, $z=x$ into $\eq(z,x)$, and $y=z$ into $\eq(y,z)$, everywhere.

To express condition~(1) we define 
\[
\phi_1 = \forall x,y.(\edge_\Gamma(x,y) \rightarrow \eq(x,y) \vee \edge_\Ae(x,y) \vee \edge_\Ae(y,x)) .
\]

To express condition~(3), let $\psi$ be a formula such that $\LL(\psi)=\{g_\init\}$,
and let 
\[
\first(X) = \forall x.(x\in X \leftrightarrow (\neg \exists y. \edge_\Ae(y,x)))
\]
which expresses that $X$ is the first component of the string-like graph.
Then we define $\varphi_3$ to be the formula $\forall X.(\first(X)\to\psi|_X)$.
\end{proof}

\subsection{Graph Acceptors}
\label{sec:automaton-model}

As at the start of the section, 
let $S=(\phi_\rc,g_\init,\Theta)$ be an MSO graph storage type over $(\Sigma,\Gamma)$ 
and let $\cA=(Q,Q_\init,Q_\rf,T)$ be an $S$-automaton over $A$. 
We now interpret $\cA$ as an acceptor of string-like graphs.

Let $g$ be a string-like graph over $S$ and $A$, i.e., $g \in \cG[S,A]$, and let 
$\com(g)=(V_1,\dots,V_{n+1})$ and $\tr(g)=\alpha_1\cdots \alpha_n$, 
for some $n \in \nat$ and $\alpha_i \in \Ae$ for each $i \in [n]$. 
The graph $g$ is \emph{accepted by} $\cA$ 
if there exist $q_1,\dots,q_{n+1}\in Q$ and $\theta_1,\dots,\theta_n\in\Theta$ such that 
(1)~$q_1\in Q_\init$, 
(2)~for every $i\in[n]$ the transition $(q_i,\alpha_i,\theta_i,q_{i+1})$ is in~$T$ and 
$\lambda_{\Ae,\nu}(g[V_i\cup V_{i+1}])\in \LL(\theta_i)$, and 
(3)~$q_{n+1}\in Q_\rf$. 
The graph language $\GLL(\cA)$ accepted by $\cA$ consists of all string-like graphs over $S$ and $A$
that are accepted by $\cA$.

Intuitively, when processing $g$, the automaton visits $V_1,\dots,V_{n+1}$ in that order. 
It visits $V_i$ in state~$q_i$,
and the subgraph $g[V_i]$ can be viewed as the storage configuration of $\cA$ 
at the current moment.
In state~$q_i$ the automaton reads the label $\alpha_i \in \Ae$ of the $\Ae$-edges from $V_i$ to $V_{i+1}$,
and uses an \mbox{$\alpha_i$-transition} $(q_i,\alpha_i,\theta_i,q_{i+1})$ 
to move to $V_{i+1}$ in state $q_{i+1}$, 
changing its storage configuration to $g[V_{i+1}]$, 
provided that the change is allowed by the instruction $\theta_i$, i.e., provided that the pair graph
$\lambda_{\Ae,\nu}(g[V_i\cup V_{i+1}])$ satisfies the formula $\theta_i$. 
The automaton starts at $V_1$ in an initial state and with storage configuration $g[V_1]$,
which is the initial storage configuration $g_\init$ of $S$. 
It accepts $g$ when it arrives at $V_{n+1}$ in a final state. 
When viewed as an acceptor of $\GLL(\cA)$ as above, the automaton~$\cA$ will also be called an 
\emph{MSO graph $S$-automaton}. 
The similarity of these automata to the graph acceptors of~\cite{tho91} will be discussed at the end of this section.

Let $S=(\phi_\rc,g_\init,\Theta)$ be an MSO graph storage type. 
A set of string-like graphs $L\subseteq \cG[S,A]$ is \emph{$S$-recognizable}
if $L=\GLL(\cA)$ for some \mbox{$S$-automaton} $\cA$ over $A$. 

Clearly, if a string-like graph $g$ 
is accepted by an $S$-automaton $\cA$, as described above, then 
the storage configurations $g[V_i]$ witness the fact that 
the sequence $\theta_1 \cdots \theta_n\in\Theta^*$ is an $S$-behaviour,
as defined in Section~\ref{sec:storagetypes}. 
For an arbitrary string-like graph $g \in \cG[S,A]$ 
such that $\com(g)=(V_1,\ldots,V_{n+1})$ for some $n \in \nat$, we define 
the set of \emph{\mbox{$S$-behaviours} on $g$}, denoted by $\cB(S,g)$, 
to be the set of all strings $\theta_1 \cdots \theta_n \in \Theta^* $ such that
$\lambda_{\Ae,\nu}(g[V_i \cup V_{i+1}]) \models \theta_i$ for every $i \in [n]$.
Thus, $\cB(S,g)\subseteq\cB(S)$.
It follows immediately from the exclusiveness of $\Theta$ that 
$\cB(S,g)$ is either a singleton or empty; and as observed above, 
it is nonempty if $g$ is accepted by an $S$-automaton. 
In other words, a string-like graph that is accepted by an $S$-automaton represents 
a unique $S$-behaviour. The next lemma is a straightforward characterization 
of the $S$-recognizable graph languages.

\begin{lemma}\rm \label{lm:decomposition} 
A graph language $L\subseteq \cG[S,A]$ 
is $S$-recognizable if and only if there exists a regular language 
$R \subseteq (\Ae \times \Theta)^*$ such that
      \begin{align*}
        L = \{g \in \cG[S,A] \mid \ & \text{there exist $n\in\nat$, $\alpha_1,\ldots,\alpha_n\in \Ae$, and $\theta_1,\ldots,\theta_n \in \Theta$}\\
        & \text{such that $\tr(g) = \alpha_1 \cdots \alpha_n$, \,$\theta_1 \cdots \theta_n \in \cB(S,g)$, and} \\
        & (\alpha_1,\theta_1) \cdots (\alpha_n,\theta_n) \in R\} \enspace.
      \end{align*}
\end{lemma}
\begin{proof} 
The proof is similar to the one of Lemma~\ref{lm:decomposition2}. 
For every $S$-automaton $\cA = (Q,Q_\init,Q_\rf,T)$ over $A$
we construct the finite-state automaton $\cA' = (Q,Q_\init,Q_\rf,T')$ over $\Ae\times\Theta$ 
as in the proof of Lemma~\ref{lm:decomposition2}, i.e., 
\[
T'=\{(q,(\alpha,\theta),q')\mid (q,\alpha,\theta,q') \in T\}.
\]
It follows directly from the definitions 
of $\GLL(\cA)$, $\cB(S,g)$, and $\LL(\cA')$, that $L=\GLL(\cA)$ and $R=\LL(\cA')$ satisfy the requirements. 
Since the transformation of $\cA$ into $\cA'$ is a bijection between $S$-automata over $A$ and 
finite-state automata over $\Ae\times\Theta$, this proves the lemma. 
\end{proof}

\begin{example}\rm \label{ex:STACK-automaton}
We continue Example~\ref{ex:stack} (of the MSO graph storage type STACK) and 
consider the STACK-automaton $\cA = (Q,Q_\init,Q_\rf,T)$ over $A=\{0,1\}$ 
that is obtained from the Stack-automaton $\cA$ of Example~\ref{ex:stackstring} 
by changing in every transition the instruction $\theta$ into $\psi_\theta$. 
Due to the isomorphism of the storage types Stack and STACK, 
the (string) language $\LL(\cA)$ accepted by~$\cA$ is still 
$\{ww^\mathrm{R}w\mid w\in A^+\}$. 
It should be clear that for every graph $g$ in the graph language $\GLL(\cA)$ accepted by $\cA$ 
there is a nonempty string $w$ over $A$ such that $\tr(g)=ww^\mathrm{R}ew$.
As an example, a graph $g\in \GLL(\cA)$ such that $\tr(g)=0110e01$ 
is displayed in Figure~\ref{fig:comput}.\footnote{The reader might think of a snake 
that has eaten eight elephants (see~\cite{des46} for the case of one elephant).}
\begin{figure}[t]
  \begin{center}
    \includegraphics[scale=0.35,trim={0 7cm 0 3cm},clip]{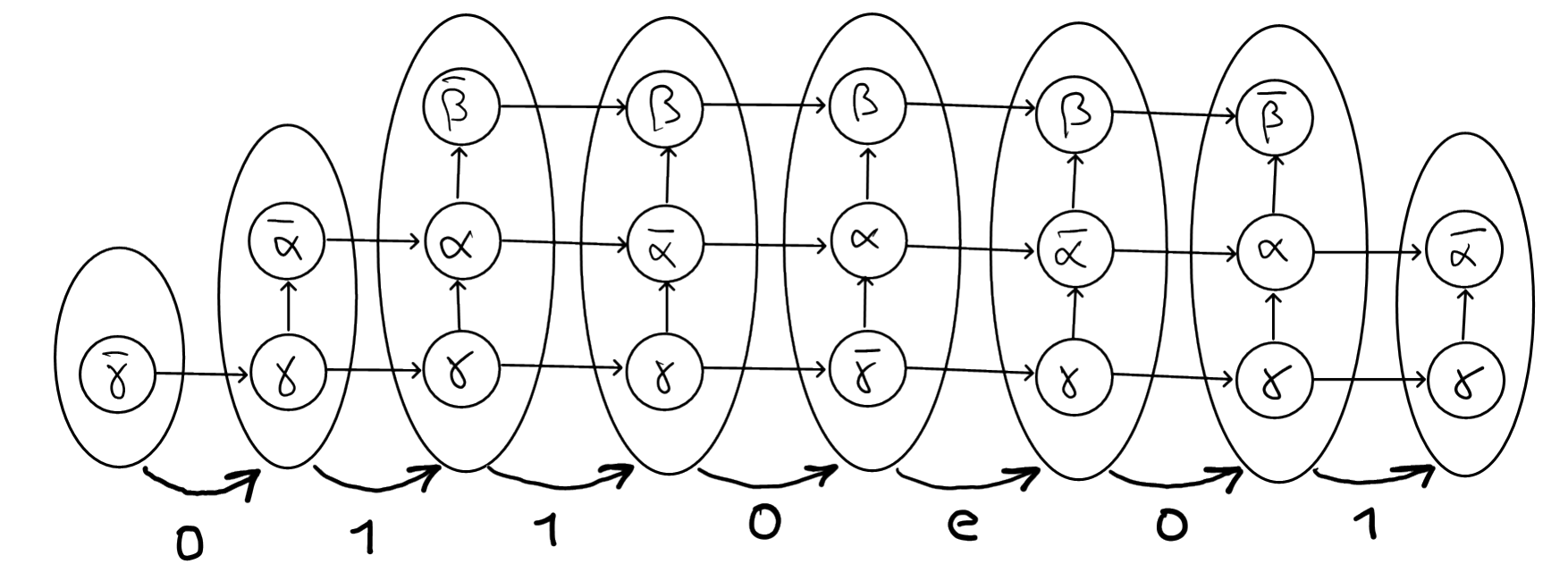}
    \end{center}
\caption{\label{fig:comput} A string-like graph $g\in \GLL(\cA)$ such that $\tr(g)=0110e01$.
The vertical edges have label $*$. The straight horizontal edges have label $d$.
As in Figure~\ref{fig:pair-graph}, the components of $g$ are represented by ovals. An $\Ae$-edge 
from one oval to another symbolizes all edges with that label from each node of the one component 
to each node of the other.} 
\end{figure}
The (unique) STACK-behaviour $b\in \cB(\mathrm{STACK},g)$ is  
\[
b = \push(\alpha);\push(\beta);\movedown(\beta);\movedown(\alpha);
\moveup(\gamma);\moveup(\alpha);\pop(\beta) 
\]
where we wrote the formulas $\psi_\theta$ as $\theta$, and separated them by semicolons.
Thus, $g$ represents both the string $0110e01$ and the behaviour $b$.
Intuitively, the MSO graph STACK-automaton $\cA$ accepts $g$ because it can check that, 
as an acceptor of $\LL(\cA)$, it has a run on input $011001$  
with the storage behaviour $b$. 
\qed
\end{example}

Intuitively, one would expect that a string $w$ over $A$ is accepted by $\cA$ 
if and only if there is a $w$-like graph that is accepted by $\cA$.
This is shown in the next lemma. 
Recall from Section~\ref{sec:string-like-graphs} that 
a string-like graph $g$ is $w$-like if $\mu_e(\tr(g))=w$,
and that the set of all $w$-like graphs is denoted $\cG[S,w]$. 

\begin{lemma}\rm \label{lm:twolang}
Let $S$ be an MSO graph storage type over $(\Sigma,\Gamma)$. 
For every $S$-automaton $\cA$ over $A$, 
$\LL(\cA)=\{w\in A^*\mid \exists \,g\in \cG[S,w]: g\in \GLL(\cA) \}$.
\end{lemma}

\begin{proof}
Let $S=(\phi_\rc,g_\init,\Theta)$ and $\cA=(Q,Q_\init,Q_\rf,T)$.
We have to show that $\LL(\cA)=\LL'(\cA)$, where 
$\LL'(\cA)=\{w\in A^*\mid \exists \,g\in \cG[S,w]: g\in \GLL(\cA) \}$.
Let $R$ be the regular language defined in the proof of both 
Lemma~\ref{lm:decomposition} and Lemma~\ref{lm:decomposition2}.
Then, by the proofs of these two lemmas,
      \begin{align*}
        \GLL(\cA) = \{g \in \cG[S,A] \mid \ & \text{there exist $n\in\nat$, $\alpha_1,\ldots,\alpha_n\in \Ae$, and $\theta_1,\ldots,\theta_n \in \Theta$}\\
        & \text{such that $\tr(g) = \alpha_1 \cdots \alpha_n$, \,$\theta_1 \cdots \theta_n \in \cB(S,g)$, and} \\
        & (\alpha_1,\theta_1) \cdots (\alpha_n,\theta_n) \in R\} 
      \end{align*}
and      \begin{align*}
        \LL(\cA) = \{w \in A^* \mid \ & \text{there exist } n\in\nat, 
             \alpha_1,\ldots,\alpha_n\in \Ae, \text{ and } \theta_1,\ldots,\theta_n \in \Theta \\
        & \text{such that } \mu_e(\alpha_1 \cdots \alpha_n)=w, \,\theta_1 \cdots \theta_n \in \cB(S), 
                \text{ and }  \\
        & (\alpha_1,\theta_1) \cdots (\alpha_n,\theta_n) \in R \} \enspace.
         \end{align*}
Since $\LL'(\cA)= \{w\in A^*\mid \text{ there exists } g\in \GLL(\cA) \text{ such that } \mu_e(\tr(g))=w\}$,
equality of $\LL(\cA)$ and $\LL'(\cA)$ is now an immediate consequence of 
the following statement.

\smallskip
{\bf Statement.}
For every $n\in\nat$, $\alpha_1,\dots,\alpha_n\in\Ae$, and $\theta_1,\dots,\theta_n\in\Theta$, 
the following two conditions are equivalent:
\begin{itemize}
\item[($1$)] there exists $g\in\cG[S,A]$ such that $\tr(g)=\alpha_1\cdots\alpha_n$ 
and $\theta_1\cdots\theta_n \in\cB(S,g)$; 
\item[($2$)] $\theta_1\cdots\theta_n \in\cB(S)$.
\end{itemize}
Note that, by definition of $\cB(S)$, (2) is equivalent to the existence of graphs $g_1,\dots,g_{n+1}\in \LL(\phi_\rc)$ such that $g_1=g_\init$ and $(g_i,g_{i+1})\in \rel(\LL(\theta_i))$ for every $i\in[n]$. From this the equivalence of (1) and (2) should be clear.
\end{proof}

\begin{example}\rm \label{ex:twolang-triv} 
Let $S=\TRIV= (\phi_\rc,g_\init,\{\theta\})$ over $(\{*\},\emptyset)$ be the MSO graph storage type from Example~\ref{ex:MSO-graph-storage-triv} and let $A$ an alphabet.
Clearly, $\cG[S,A] = \edgr(Ae^*)$. 
Let $\cA=(Q,Q_\init,Q_\rf,T)$ be an $S$-automaton, and consider the finite-state automaton
$\cA'=(Q,Q_\init,Q_\rf,T')$ over the alphabet $\Ae$
such that $T'= \{(q,\alpha,q')\mid (q,\alpha,\theta,q')\in T\}$. 
Obviously, $\LL(\cA)=\mu_e(\LL(\cA'))$. Moreover, it 
is easy to see from the definitions that $\GLL(\cA)=\edgr(\LL(\cA'))$.
An equivalent way of expressing Lemma~\ref{lm:twolang} is to say that $\LL(\cA)=\mu_e(\tr(\GLL(\cA)))$.
Hence the above illustrates that lemma, because $\tr(\edgr(\tau))=\tau$ for every $\tau\in\Ae^*$.
\qed
\end{example}

In~\cite{tho91} finite-state graph acceptors are introduced that recognize graphs by tilings.
Roughly speaking, such a graph acceptor $\cA$ consists of a finite set of states and a finite set of transitions (or tiles), which are graphs of which the nodes are additionally labeled by states.
Roughly speaking, $\cA$ accepts a graph $g$ if the nodes of $g$ can be additionally labeled by states
such that, in the resulting graph, every node belongs to an induced subgraph that 
is isomorphic to a tile and contains all edges incident with the node. 
This easily implies, as stated in~\cite[Theorem~2.4]{tho91}, that the graph language accepted by $\cA$ is 
MSO-definable. Moreover, as discussed in~\cite[Theorem~3.1]{tho91}, this framework captures the known 
finite-state automata on strings and trees. Suppose now that we generalize these 
tiling graph acceptors such that the set of tiles is allowed to be an arbitrary MSO-definable set of graphs. 
Then, on string-like graphs, every MSO graph $S$-automaton $\cA$ can be simulated by such a generalized 
tiling graph acceptor $\cA'$ with the same set of states. That should be clear from 
the intuitive description of the way in which $\cA$ accepts $\GLL(\cA)$. 
Indeed (cf. the intermediate transitions in the proof of~\cite[Theorem~3.1]{tho91}), 
if $(q_1,\alpha,\theta_1,q_2)$ and $(q_2,\beta,\theta_2,q_3)$ 
are transitions of $\cA$ and $g$ is a string-like graph with three components $V_1,V_2,V_3$ such that 
$\tr(g)= \alpha\beta$ and $\lambda_{\Ae,\nu}(g[V_i\cup V_{i+1}])\in \LL(\theta_i)$ for $i\in\{1,2\}$, then 
$\cA'$ has a tile that is obtained from $g$ by additionally labeling the nodes of $V_i$ 
by $q_i$, for $i\in\{1,2,3\}$. Additional tiles are needed to handle the initial and final states 
of $\cA$, cf. the proof of~\cite[Theorem~3.1]{tho91}.
This simulation shows, by Observation~\ref{ob:stringlike} and~\cite[Theorem~2.4]{tho91}
(which clearly still holds) that $\GLL(\cA)$ is MSO-definable. That will, of course, be a consequence of 
our first BET-theorem for $S$-automata (Theorem~\ref{thm:main}).

\section{A Logic for String-Like Graphs}
\label{sec:logic-on-string-graphs}

For every MSO graph storage type $S$ we want to design a logic of which the formulas
define the graph languages accepted by MSO graph \mbox{$S$-automata}, and hence
also the (string) languages accepted by \mbox{$S$-automata}, as expressed in Lemma~\ref{lm:twolang}.   
Each formula of the logic has two levels, 
an outer level that only considers the ``string aspect'' of the string-like graph,
and an inner level that only considers the ``storage behaviour aspect'' of the graph. 

Let $S=(\phi_\rc,g_\init,\Theta)$ be an MSO graph storage type over $(\Sigma,\Gamma)$, 
and let $A$ be an alphabet.

The \emph{set of MSO-logic formulas over $S$ and $A$}, denoted by $\MSOL(S,A)$, 
is the smallest set $M$ of expressions such that
\begin{enumerate}
\item[($1$)] for every $\alpha\in \Ae$, the set $M$ contains $\edge_\alpha(x,y)$ and $x\eeuro X$,
 \item[($2$)] for every $\theta\in\Theta$,
the set $M$ contains $\nnext(\theta,x,y)$, 
\item[($3$)] if $\phi,\phi' \in M$, then the set $M$ contains $(\neg \phi)$, $(\phi \vee \phi')$, $(\exists x. \phi)$, and $(\exists X. \phi)$.
  \end{enumerate}
  
For a formula $\phi\in \MSOL(S,A)$, the subformulas $\nnext(\theta,x,y)$
of $\phi$ form its inner level that refers to the storage behaviour aspect, 
whereas the remainder of~$\phi$ forms its outer level that refers to the string aspect.
We define the set $\Free(\varphi)$ of free variables of $\varphi$ in the usual way; in particular,
$\Free(\nnext(\theta,x,y))=\{x,y\}$.

Intuitively, this logic is interpreted for a string-like graph $g$ as follows.

(1)  The meaning of $\edge_\alpha(x,y)$ is the standard one. The meaning of $x\eeuro X$ 
is a variant of the meaning of $x\in X$: either $x\in X$ or
there is an element $y$ of~$X$ such that 
$x$ and $y$ are in the same component of $g$.  

(2) The meaning of $\nnext(\theta,x,y)$ is that $x$ and $y$ belong to consecutive components,
and that the subgraph of $g$ induced by the union of these components (with the $\Ae$-edges replaced by $\nu$-edges) satisfies $\theta$.

(3) The meaning of these formulas is standard. 

\noindent
Formally, let $g \in \cG[S,A]$ be a  string-like graph 
and let $\com(g)=(V_1,\dots,V_{n+1})$ for some $n\in\nat$.
Moreover, let $\phi \in \MSOL(S,A)$
and let $\cV \supseteq \Free(\phi)$. Finally, let $\rho$ be a $\cV$-valuation on $g$. We define the \emph{models relationship} $(g,\rho) \models \phi$ by induction on the structure of $\phi$ as follows.

\begin{itemize}
\item Let $\phi = \edge_\alpha(x,y)$. Then $(g,\rho) \models \phi$ if $(\rho(x),\alpha,\rho(y)) \in E_g$,
as defined for $\MSOL(\Sigma,\Gamma\cup\Ae)$. 

\item Let $\phi = (x\eeuro X)$. Then $(g,\rho) \models \phi$ if 
  $(g,\rho) \models \exists y. (y\in X \wedge \eq(x,y))$.
(Recall the definition of $\eq(x,y)$ before Observation~\ref{ob:stringlike}.)

\item Let $\phi = \nnext(\theta,x,y)$ for some $\theta \in \Theta$. Then $(g,\rho) \models \phi$ if there exists $i\in[n]$ such that $\rho(x) \in V_i$, $\rho(y)\in V_{i+1}$, and $\lambda_{\Ae,\nu}(g[V_i\cup V_{i+1}])\models \theta$.

\item Let $\phi$ be formed according to the third item of the definition of $\MSOL(S,A)$ (i.e., $\phi$ contains at least one occurrence of $\neg$, $\vee$, or $\exists$). Then $(g,\rho) \models \phi$ is defined as for $\MSOL(\Sigma,\Gamma\cup\Ae)$.
\end{itemize}
As in the case of $\MSOL(\Sigma,\Gamma)$, we identify $(g,\emptyset)$ with $g$.

We know from Lemma~\ref{lm:decomposition} that for every $S$-recognizable graph language 
$L \subseteq \cG[S,A]$, $\cB(S,g)\neq \emptyset$ for every $g\in L$.
But, obviously, there are formulas $\phi\in\MSOL(S,A)$ such that 
the graph language $\{g\in\cG[S,A]\mid g\models \phi\}$ does not satisfy this requirement. 
Hence, to obtain a logic equivalent to \mbox{$S$-recognizability}, we need to restrict $\MSOL(S,A)$ to 
formulas that do satisfy the requirement, as follows.
Let $\beh$ be the following closed formula in $\MSOL(S,A)$:
\begin{align*}
\beh = & \ \forall x,y. \bigwedge_{\alpha \in \Ae} \Big(\edge_\alpha(x,y) \rightarrow
          \bigvee_{\theta \in \Theta} \nnext(\theta,x,y)\Big).
\end{align*}

\begin{observation}\label{ob:beta}\rm 
For every $g \in \cG[S,A]$, we have
$g \models \beh \text{ if and only if }\cB(S,g) \neq \emptyset$.
\qed 
\end{observation}

A set of string-like graphs $L \subseteq \cG[S,A]$ is \emph{$\MSOL(S,A)$-definable} if 
there exists a closed  formula $\phi \in \MSOL(S,A)$ such that 
\[
L = \{g\in \cG[S,A] \mid g\models \beh \wedge \phi\} \enspace.
\]
Similarly, a string language $L \subseteq A^*$ is \emph{$\MSOL(S,A)$-definable} 
if there exists a closed formula $\phi \in \MSOL(S,A)$ such that 
\[
L = \{w \in A^* \mid \exists \,g\in\cG[S,w]: \;g\models \beh \wedge \phi\}
\]
(or in words, $L$ consists of all strings $w$ for which there exists a $w$-like graph that satisfies
the formula $\beh \wedge \phi$). 
An equivalent formulation is that $L \subseteq A^*$ is $\MSOL(S,A)$-definable if there exists an 
$\MSOL(S,A)$-definable graph language $G \subseteq \cG[S,A]$ such that 
$L = \{w \in A^* \mid \exists \,g\in\cG[S,w]: 
\;g\in G\} = \mu_e(\tr(G))$. 

\begin{example}\rm \label{ex:formula}
A formula that defines the graph language $\GLL(\cA)$ accepted by 
the STACK-automaton $\cA= (Q,Q_\init,Q_\rf,T)$ of Example~\ref{ex:STACK-automaton}, 
has a structure that is familiar from the proof of the classical BET-theorem. It is the formula 
$\beh\wedge \phi_\cA$ such that
\begin{align*}
\phi_\cA = &\ \exists X_1,X_2,X_3,X_4. (\,\mathrm{part}(X_1,X_2,X_3,X_4) \\
& \wedge \forall x.(\first(x)\to x\eeuro X_1) \\
& \wedge \forall x.(\last(x)\to x\eeuro X_4) \\
& \wedge \phi_0 \wedge \phi_1 \wedge \phi_e) 
\end{align*}
with the following subformulas.
First,   
\begin{align*}
\mathrm{part}(X_1,X_2,X_3,X_4) = &\ \forall x. 
\bigvee_{i\in[4]} (x\eeuro X_i \wedge \neg\bigvee_{j\in[4]\setminus\{i\}} x\eeuro X_j)
\end{align*}
which expresses that $X_1,\dots,X_4$ define a partition of the set of nodes of the string-like graph into unions of components.
Second, 
\begin{align*}
\first(x) &= (\neg \exists y. \edge_\Ae(y,x)) \\
\last(x) &= (\neg \exists y. \edge_\Ae(x,y))
\end{align*}
which express that $x$ is in the first/last component of the graph, respectively.
And third, the formulas $\phi_0$, $\phi_1$, and $\phi_e$ that express the $0$-transitions, $1$-transitions, and $e$-transitions of $\cA$, respectively (see Example~\ref{ex:stackstring}).
\begin{align*}
\phi_e &= \forall x,y. (\edge_e(x,y) \to (x \eeuro X_2 \wedge \nnext(\psi_{\moveup(\gamma)},x,y)\wedge y\eeuro X_3)) \\
\phi_0  &= \forall x,y. (\edge_0(x,y) \to \\
& \Big((x \eeuro X_1 \wedge \nnext(\psi_{\push(\alpha)},x,y)\wedge y\eeuro X_1) \\
& \vee ((x \eeuro X_1 \vee x \eeuro X_2) \wedge \nnext(\psi_{\movedown(\alpha)},x,y)\wedge y\eeuro X_2) \\
& \vee (x \eeuro X_3 \wedge \nnext(\psi_{\moveup(\alpha)},x,y)\wedge y\eeuro X_3) \\
& \vee (x \eeuro X_3 \wedge \nnext(\psi_{\pop(\alpha)},x,y)\wedge y\eeuro X_4)\Big)) 
\end{align*}
The formula $\phi_1$ is obtained from $\phi_0$ by changing $0$ in $1$, and $\alpha$ in $\beta$, everywhere.

Note that, informally speaking, $\GLL(\cA)$ is also defined by 
the simpler formula~$\phi_\cA$
because that formula implies the formula $\beh$ 
for every string-like graph over $S$ and $A$.
\qed
\end{example}

The next observation shows that $\MSOL(S,A)$ can be viewed as a subset of $\MSOL(\Sigma,\Gamma\cup\Ae)$.
It implies that every $\MSOL(S,A)$-definable graph language $L\subseteq \cG[S,A]$ is MSO-definable.
That, on its turn, implies that if $L \subseteq A^*$ is $\MSOL(S,A)$-definable, then there exists an 
MSO-definable graph language $G \subseteq \cG[S,A]$ such that $L = \mu_e(\tr(G))$, cf. the discussion on the scheme of BET-theorems in the Introduction.

\begin{observation}\rm
\label{obs:msolsubset}
For every set of string-like graphs $L\subseteq \cG[S,A]$, if $L$ is
$\MSOL(S,A)$-definable, then $L$~is $\MSOL(\Sigma,\Gamma\cup \Ae)$-definable.
\end{observation}

\begin{proof}
Since the set $\cG[S,A]$ is MSO-definable by Observation~\ref{ob:stringlike},
it suffices to show that for every formula $\phi\in\MSOL(S,A)$ there is a formula $\phi'\in\MSOL(\Sigma,\Gamma\cup\Ae)$ such that, 
for every $g\in \cG[S,A]$ and every valuation $\rho$ on~$g$,
$(g,\rho)\models \phi$ if and only if $(g,\rho)\models \phi'$. 
The translation of $\phi$ into $\phi'$ is straightforward. 
Let the following formula express that $x$ is in the equivalence class $X$ 
of the equivalence relation $\equiv_\Ae$, i.e., 
that $X$ is the component to which $x$ belongs:
\[
\ec(x,X) = \forall y.(y\in X \leftrightarrow \eq(x,y)) \enspace.
\]
We define $\phi'$ to be the formula obtained from $\phi$ by 
the following replacements of subformulas. 
\begin{itemize}
\item Every $x\eeuro X$ is replaced by $\exists y. (y\in X \wedge \eq(x,y))$.
\item Every $\nnext(\theta,x,y)$ is replaced by 
\[
\edge_\Ae(x,y)\wedge \forall X,Y,Z.((\ec(x,X)\wedge \ec(y,Y) \wedge \mathrm{union}(X,Y,Z))\to 
\tilde{\theta}\,|_Z)
\]
where $\mathrm{union}(X,Y,Z)$ expresses that $Z$ is the union of $X$ and $Y$, and 
$\tilde{\theta}$ is obtained from $\theta$ by changing every subformula $\edge_\nu(x,y)$ into $\edge_\Ae(x,y)$.
\end{itemize}
Note that every subformula $\edge_\alpha(x,y)$ remains unchanged. 
\end{proof}

\begin{example}\rm \label{ex:formula-triv} 
Let $S=\TRIV= (\phi_\rc,g_\init,\{\theta\})$ over $(\{*\},\emptyset)$ be the MSO graph storage type 
from Example~\ref{ex:MSO-graph-storage-triv} and let $A$ an alphabet. 
As already observed in Example~\ref{ex:twolang-triv}, $\cG[S,A] = \edgr(Ae^*)$.
In the next paragraph we show that a set of string-like graphs $L\subseteq \cG[S,A]$ is 
$\MSOL(S,A)$-definable if and only if it is MSO-definable (i.e., $\MSOL(\{*\},\Ae)$-definable).

Obviously, for every graph $g\in \cG[S,A] = \edgr(Ae^*)$,
the formula $\nnext(\theta,x,y)$
is equivalent to $\edge_\Ae(x,y)$, for all nodes $x$ and $y$ of $g$. Moreover, 
since $\equiv_\Ae$ is the identity on $V_g$, the formula
$x \eeuro X$ is equivalent to the formula $x\in X$. 
This shows that if $L$ is $\MSOL(S,A)$-definable, then it is $\MSOL(\{*\},\Ae)$-definable
(which is in accordance with Observation~\ref{obs:msolsubset}). 
On the other hand, the formula $\lab_*(x)$ is true for $g$.
Since, by the above, the formula $\beh$ is true for $g$, 
this shows that if $L$ is $\MSOL(\{*\},\Ae)$-definable, then it is $\MSOL(S,A)$-definable.

Consequently, a string language $L\subseteq A^*$ is $\MSOL(\TRIV,A)$-definable if and only if 
there exists an $\MSOL(\{*\},\Ae)$-definable graph language $G\subseteq \edgr(Ae^*)$ such that 
$L= \mu_e(\tr(G))$. Since $\edgr$ is a bijection between $\Ae^*$ and $\edgr(Ae^*)$ with inverse $\tr$
(cf. Example~\ref{ex:twolang-triv}), 
that is if and only if there exists a string language $L'\subseteq \Ae^*$ such that 
$L= \mu_e(L')$ and $\edgr(L')$ is $\MSOL(\{*\},\Ae)$-definable.
\qed
\end{example}

\section{The BET-Theorems for $S$-Automata}
\label{sec:bet}

In this section we prove our B\"uchi-Elgot-Trakhtenbrot theorems: the equivalence of $S$-recognizability and 
$\MSOL(S,A)$-definability, for graph languages \mbox{$L\subseteq \cG[S,A]$}
and for string languages $L\subseteq A^*$, where $S=(\phi_\rc,g_\init,\Theta)$ 
is an MSO graph storage type.
Since we have characterized the $S$-recognizable graph languages in terms of regular languages 
over $\Ae \times \Theta$ in Lemma~\ref{lm:decomposition}, and since 
a language $L$ over $\Ae \times \Theta$ is regular if and only if $\edgr(L)$ is
$\MSOL(\{*\}, \Ae \times \Theta)$-definable
by the classical BET-theorem (Proposition~\ref{pro:MSO-string-graph}),
it now suffices to translate $\MSOL(\{*\}, \Ae \times \Theta)$ into $\MSOL(S,A)$, 
and back, which we do in the next two lemmas.

\begin{lemma}\label{lm:string-edge-graph-rep}\rm 
For each closed formula $\phi\in\MSOL(\{*\}, \Ae \times \Theta)$ there exists a closed formula 
$\phi' \in \MSOL(S,A)$ such that for all $n\in\nat$, $\alpha_1,\ldots,\alpha_n \in \Ae$, 
$\theta_1,\ldots,\theta_n \in \Theta$, and 
$g \in \cG[S,A]$, if $\tr(g)=\alpha_1\cdots\alpha_n$ and 
$\theta_1\cdots\theta_n \in \cB(S,g)$, then
\[
\edgr((\alpha_1,\theta_1) \cdots (\alpha_n,\theta_n)) \models \phi \text{ if and only if } g \models \phi'.
\]
\end{lemma}

\begin{proof}
We may assume that the formulas in $\MSOL(\{*\}, \Ae \times \Theta)$ do not have 
atomic subformulas $\lab_*(x)$, which are always true.
For every formula $\phi\in\MSOL(\{*\}, \Ae \times \Theta)$ we define 
$\phi' \in \MSOL(S,A)$ to be the formula obtained from~$\phi$ by changing every atomic subformula 
$\edge_{(\alpha,\theta)}(x,y)$ into the formula $\edge_\alpha(x,y)\wedge \nnext(\theta,x,y)$ 
and every atomic subformula $x\in X$ into $x\eeuro X$. 
Note that $\Free(\phi')=\Free(\phi)$. 

Now let $g \in \cG[S,A]$, $\tr(g)=\alpha_1\cdots \alpha_n$, and
$\theta_1\cdots\theta_n \in \cB(S,g)$. 
Let $\com(g)=(V_1,\dots,V_{n+1})$, and 
let $\edgr((\alpha_1,\theta_1) \cdots (\alpha_n,\theta_n))$
be the graph $(V,E,\nlab)$ with $V=[n+1]$ and $E=\{(i,(\alpha_i,\theta_i),i+1)\}$.

The lemma follows from the following statement.

\smallskip
{\bf Statement}. Let $\phi\in\MSOL(\{*\}, \Ae \times \Theta)$, and 
let $\cV$ be a set of variables such that $\Free(\phi)\subseteq \cV$. 
Let $\rho$ be a $\cV$-valuation on $\edgr((\alpha_1,\theta_1) \cdots (\alpha_n,\theta_n))$
and let $\rho'$ be a $\cV$-valuation on $g$ such that (1)~$\rho'(x)\in V_{\rho(x)}$, 
for every node variable $x\in\cV$,
and (2)~$\rho'(X)\cap V_i \neq\emptyset$ if and only if $i\in \rho(X)$, 
for every $i\in[n+1]$ and node-set variable $X\in\cV$.
Then 
\[
(\edgr((\alpha_1,\theta_1) \cdots (\alpha_n,\theta_n)),\rho) \models \phi 
\text{ if and only if } (g,\rho') \models \phi' \enspace.
\]

\emph{Proof of Statement.} We prove this statement by induction on the structure of~$\phi$.  It follows from $\theta_1\cdots\theta_n \in \cB(S,g)$ that 

(a) for every $i \in [n]$, $\lambda_{Ae,\nu}(g[V_i \cup V_{i+1}]) \models \theta_i$ and, as a consequence of the exclusiveness of~$\Theta$, $\lambda_{Ae,\nu}(g[V_i \cup V_{i+1}]) \not\models \theta$ for every $\theta \in \Theta \setminus \{\theta_i\}$. 

\noindent
Moreover, since $g \in \cG[S,A]$ and $\tr(g)=\alpha_1\cdots \alpha_n$, we have that 

(b) for every $i \in [n]$, $u \in V_i$, and $v \in V_{i+1}$, the edge $(u,\alpha_i,v) \in E_g$ is the only $Ae$-edge between $u$ and~$v$, and 

(c) for every $u,v \in V_g$, $u \equiv_\Ae v$ if and only if 
there exists $i \in [n+1]$ such that $u\in V_i$ and $v\in V_i$.

Let $\phi=\edge_{(\alpha,\theta)}(x,y)$. 
\begin{align*}
  & (\edgr((\alpha_1,\theta_1) \cdots (\alpha_n,\theta_n)),\rho) \models \edge_{(\alpha,\theta)}(x,y)\\
  \Leftrightarrow  \ & \exists i \in [n]: \rho(i)=i, \rho(y)=i+1, \alpha=\alpha_i, \theta_i=\theta\\
   \Leftrightarrow \ & \exists i \in [n]: \rho'(i)\in V_i, \rho'(y)\in V_{i+1}, \alpha=\alpha_i, \theta_i=\theta\\
    \Leftrightarrow  \ & (g,\rho') \models \edge_\alpha(x,y) \wedge \nnext(\theta,x,y)
  \end{align*}
where the last equivalence follows from (a) and (b).

Let $\phi = (x\in X)$.
\begin{align*}
  & (\edgr((\alpha_1,\theta_1) \cdots (\alpha_n,\theta_n)),\rho) \models (x\in X)\\
  \Leftrightarrow  \ & \exists i \in [n+1]: (\rho(x)=i) \wedge (i \in \rho(X))\\
  \Leftrightarrow  \ & \exists i \in [n+1]: (\rho'(x)\in V_i) \wedge (V_i \cap \rho'(X) \not= \emptyset)\\
  \Leftrightarrow  \ & \exists v\in \rho'(X): \rho'(x) \equiv_\Ae v\\
  \Leftrightarrow  \ & (g,\rho') \models (x \eeuro X)
\end{align*}
where the last but one equivalence uses (c).

If $\phi$ has any other form, then the statement follows by induction. 
Note that for every valuation $\rho$ there is a valuation $\rho'$ that satisfies the requirements in the statement, and vice versa.
\end{proof}

\begin{lemma}\rm \label{lm:string-edge-graph-rep2}
For each closed formula $\phi \in \MSOL(S,A)$ there exists a closed formula
$\phi'\in\MSOL(\{*\},$ $Ae \times \Theta)$  such that 
for every $n\in\nat$, $\alpha_1,\ldots,\alpha_n \in \Ae$,
$\theta_1,\ldots,\theta_n \in \Theta$, and 
$g \in \cG[S,A]$, if $\tr(g)=\alpha_1\cdots\alpha_n$ and 
$\theta_1\cdots\theta_n \in \cB(S,g)$, then
\[
g \models \phi \text{ if and only if }  \edgr((\alpha_1,\theta_1) \cdots (\alpha_n,\theta_n)) \models \phi' .
\]
\end{lemma}

\begin{proof}
For every formula $\phi\in\MSOL(S,A)$ we define 
$\phi' \in \MSOL(\{*\}, Ae \times \Theta)$ to be the formula obtained from~$\phi$ 
by the following replacements.
\begin{compactitem}
\item Every subformula 
$\edge_\alpha(x,y)$ is replaced by $\bigvee_{\theta \in \Theta} \edge_{(\alpha,\theta)}(x,y)$.

\item Every $x\eeuro X$ is replaced by $x\in X$.

\item Every $\nnext(\theta,x,y)$ is replaced by $\bigvee_{\alpha \in Ae} \edge_{(\alpha,\theta)}(x,y)$.
\end{compactitem}
Note that $\Free(\phi')=\Free(\phi)$.

Now let $g \in \cG[S,A]$, $\tr(g)=\alpha_1\cdots \alpha_n$, and
$\theta_1\cdots\theta_n \in \cB(S,g)$. 
First, we prove the following statement by induction on the structure of $\phi$.

\smallskip
{\bf Statement.} Let $\phi\in\MSOL(S,A)$, let $\cV$ be a set of variables 
such that $\Free(\phi)\subseteq\cV$, and let $\rho$ be a $\cV$-valuation on $g$.
Then 
\[
(g,\rho) \models \phi \text{ if and only if } (g,\rho) \models \phi'',
\]
where $\phi''$ is the formula in $\MSOL(S,A)$ obtained from $\phi'$ by the transformation defined in the proof of Lemma~\ref{lm:string-edge-graph-rep}.

\emph{Proof of Statement.}
It follows from $\theta_1\cdots\theta_n \in \cB(S,g)$ and Observation~\ref{ob:beta} that 
\[
\text{(a)} \quad (g,\rho) \models \edge_\alpha(x,y)\to\bigvee_{\theta \in \Theta}\nnext(\theta,x,y) 
\]
for all $\alpha\in\Ae$ and $x,y\in\Free(\phi)$. Since $g\in\cG[S,A]$, we also have 
\[
\text{(b)} \quad (g,\rho) \models \nnext(\theta,x,y)\to \bigvee_{\alpha \in \Ae}\edge_\alpha(x,y)
\]
for all $x,y\in\Free(\phi)$.

Let $\phi = \edge_\alpha(x,y)$.  Then $\phi''= \bigvee_{\theta \in \Theta}(\edge_\alpha(x,y) \wedge \nnext(\theta,x,y))$. The statement follows from distributivity of $\wedge$ over $\vee$ and (a).

Let $\phi = \nnext(\theta,x,y)$.  Then $\phi''= \bigvee_{\alpha \in \Ae}(\edge_\alpha(x,y) \wedge \nnext(\theta,x,y))$. The statement follows from distributivity of $\wedge$ over $\vee$ and (b). 

Let $\phi=(x \eeuro X)$. Then $\phi''=\phi$ and the statement trivially holds.

If $\phi$ is of any other form, e.g., $\phi = \phi_1 \wedge \phi_2$, then the statement follows by induction, e.g., from the hypotheses that $(g,\rho) \models \phi_i$ if and only if $(g,\rho) \models \phi_i''$, for each $i \in \{1,2\}$. 
\hfill\emph{End of Proof of Statement.}

\smallskip
Let $\phi \in \MSOL(S,A)$ be a closed formula, let $\phi' \in \MSOL(\{*\},Ae\times \Theta)$ be the closed formula obtained from $\phi$ by the transformation from the beginning of this proof, and let $\phi'' \in \MSOL(S,A)$ be the closed formula obtained from~$\phi'$ by the transformation defined in the proof of Lemma~\ref{lm:string-edge-graph-rep}.  Then: 
\begin{align*}
  & g \models \phi\\
  \Leftrightarrow  \ \ & g \models \phi'' \ \text{(by the Statement)}\\
  \Leftrightarrow  \ \ & \edgr((\alpha_1,\theta_1) \cdots (\alpha_n,\theta_n)) \models \phi' \ \text{(by the proof of Lemma \ref{lm:string-edge-graph-rep})}
\end{align*}
\end{proof}

It is now straightforward to prove our BET-theorems, first for graph languages and then for string languages.

\begin{theorem}\label{thm:main}\rm 
For every MSO graph storage type $S$ and alphabet $A$, a graph language $L\subseteq \cG[S,A]$ 
is $S$-recognizable if and only if $L$ is $\MSOL(S,A)$-definable. 
\end{theorem}

      \begin{proof}  
By Lemma~\ref{lm:decomposition}, $L$ is $S$-recognizable if and only if 
there is a regular language $R \subseteq (\Ae \times \Theta)^*$ such that
      \begin{align*}
        L = \{g \in \cG[S,A] \mid \ & \text{there exist $n\in\nat$, $\alpha_1,\ldots,\alpha_n\in \Ae$, and $\theta_1,\ldots,\theta_n \in \Theta$}\\
              & \text{such that $\tr(g) = \alpha_1 \cdots \alpha_n$, \,$\theta_1 \cdots \theta_n  \in \cB(S,g)$, and} \\
        & (\alpha_1,\theta_1) \cdots (\alpha_n,\theta_n) \in R \}\enspace.
      \end{align*}
        By Proposition~\ref{pro:MSO-string-graph}, the classical BET-theorem for string languages, 
$R$ is regular if and only if there is a closed formula $\phi\in\MSOL(\{*\},\Ae\times\Theta)$ 
such that for every $w\in (\Ae\times\Theta)^*$, $w\in R$ if and only if $\edgr(w)\models \phi$.
Thus, $L$ is $S$-recognizable if and only if there is a closed formula $\phi\in\MSOL(\{*\},\Ae\times\Theta)$ 
such that 
      \begin{align*}
        L = \{g \in \cG[S,A] \mid \ & \text{there exist $n\in\nat$, $\alpha_1,\ldots,\alpha_n\in \Ae$, and $\theta_1,\ldots,\theta_n \in \Theta$}\\
        & \text{such that $\tr(g) = \alpha_1 \cdots \alpha_n$, \,$\theta_1 \cdots \theta_n \in \cB(S,g)$, and} \\        
        & \edgr((\alpha_1,\theta_1) \cdots (\alpha_n,\theta_n)) \models \phi\} \enspace.
      \end{align*}
        Due to Lemmas~\ref{lm:string-edge-graph-rep} and~\ref{lm:string-edge-graph-rep2} 
this holds if and only if there is a closed formula $\phi' \in \MSOL(S,A)$ such that  
      \begin{align*}
        L = \{g \in \cG[S,A] \mid \ & \text{there exist $n\in\nat$, $\alpha_1,\ldots,\alpha_n\in \Ae$, and $\theta_1,\ldots,\theta_n \in \Theta$}\\
        & \text{such that $\tr(g) = \alpha_1 \cdots \alpha_n$, \,$\theta_1 \cdots \theta_n \in \cB(S,g)$, and}\\
        & g \models \phi'\} \enspace,
      \end{align*}
      i.e.,
           \begin{align*}
        L = \{g \in \cG[S,A] \mid \ & \text{there exist $n\in\nat$ and $\theta_1,\ldots,\theta_n \in \Theta$}\\
        & \text{such that $\theta_1 \cdots \theta_n \in \cB(S,g)$ and}\\
        & g \models \phi'\} \enspace.
           \end{align*}
           By Observation~\ref{ob:beta} this holds if and only if there is a closed formula $\phi' \in \MSOL(S,A)$ such that
\[
L = \{g\in \cG[S,A] \mid g\models \beh \wedge \phi'\} \enspace,
\]
i.e., if and only if $L$ is $\MSOL(S,A)$-definable.
\end{proof}

Thus, by Theorem~\ref{thm:main} and Observation~\ref{obs:msolsubset}, every $S$-recognizable 
graph language is MSO definable, cf. the last paragraph of Section~\ref{sec:automaton-model}.

\begin{theorem}\rm \label{thm:main-string}
For every MSO graph storage type $S$ and alphabet $A$, a string language $L\subseteq A^*$ is $S$-recognizable if and only if $L$ is $\MSOL(S,A)$-definable.
\end{theorem}

\begin{proof}
By Lemma~\ref{lm:twolang}, $L$ is $S$-recognizable if and only if there is an $S$-recognizable graph language $G$ such that 
\[
L=\{w\in A^*\mid \exists \,g\in \cG[S,w]: g\in G \} \enspace.
\]
By definition, $L$ is $\MSOL(S,A)$-definable if and only if there is 
a closed formula $\phi \in \MSOL(S,A)$ such that 
\[
L = \{w \in A^* \mid \exists \,g\in\cG[S,w]: 
\;g\models \beh\wedge \phi\} \enspace.
\]
These statements are equivalent by Theorem~\ref{thm:main}. 
\end{proof}

Thus, for MSO graph storage types $S$, we have expressed the $S$-recognizability 
of string languages over an alphabet $A$ in terms of the special logic $\MSOL(S,A)$, which
is a subset of the standard logical language $\MSOL(\Sigma,\Gamma\cup\Ae)$ on graphs,
by Observation~\ref{obs:msolsubset}. 

In view of Examples~\ref{ex:triv}, \ref{ex:MSO-graph-storage-triv}, 
and~\ref{ex:formula-triv}, we obtain as the special case of Theorem~\ref{thm:main-string} where $S=\TRIV$, that a language $L \subseteq A^*$ is regular if and only if 
there exists a language $L'\subseteq \Ae^*$ such that $L=\mu_e(L')$ and $\edgr(L')$ is $\MSOL(\{*\},\Ae)$-definable. That is very close to, but not the same as, Proposition~\ref{pro:MSO-string-graph}.
It can however easily be checked that all our results 
(except for the closure properties after Lemma~\ref{lm:decomposition2})
are also valid when we forbid $S$-automata to have $e$-transitions, and replace $\Ae$ everywhere by $A$,
and $\mu_e$ by the identity on $A$. Then the corresponding analogue of Theorem~\ref{thm:main-string} for $S=\TRIV$ is exactly Proposition~\ref{pro:MSO-string-graph}.

To end this section we state an easy corollary of Theorem~\ref{thm:main}.

\begin{corollary}\rm \label{cor:boolean-algebra}
Let $S$ be an MSO graph storage type
and $A$ an alphabet. Then the class of all $S$-recognizable graph languages 
$L\subseteq \cG[S,A]$ is a Boolean algebra.
\end{corollary}

\begin{proof}
For a closed formula $\phi \in \MSOL(S,A)$, 
we denote the graph language $\{g\in \cG[S,A] \mid g\models \beh \wedge \phi\}$ by $\LL(\phi)$.
The largest $S$-recognizable graph language is $\LL(\true)$, 
which is the set of all string-like graphs over $S$ and $A$ that satisfy $\beh$, 
and the smallest is $\LL(\false)=\emptyset$.
Moreover, for closed formulas $\phi_1,\phi_2 \in \MSOL(S,A)$ we have 
$\LL(\phi_1)\cup \LL(\phi_2)=\LL(\phi_1\vee\phi_2)$ and 
$\LL(\true)\setminus \LL(\phi_1)=\LL(\neg\,\phi_1)$.
\end{proof}

\section{MSO-Expressible Storage Types}
\label{sec:MSO-expres}

We will say that a storage type is \emph{MSO-expressible} if it is isomorphic  
to some MSO graph storage type.
If storage type $S'$ is isomorphic to MSO graph storage type $S$, then they are language equivalent, 
i.e., $S'\text{-REC} = S\text{-REC}$,
and hence Theorem~\ref{thm:main-string} can be viewed as a BET-theorem for $S'$ and the logic $\MSOL(S,A)$.
Thus, we wish to know which storage types are MSO-expressible.
It is not difficult to prove that all well-known concrete storage types (such as the nested stack, the queue, 
the Turing tape, etc.) are MSO-expressible. 
For the storage type Stack of Example~\ref{ex:stackstring} we have shown that in Example~\ref{ex:stack}.
In the following three subsections we prove  three general results.

First, if $S=(C,c_\init,\Theta,m)$ is a storage type
such that (as for MSO graph storage types) $C$ is an
MSO-definable set of graphs, and moreover, for every $\theta\in\Theta$,
the storage transformation $m(\theta)$ is an MSO graph transduction 
(in the sense of~\cite[Chapter~7]{coueng12}), then $S$ is MSO-expressible. 

Second, if $S=(C,c_\init,\{\theta_1,\dots,\theta_n\},m)$ is a storage type
such that the set $C$, together with the binary relations 
$m(\theta_1),\dots,m(\theta_n)$, is an automatic structure 
(in the sense of~\cite{khomin,kus09}), then $S$ is \mbox{MSO-expressible}.

Third, if the storage type $S$ is MSO-expressible, 
then so is the storage type $\rP(S)$ 
of which the storage configurations are pushdowns of storage configurations of $S$
(see, e.g., \cite{gre70,eng86,engvog86,eng91c}).

\subsection{MSO Graph Transductions}
\label{sec:msogratra}

Recall from Section~\ref{sec:pairgraphs} that 
a relation $R\subseteq \cG_{\Sigma,\Gamma}\times \cG_{\Sigma,\Gamma}$ is MSO-expressible 
if there are an alphabet~$\Delta$ and an MSO-definable set of pair graphs 
$H\subseteq \cP\cG_{\Sigma,\Gamma\cup\Delta}$ such that $\rel(H)=R$.
To prove the first general result mentioned in the introduction of this section,
it suffices, by Observation~\ref{ob:isom}, to show that all MSO-definable graph transductions are MSO-expressible.

To define MSO graph transductions
we recall basic notions from~\cite[Section~7.1]{coueng12} and apply appropriate (small) modifications 
to restrict them to graphs (as in~\cite{bloeng00}).
As in Section~\ref{sec:graphs-MSO}, we denote 
$\{\varphi \in \MSOL(\Sigma,\Gamma)\mid \Free(\varphi)\subseteq\cV\}$ 
by $\MSOL(\Sigma,\Gamma,\cV)$.

Let $\cV$ be a set of node-set variables, called \emph{parameters}, and 
let $x$ and $x'$ be two distinct node variables. 
An \emph{MSO graph transducer over $(\Sigma,\Gamma,\cV)$} is a tuple $T = (\chi,D,\Psi,\Phi)$ where $\chi$ is an MSO-logic formula in $\MSOL(\Sigma,\Gamma,\cV)$ (\emph{domain formula}), $D$ is a finite set (of \emph{duplicate names}), $\Psi = (\psi_{\sigma,d}(x) \mid \sigma \in \Sigma, d \in D)$ is a family of MSO-logic formulas where each $\psi_{\sigma,d}(x)$ is in $\MSOL(\Sigma,\Gamma,\cV\cup\{x\})$ (\emph{node formulas}), and $\Phi = (\phi_{\gamma,d,d'}(x,x') \mid \gamma \in \Gamma, d,d' \in D)$ 
is a family of MSO-logic formulas where each $\phi_{\gamma,d,d'}(x,x')$ is in $\MSOL(\Sigma,\Gamma,\cV\cup\{x,x'\})$ (\emph{edge formulas}).

The \emph{graph transduction $\sem{T}\subseteq \cG_{\Sigma,\Gamma} \times \cG_{\Sigma,\Gamma}$ induced by $T$} is defined as follows. Let $g_1 = (V_1,E_1,\nlab_1)$ be in $\cG_{\Sigma,\Gamma}$,
and let $\rho$ be a $\cV$-valuation on $g_1$ such that $(g_1,\rho)\models \chi$. Then $\sem{T}$
contains the pair $(g_1,g_2)$, where the graph $g_2=(V_2,E_2,\nlab_2)$ is defined by 
\begin{itemize}
\item $V_2 = \{(d,u)   \mid d \in D, u \in V_1,  \text{ and there is exactly one } \sigma \in \Sigma \text{ such that } (g_1,\rho,u) \models \psi_{\sigma,d}(x)\}$,

\item $E_2 = \{((d,u),\gamma,(d',u'))   \mid (d,u), (d',u') \in V_2, \text{ and  } (g_1,\rho,u,u') \models \phi_{\gamma,d,d'}(x,x')\}$,

\item $\nlab_2 = \{((d,u),\sigma) \mid (d,u) \in V_2, \sigma \in \Sigma,  \text{ and } (g_1,\rho,u) \models \psi_{\sigma,d}(x)\}$.
\end{itemize}
Clearly, the graph $g_2$ is uniquely determined by $g_1$ and $\rho$, 
and will therefore be denoted by $T(g_1,\rho)$.
Thus, 
\[
\sem{T}=\{(g,T(g,\rho)))\mid g\in\cG_{\Sigma,\Gamma}, \rho \text{ is a $\cV$-valuation on } g, 
\text{ and } (g,\rho)\models \chi\} \enspace.
\]

A relation $R\subseteq \cG_{\Sigma,\Gamma} \times \cG_{\Sigma,\Gamma}$ is an \emph{MSO graph transduction} if there is an MSO graph transducer $T$ such that $R = \sem{T}$.

It is easy to see that the class of MSO-expressible graph relations is closed under inverse,
i.e., if $R$ is MSO-expressible, then $R^{-1}$ is also MSO-expressible
(just reverse the $\nu$-edges).
Since the class of MSO graph transductions is not closed under inverse, 
this shows that there exist MSO-expressible graph relations 
that are not MSO graph transductions.\footnote{The relation 
$R=\{(\edgr(w),\edgr(\epsilon))\mid w\in \Gamma^*\} 
\subseteq \cG_{\{*\},\Gamma}\times \cG_{\{*\},\Gamma}$ 
is an MSO graph transduction. Since $\sem{T}(g)$ is finite for every MSO graph transducer $T$ 
and every input graph~$g$, $R^{-1}$ is not an MSO graph transduction.
\label{foot:inverse}}

\begin{theorem}\rm \label{thm:MSO-graph-transduction-MSO-expressible}
Every MSO graph transduction is MSO-expressible. 
\end{theorem}

\begin{proof}
Let $R\subseteq \cG_{\Sigma,\Gamma} \times \cG_{\Sigma,\Gamma}$ be an MSO graph transduction, and 
let $T = (\chi,D,\Psi,\Phi)$ be an MSO graph transducer over $(\Sigma,\Gamma,\cV)$ such that $\sem{T}=R$,
with $\Psi = (\psi_{\sigma,d}(x) \mid \sigma \in \Sigma, d \in D)$ and 
$\Phi = (\phi_{\gamma,d,d'}(x,x') \mid \gamma \in \Gamma, d,d' \in D)$.
We define a formula $\phi$ in $\MSOL(\Sigma,\Gamma\cup D\cup\{\nu\})$ such that 
$\LL(\varphi)\subseteq \cP\cG_{\Sigma,\Gamma\cup D}$ and $\rel(\LL(\phi))=R$, as follows.

The set $\LL(\phi)$ consists of all pair graphs $h_{g,\rho}$ such that $(g,\rho)\models \chi$.
We will denote $g$ by $g_1$ and $T(g,\rho)$ by~$g_2$.
The graph $h_{g,\rho}$ is the disjoint union of $g_1$ and $g_2$,
with $\nu$-edges from every node of $g_1$ to every node of~$g_2$, and, moreover, 
with the following $D$-edges from nodes of $g_1$ to nodes of~$g_2$: 
there is a $d$-edge from node $u$ of~$g_1$ to node $(d,u)$ of $g_2$
if there is exactly one $\sigma \in \Sigma$ such that $(g_1,\rho,u) \models \psi_{\sigma,d}(x)$.
The $D$-edges can be called ``origin edges'', because they indicate for every node $(d,u)$ of $g_2$
that it originates from the node $u$ of $g_1$, cf. the ``origin semantics'' of 
MSO graph transductions (see, e.g., \cite{boj14,bojdavguipen17,bosmuspenpup18}). 

The formula $\phi$ is built from the formulas of $T$. We now describe the graphs $h=h_{g,\rho}$
in such a way that the existence of $\phi$ should be clear. If $\cV=\{Y_1,\dots,Y_n\}$, then $\phi$
is of the form $\exists Y_1,\dots,Y_n. \, \psi$ where $\psi$ expresses the following.
\begin{itemize}
\item The set of nodes of $h$ is partitioned into two nonempty sets $X_1$ and $X_2$ 
(node-set variables that correspond to $V_{g_1}$ and $V_{g_2}$), such that $h$ is a pair graph 
with ordered partition $(X_1,X_2)$. The sets $Y_1,\dots,Y_n$ are subsets of~$X_1$. 
\item There are no $D$-edges between nodes of $X_1$ or between nodes of $X_2$,
and there are no $\Gamma$-edges between nodes of $X_1$ and nodes of $X_2$. 
\item The subgraph of $h$ induced by $X_1$ satisfies the formula $\chi$, i.e.,
$h$ satisfies the relativized formula $\chi|_{X_1}$. 
\item For every $y\in X_2$ there are a unique $d\in D$ and a unique $x\in X_1$ such that $\edge_d(x,y)$, and 
for every $x\in X_1$ and $d\in D$ there is at most one $y\in X_2$ such that $\edge_d(x,y)$.
Thus, intuitively, $\edge_d(x,y)$ means that $y$ is the node $(d,x)$ of $g_2$ 
(or, in other words, that $x$ is the ``origin'' of $y$, and $y$ is the $d$-th duplicate of $x$).
\item For every $x\in X_1$ and $d\in D$, there is a $y\in X_2$ such that $\edge_d(x,y)$ if and only if
there is exactly one $\sigma\in \Sigma$ such that $\psi_{\sigma,d}(x)|_{X_1}$. 
\item If $\edge_d(x,y)$ and $\edge_{d'}(x',y')$, then there is a $\gamma$-edge from $y$ to $y'$ 
if and only if $\phi_{\gamma,d,d'}(x,x')|_{X_1}$.
\item If $\edge_d(x,y)$, then $y$ has label $\sigma$ if and only if $\psi_{\sigma,d}(x)|_{X_1}$. \qedhere
\end{itemize}
\end{proof}

Let us say that a storage type is \emph{MSO-definable}
if it is isomorphic to a storage type $S=(C,c_\init,\Theta,m)$ such that $C$ is an 
MSO-definable set of graphs, and 
$m(\theta)$ is an MSO graph transduction for every $\theta\in\Theta$.

\begin{corollary}\rm\label{cor:defexp}
Every MSO-definable storage type is MSO-expressible.
\end{corollary}

It is not difficult to see that the storage transformations of the MSO graph storage type STACK 
of Example~\ref{ex:stack} are in fact MSO graph transductions, with the set of duplicate names $D=\{d,d'\}$. 
For $\pop$-, $\movedown$-, and $\moveup$-operations this should be clear 
from Figures~\ref{fig:pop-alpha} and~\ref{fig:moveup-beta} (for which $d'$ is not needed). 
For $\push$-operations it should be clear after adding, in Figure~\ref{fig:push-alpha-beta}, 
an edge with label~$d'$ from the node with label~$\overline{\beta}$ 
to the node with label $\overline{\alpha}$. 
Thus, the storage type Stack of Example~\ref{ex:stackstring} is even MSO-definable. 

The results of this (sub)section can easily be generalized to ``$k$-dimensional'' MSO graph transductions
(for $k\geq 2$), which are defined in the same way as MSO graph transductions, but with the variable $x$ 
replaced by the sequence of variables $x_1,\dots,x_k$ (and similarly for $x'$). For an input graph $g_1$
and a valuation~$\rho$, the nodes of the output graph $g_2$ are now of the form $(d,u_1,\dots,u_k)$ 
such that $d\in D$ and $u_1,\dots,u_k$ is a $k$-tuple of nodes of $g_1$. 
In the proof of Theorem~\ref{thm:MSO-graph-transduction-MSO-expressible} we use edge labels 
from $D\times[k]$ instead of $D$; in a pair graph $h_{g_1,\rho}$ there is a $\tup{d,i}$-edge 
from node $u_i$ of~$g_1$ to node $(d,u_1,\dots,u_k)$ of~$g_2$ (for all $i\in[k]$) 
if there is exactly one $\sigma \in \Sigma$ such that 
$(g_1,\rho,u_1,\dots,u_k) \models \psi_{\sigma,d}(x_1,\dots,x_k)$.
Such $k$-dimensional MSO graph transductions were recently studied for strings in~\cite{bojkielho19}.

\subsection{Automatic Structures}
\label{sec:auto}

Roughly speaking, the MSO graph storage type is a generalization of the automatic structure
with binary relations only. The next definition of an automatic structure is taken from~\cite{khomin},
see also~\cite{kus09,rub08}. 

A (logical) \emph{structure} is a tuple $S=(C; R_1,\dots,R_n)$ where $C$ is a nonempty set,
called the domain of~$S$, and $R_1,\dots,R_n$ are relations on $C$, called the basic relations of $S$
(where a relation on $C$ is a subset of~$C^k$ for some $k\geq 1$, 
in which case it is called a $k$-ary relation). 
The structure $S$ is \emph{automatic} if there is an alphabet $\Omega$ such that 
its domain $C$ is a regular language over $\Omega$ and each of its basic relations $R_i$ is 
a regular relation on $\Omega^*$. This means that its domain and basic relations are recognized 
by finite automata. For relations this is defined as follows. 

To define regular relations, we need to define
the convolution of strings. Let $w_1,\dots,w_k$ be strings over $\Omega$,
let $\ell=\max\{|w_1|,\dots,|w_k|\}$ be the maximum of their lengths, and let $w'_j=w_j\#^{\ell-|w_j|}$
for every $j\in[k]$, where $\#$ is a new symbol. 
Thus, we append the symbol $\#$ to the end of $w_j$ as many times as necessary to make
the padded version $w'_j$ of $w_j$ have length $\ell$. 
The \emph{convolution of $w_1,\ldots,w_k$}, denoted 
by $\con(w_1,\dots,w_k)$, is the string $w$ over $(\Omega\cup\{\#\})^k$
of length $\ell$ such that $w(i)=\tup{w'_1(i),\dots,w'_k(i)}$ for every $i\in[\ell]$
(where $w(i)$ denotes the $i$-th symbol of $w$). 
Clearly, $\con$ is injective, i.e., $\con(w_1,\dots,w_k)$ uniquely represents the sequence $(w_1,\dots,w_k)$.

A $k$-ary relation $R$ on $\Omega^*$ is called \emph{regular} if 
its convolution $\con(R)$ is a regular language over $(\Omega\cup\{\#\})^k$.\footnote{For $k=2$ these relations are also called (left-)synchronous rational relations, cf.~\cite{car09,frosak93}. 
They are a proper subclass of the rational relations, which are the relations computed by (one-way) finite-state transducers. Recently, the relations computed by two-way deterministic finite-state transducers are called regular relations (e.g., in~\cite{bojkielho19}).}

\begin{example} \rm\label{ex:regrel}
Let $\Omega=\{\alpha,\beta,\gamma\}$, and let $R$ be the binary relation 
$\{(w,w\alpha\beta^n)\mid w\in \Omega^*,n\in\nat\}$
(intuitively, $R$ pushes an arbitrary string $\alpha\beta^n$ on the pushdown~$w$). 
Then $\con(R)$ is the regular language
\[
\{\tup{\alpha,\alpha},\tup{\beta,\beta},\tup{\gamma,\gamma}\}^*\tup{\#,\alpha}\tup{\#,\beta}^*
\]
and so, $R$ is a regular relation. 
\qed
\end{example}

\sloppy We will say that a storage type $S=(C,c_\init,\{\theta_1,\dots,\theta_n\},m)$ is \emph{automatic} 
if the structure $(C;m(\theta_1),\dots,m(\theta_n))$ is automatic.
It is straightforward to show that the storage type Stack of Example~\ref{ex:stackstring} is automatic. 
The storage type of the Turing machine is also automatic (cf.~\cite[Example~11]{khomin}).
 
The next theorem shows how MSO graph storage types generalize automatic structures with binary relations. 

\begin{theorem}\rm\label{thm:automatic}
Every automatic storage type is MSO-expressible.
\end{theorem}

\begin{proof}
In short, we define for every pair of strings $(w_1,w_2)$ a pair graph $p(w_1,w_2)$ 
such that $\pair(p(w_1,w_2))=(\edgr(w_1),\edgr(w_2))$ and the $i$-th node of $\edgr(w_1)$ has an 
intermediate edge to the $i$-th node of $w_2$ (as far as these nodes exist). Then 
there is an MSO graph transduction $f$ that translates $p(w_1,w_2)$ into 
$\edgr(\con(w_1,w_2))$. That implies, for every regular binary relation $R$, 
that the set of pair graphs $p(R)$ equals $f^{-1}(\edgr(\con(R)))$ and hence is MSO-definable.
Here is the long version of the proof.

Let $S=(C,c_\init,\Theta,m)$ be an automatic storage type, over the alphabet $\Omega$.
We first assume, for simplicity, that $\Theta$ is a singleton, i.e., $\Theta=\{\theta\}$.
We define an MSO graph storage type $S'$ that is isomorphic to~$S$. 
It is an MSO graph storage type over $(\Sigma,\Gamma)$, 
where $\Sigma=\{*\}$ and $\Gamma=\Omega\cup\{d\}$. 
It has the set of storage configurations $\edgr(C)$ and the initial storage configuration $\edgr(c_\init)$. 
Since $S$ is automatic, $C$ is a regular language over $\Omega$.
Hence $\edgr(C)$ is $\MSOL(\{*\},\Omega)$-definable by Proposition~\ref{pro:MSO-string-graph}, 
and so $\MSOL(\Sigma,\Gamma)$-definable.

For every binary relation $R\subseteq \Omega^*\times\Omega^*$ we define
\[ 
\edgr(R)=\{(\edgr(w_1),\edgr(w_2))\mid (w_1,w_2)\in R\}.
\]
Since $\edgr$ is a bijection between $C$ and $\edgr(C)$, 
it remains to define an $\MSOL(\{*\},\Omega\cup\{d,\nu\})$-definable 
set of pair graphs $\LL(\theta)\subseteq \cP\cG_{\Sigma,\Gamma}$ such that 
$\rel(\LL(\theta))=\edgr(m(\theta))$: taking the formula defining $\LL(\theta)$ as 
the unique instruction of $S'$, it is clear that $S'$ is isomorphic to $S$.
Since $S$ is automatic, $m(\theta)$ is a regular relation on $\Omega^*$.
Hence $\con(m(\theta))$ is a regular language over $(\Omega\cup\{\#\})^2$, and so 
$\edgr(\con(m(\theta)))$ is $\MSOL(\{*\},\Omega)$-definable by Proposition~\ref{pro:MSO-string-graph}. 

For strings $w_1,w_2\in\Omega^*$ we will represent the pair $(w_1,w_2)$, 
and hence their convolution $\con(w_1,w_2)$, by the pair graph $p(w_1,w_2)$ defined as follows. 
Let $w_1=a_1\cdots a_k$ and $w_2=b_1\cdots b_\ell$. 
Then $p(w_1,w_2)$ is the graph $h$ over $(\{*\},\Omega\cup\{d,\nu\})$ such that 
\begin{figure}[t]
  \begin{center}
   \includegraphics[scale=0.3]{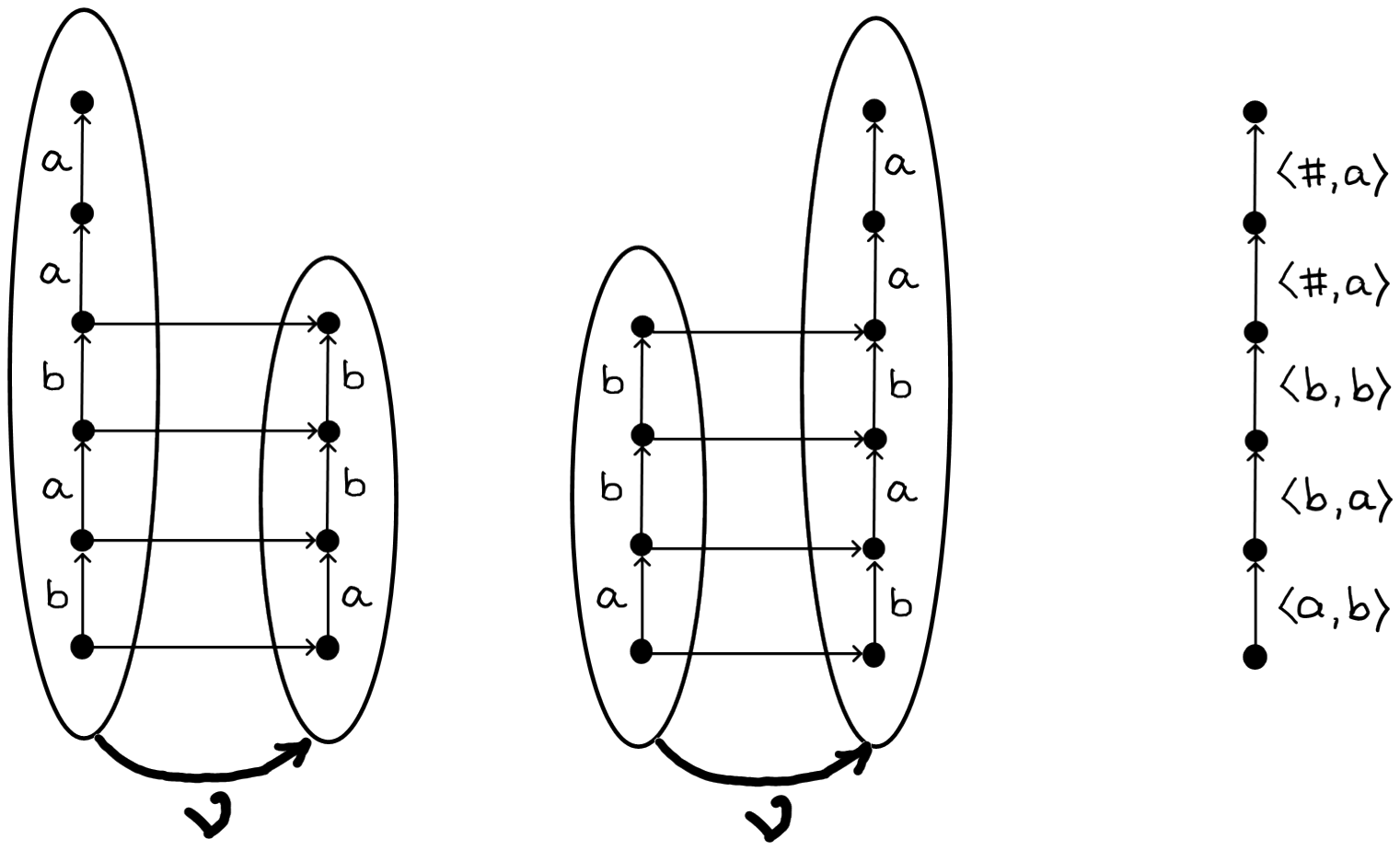}
    \end{center}
\caption{\label{fig:convol} The pair graphs $p(babaa,abb)$ (to the left) 
and $p(abb,babaa)$ (in the middle), and the string graph $\edgr(\con(abb,babaa))$ (to the right). 
All nodes have label $*$ and all straight horizontal edges have label~$d$.}
\end{figure}
$V_h=V_1\cup V_2$ where $V_1=\{u_1,\dots,u_{k+1}\}$ 
and $V_2=\{v_1,\dots,v_{\ell+1}\}$ 
and $E_h$ consists of 
\begin{itemize}
\item the edges $(u_i,a_i,u_{i+1})$ for every $i\in[k]$,
and the edges $(v_i,b_i,v_{i+1})$ for every $i\in[\ell]$,
\item the intermediate edges $(u_i,d,v_i)$ for every $i\in[\min\{k+1,\ell+1\}]$, and 
\item the edges $(u_i,\nu,v_j)$ for every $i\in[k+1]$ and $j\in[\ell+1]$.
\end{itemize}
It should be clear from the first item that $\pair(h)=(\edgr(w_1),\edgr(w_2))$, and 
from the third item that $h$ is a pair graph with the ordered partition $(V_1,V_2)$.
Two examples are shown in Figure~\ref{fig:convol}.

For every $R\subseteq \Omega^*\times\Omega^*$ we define
$p(R)=\{p(w_1,w_2)\mid (w_1,w_2)\in R\}$, and we denote 
$p(\Omega^*\times\Omega^*)$ by~$P_\Omega$.
Obviously, if $R\subseteq \Omega^*\times\Omega^*$, then 
$p(R)\subseteq P_\Omega$ and $\rel(p(R))=\edgr(R)$.

We now define $\LL(\theta)=p(m(\theta))$. So $\rel(\LL(\theta))=\edgr(m(\theta))$, as required.

Let $f$ be the function with domain $P_\Omega$ such that 
\[
f(p(w_1,w_2))=\edgr(\con(w_1,w_2))
\]
for all $w_1,w_2\in\Omega^*$. 
As an example, in Figure~\ref{fig:convol}, 
$f$ transforms the graph in the middle into the graph to the right. 
Since $\con$ and $\edgr$ are injective, 
\[
\LL(\theta)=f^{-1}(\edgr(\con(m(\theta)))).
\] 
Since $\edgr(\con(m(\theta)))$ is MSO-definable (as observed above)
and since MSO-definability is preserved by inverse MSO graph transductions
(see, e.g., \cite[Corollary~7.12]{coueng12}),
it now suffices to prove that $f$ is an MSO graph transduction
(see Section~\ref{sec:msogratra}). That is a straightforward exercise;
here are some details. 

We define an MSO graph transducer $T = (\chi,D,\Psi,\Phi)$ such that $\sem{T}=f$. 
In fact, $f$ is even a ``parameterless'' and ``noncopying'' MSO graph transduction, 
which means that its set $\cV$ of parameters is empty, and 
its set $D$ of duplicate names is a singleton. 
Hence, $D$ will be disregarded, and the duplicate names $d$ and~$d'$ 
will be dropped from the formulas of $T$.  

The MSO graph transducer $T$ is over $(\Sigma,\Gamma',\emptyset)$ where $\Sigma=\{*\}$ and 
$\Gamma'=\Omega\cup\{d,\nu\}\cup(\Omega\cup\{\#\})^2$.

It should be clear (e.g., from Examples~\ref{ex:pairgraphs} and~\ref{ex:stack}) 
that the domain $P_\Omega$ of $f$ is MSO-definable, 
by some formula~$\chi$ which we take as the domain formula of $T$. 
Let $P_{\Omega,1}=\{p(w_1,w_2)\mid |w_1|>|w_2|\}$ and $P_{\Omega,2}=\{p(w_1,w_2)\mid |w_1|\leq|w_2|\}$.
Then $P_{\Omega,1}$ is defined by the formula 
\[
\chi_1=\chi\wedge \exists x(\exists y\,.\edge_\nu(x,y)\wedge \neg \exists z.\,\edge_d(x,z)),
\]
and $P_{\Omega,2}$ by the formula $\chi_2=\chi\wedge \neg \chi_1$.

The (unique) node formula $\psi_*(x)$ of $T$ is defined by
\[
\begin{array}{lll}
\psi_*(x) & = & (\chi_1\wedge \exists y.\,\edge_\nu(x,y)) \hspace{1.5mm} \vee \\[1mm]
&& (\chi_2\wedge \exists y.\,\edge_\nu(y,x)).
\end{array}
\]
Thus, the set of nodes of $\sem{T}(p(w_1,w_2))$ is $V_1$ if $|w_1|>|w_2|$, and $V_2$ if $|w_1|\leq|w_2|$. 

For $a,b\in\Omega$, the edge formula $\phi_{\tup{a,b}}(x,x')$ is defined to be
\[
\begin{array}{l}
(\chi_1\wedge \edge_a(x,x') \wedge 
  \exists y,y'.(\edge_d(x,y)\wedge\edge_d(x',y')\wedge\edge_b(y,y'))) \hspace{1.5mm} \vee \\[1mm]
(\chi_2\wedge \edge_b(x,x') \wedge 
     \exists y,y'.(\edge_d(y,x)\wedge\edge_d(y',x')\wedge\edge_a(y,y'))),
\end{array}
\]
and the edge formulas $\phi_{\tup{a,\#}}(x,x')$ and $\phi_{\tup{\#,b}}(x,x')$ 
are defined by
\[
\begin{array}{lll}
\phi_{\tup{a,\#}}(x,x') & = & \chi_1 \wedge \edge_a(x,x')\wedge \neg \exists y'.\edge_d(x',y') \\[1mm]
\phi_{\tup{\#,b}}(x,x') & = & \chi_2 \wedge \edge_b(x,x')\wedge \neg \exists y'.\edge_d(y',x').
\end{array}
\]
This ends the definition of $T$.
It should be clear that 
\[
\sem{T}=\{(p(w_1,w_2),\edgr(\con(w_1,w_2)))\mid w_1,w_2\in\Omega^*\},
\]
i.e., $\sem{T}=f$. 

In the above proof we have assumed that $\Theta$ is a singleton. The proof for the general case
is exactly the same, except that since the sets $\LL(\theta)$ must be exclusive 
(i.e., if $\theta\neq \theta'$ then $\LL(\theta)\cap \LL(\theta')=\emptyset$), 
we replace the intermediate edge label $d$ by $d_\theta$, for each $\theta\in\Theta$. 
Thus, we now have $\Gamma=\Omega\cup\{d_\theta\mid \theta\in\Theta\}$, 
and the mappings $p$ and $f$ depend on $\theta$.  
\end{proof}

Note that not every automatic storage type $S$ is MSO-definable (as defined before Corollary~\ref{cor:defexp}).
If $S$~has an instruction $\theta$ such that $m(\theta)$ is the relation~$R$ of Example~\ref{ex:regrel},
then $S$ is not MSO-definable because for every MSO graph transducer~$T$ and every input graph $g$ 
the set of output graphs $\sem{T}(g)$ is finite (also in the $k$-dimensional case), 
cf. footnote~\ref{foot:inverse}.

A structure is \emph{automatic-representable} if it is isomorphic to an automatic structure.
(In fact, automatic-representable structures are often also called automatic structures.)
Let us say that a storage type is automatic-representable 
if it is isomorphic to an automatic storage type. 
Thus, by Theorem~\ref{thm:automatic}, 
every automatic-representable storage type is MSO-expressible.

As discussed in~\cite{khomin,rub08}, automatic structures have been generalized in the literature such that the domain~$C$ consists of trees (with the usual notion of a regular tree language defined by a finite tree automaton), and it is straightforward to generalize Theorem~\ref{thm:automatic} to that case. They have also been generalized to infinite strings and trees,  
for which regularity is defined by B\"uchi automata or Rabin automata. 
We cannot generalize Theorem~\ref{thm:automatic} to this case 
because we only consider finite graphs. 

\subsection{Iterated Pushdowns}
\label{sec:pushdown}

Let $S=(C,c_\init,\Theta,m)$ be a storage type.  
The storage type \emph{pushdown of $S$}, denoted $\rP(S)$, has 
configurations that are nonempty pushdowns of which each cell contains a pair $(\omega, c)$, where $\omega$ is a pushdown symbol in some alphabet~$\Omega$ and $c$ is a storage configuration of $S$. 
It has the instructions $\push(\omega,\theta)$, $\pop$, and $\ttop(\omega)$, 
for every $\omega\in\Omega$ and $\theta\in\Theta$, 
with the following meaning: the $\ttop(\omega)$-instruction checks that the top-most pushdown symbol 
is $\omega$, the $\pop$-instruction pops the top-most cell, and if the top-most cell contains the storage configuration~$c$ of $S$, then the $\push(\omega,\theta)$-instruction pushes a cell with content $(\omega,c')$ on the pushdown, where $c'$ is such that $(c,c') \in m(\theta)$.
We use the pushdown alphabet $\Omega=\{\alpha,\beta,\gamma\}$, 
with initial pushdown symbol~$\gamma$.  
As in Example~\ref{ex:stackstring}, we will view a pushdown configuration as a nonempty sequence 
of pairs $(\omega, c)$, such that the last pair represents the top of the pushdown. 
The bottom cell of the pushdown configuration cannot be changed; it always equals $(\gamma,c_\init)$.

Formally, we define $\rP(S)=(C',c'_\init,\Theta',m')$ where 
\begin{itemize}
\item $C'=\{c'_\init\}\cdot(\Omega\times C)^*$,
\item $c'_\init=(\gamma,c_\init)$,
\item $\Theta'= \{\ttop(\omega) \mid \omega\in\Omega\} \cup 
\{\push(\omega,\theta)\mid \omega\in\Omega, \,\theta\in\Theta\}\cup\{\pop\}$, and 
\item $m'(\theta') = m''(\theta') \cap (C'\times C')$ for every $\theta'\in\Theta'$, 
where 
\end{itemize}
for every $\omega\in\Omega$ and $\theta\in\Theta$, 
\begin{itemize}
\item $m''(\ttop(\omega))=\{(\xi(\omega,c),\xi(\omega,c))\mid \xi\in (\Omega\times C)^*, \,c\in C\}$, 
\item $m''(\pop)= \{(\xi(\omega',c),\xi)\mid \xi\in (\Omega\times C)^+, \,\omega'\in \Omega, \,c\in C\}$, and 
\item $m''(\push(\omega,\theta))= \{(\xi(\omega',c),\xi(\omega',c)(\omega,c'))\mid 
     \xi\in (\Omega\times C)^*, \,\omega'\in \Omega, \,(c,c')\in m(\theta)\}$.
\end{itemize}

Clearly, the operator $\rP$ preserves isomorphism, i.e., 
if the storage types $S$ and $S'$ are isomorphic, then so are $\rP(S)$ and $\rP(S')$.

We now prove the third general result mentioned in the introduction of this section. 

\newpage

\begin{theorem}\rm
\label{thm:P}
If $S$ is MSO-expressible, then $\rP(S)$ is MSO-expressible.
\end{theorem}

\begin{proof}
Since the operator $\rP$ preserves isomorphism, it suffices to prove that 
if~$S$ is an MSO graph storage type, then $\rP(S)$ is isomorphic to an MSO graph storage type.
Let $S=(\phi_\rc,g_\init,\Theta)$ be an MSO graph storage type over $(\Sigma,\Gamma)$. 
We will construct an MSO graph storage type $\overline{\rP}(S)$ that is isomorphic to $\rP(S)$.
The storage configurations of $\overline{\rP}(S)$ are all string-like graphs 
$g\in \cG[S,\Omega]$ over~$S$ and~$\Omega$
without $e$-edges and without $\Gamma$-edges between consecutive components. 
Thus, $\overline{\rP}(S)$ is an MSO graph storage type over $(\Sigma,\Gamma\cup\Omega\cup\{d\})$,
where $d$ is a new symbol that will be used to label some of the intermediate edges of pair graphs 
(as in the MSO storage type STACK of Example~\ref{ex:stack}). 
Without loss of generality we assume that $\Gamma\cap\Omega=\emptyset$. 

Since the initial pushdown symbol is fixed to be $\gamma$, 
a pushdown configuration $(\omega_1,g_1)(\omega_2,g_2)\cdots(\omega_{n+1},g_{n+1})$ of $\rP(S)$, 
with $n\in\nat$ and $\omega_1=\gamma$ (and $g_1=g_\init$), is uniquely
represented by the string-like graph $g$, as above, such that $\tr(g)= \omega_2\cdots\omega_{n+1}$ and 
$g[V_i]=g_i$ for every $i\in[n+1]$, where $V_i$ is the $i$-th component of $g$. 
The formula $\overline{\phi}_\rc$ that defines these string-like graphs is
\[
\overline{\phi}_\rc = \sstringlike_\Omega \wedge 
\forall x,y. (\edge_\Gamma(x,y)\to \mathrm{eq}_\Omega(x,y))
\]
where $\mathrm{eq}_\Omega(x,y)$ is the fomula $\eq(x,y)$ defined before Observation~\ref{ob:stringlike}
with $\Ae$ replaced by $\Omega$, and 
$\sstringlike_\Omega$ is the formula $\sstringlike$ of Observation~\ref{ob:stringlike} 
with $\Ae$ replaced by $\Omega$. 
This formula $\overline{\phi}_\rc$ defines the set $\overline{C}$ of storage configurations 
of~$\overline{\rP}(S)$. 
The initial storage configuration of $\overline{\rP}(S)$ is $g_\init$. 

It remains to implement the instructions of $\rP(S)$. Whenever we consider a pair graph $h$ for this purpose, 
we will assume that its ordered partition is $(V_1,V_2)$, and 
that both $h[V_1]$ and $h[V_2]$ satisfy $\overline{\phi}_\rc$.
Let $\first(x) = (\neg \exists y. \edge_\Omega(y,x))$ and $\last(x) = (\neg \exists y. \edge_\Omega(x,y))$, as in Example~\ref{ex:formula} but with $\Ae$ replaced by $\Omega$.

To implement an instruction $\ttop(\omega)$ with $\omega\in\Omega$, 
we first consider the set $\overline{C}_\omega$ of graphs $g\in\overline{C}$ that represent a configuration of $\rP(S)$ with top-most pushdown symbol $\omega$. 
Let $\phi_\omega$ be the formula  $\forall x.(\last(x)\to \exists y.\,\edge_\omega(y,x))$. 
Moreover, let $\phi'_\omega$ be the formula such that 
$\phi'_\omega = \phi_\omega$ for $\omega\neq\gamma$, and 
\[
\phi'_\gamma=\phi_\gamma\vee \forall x.(\first(x)\wedge\last(x)) \enspace.
\]
Obviously $g\in\overline{C}$ satisfies $\phi'_\omega$ if and only if $g\in\overline{C}_\omega$.
The pair graphs $h$ for the instruction $\ttop(\omega)$ are defined such that $\pair(h)=(g,g)$ 
for some $g\in\overline{C}_\omega$, as follows. There are $d$-edges from $V_1$ to $V_2$ 
that establish an isomorphism between $h[V_1]$ and $h[V_2]$
(as in Examples~\ref{ex:pairgraphs} and~\ref{ex:stack}), and $h[V_1]$ satisfies $\phi'_\omega$.
It should be clear that this set of pair graphs is MSO-definable. 

To implement the pop-instruction, we consider all pair graphs $h$ with $d$-edges from $V_1$ to $V_2$ 
that establish an isomorphism between $h[V_1\setminus T_1]$ and $h[V_2]$, 
where $T_1\subseteq V_1$ is the last component of the string-like graph $h[V_1]\in \overline{C}$.

Finally we implement an instruction $\push(\omega,\theta)$ with $\omega\in\Omega$ and $\theta\in\Theta$.
Symmetrically to the previous case, each pair graph $h$ has $d$-edges from $V_1$ to $V_2$ 
that establish an isomorphism between $h[V_1]$ and $h[V_2\setminus T_2]$ 
where $T_2$ is the last component of $h[V_2]$. Moreover, as in the first case,
$h[V_2]$ satisfies the formula $\phi'_\omega$ (or equivalently, $\phi_\omega$). 
In addition to the $d$-edges, $h$ has intermediate $\Gamma$-edges
between $T_1$ and~$T_2$ (where $T_1$ is the last component of $h[V_1]$, as before) such that 
the pair graph $h[T_1\cup T_2]$ satisfies~$\theta$. 
These $\Gamma$-edges ensure that $(h[T_1],h[T_2])\in\rel(\LL(\theta))$ and hence, 
since $h[T_1]$ is isomorphic to $h[T'_2]$, 
where $T'_2$ is the one-before-last component of $h[V_2]$,
that $(h[T'_2],h[T_2])\in\rel(\LL(\theta))$ as required
(see the next paragraph for an example).
\end{proof}

\begin{figure}[t]
  \begin{center}
    \includegraphics[scale=0.39,trim={7cm 3cm 9cm 0},clip]{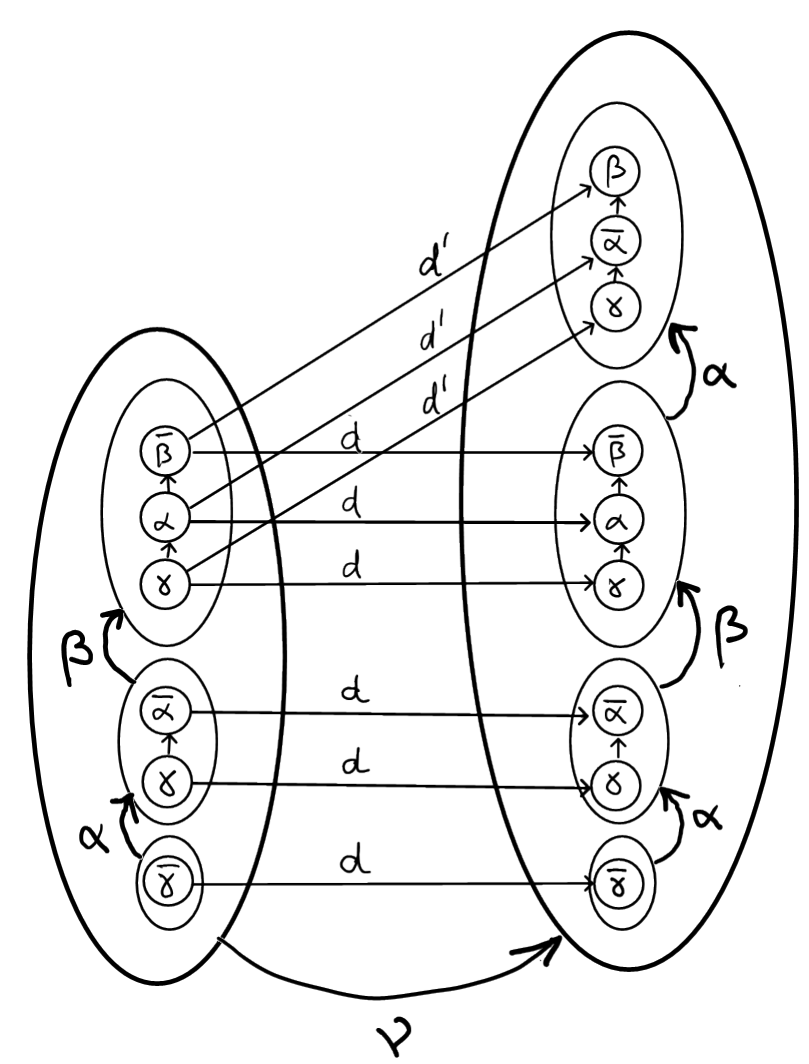}
    \end{center}
\caption{\label{fig:pd-stacks-J} A pair graph for the instruction $\push(\alpha,\movedown(\beta))$
as implemented in the MSO graph storage type $\overline{\rP}(\mathrm{STACK})$.}
\end{figure}

In Figure \ref{fig:pd-stacks-J} a pair graph is shown for the instruction $\push(\alpha,\movedown(\beta))$
as implemented in the MSO graph storage type $\overline{\rP}(\mathrm{STACK})$, 
where the MSO graph storage type $\mathrm{STACK}$ is defined in Example~\ref{ex:stack}. 
We assume here that the edge label alphabet of $\mathrm{STACK}$ is $\Gamma=\{*,d'\}$
(instead of $\Gamma=\{*,d\}$), because we use the new symbol $d$ in the edge label alphabet 
of $\overline{\rP}(\mathrm{STACK})$. 
For this pair graph $h$ we have, in the notation of the above proof, 
$h[T_1]=h[T'_2]=\ndgr(\gamma \alpha \overline{\beta})$ and 
$h[T_2]=\ndgr(\gamma \overline{\alpha} \beta)$.
Note that the storage configuration of $\overline{\rP}(\mathrm{STACK})$
in the first component of the pair graph can be reached from the initial storage configuration $\ndgr(\overline{\gamma})$ by the two consecutive instructions $\push(\alpha,\push(\alpha))$ and
$\push(\beta,\push(\beta))$.

For $n\in\nat$ we define the \emph{$n$-iterated pushdown} to be the storage type $\rP^n$, 
such that $\rP^0=\Triv$, as defined in Example~\ref{ex:triv}, and $\rP^{n+1}=\rP(\rP^n)$. 
The trivial storage type $\Triv$ is MSO-expressible by Example~\ref{ex:MSO-graph-storage-triv}. 
Hence, the next corollary is immediate from Theorem~\ref{thm:P}.

\begin{corollary}\rm \label{cor:Pn}
For every $n\in\nat$, the storage type $\rP^n$ is MSO-expressible.
\end{corollary}

It is not difficult to see from the proof of Theorem~\ref{thm:P} that if $S$ is MSO-definable
(as defined in Section~\ref{sec:msogratra}), then the storage transformations of $\overline{\rP}(S)$
are MSO graph transductions, and so $\rP(S)$ is MSO-definable too. Obviously, $\Triv$ is also MSO-definable.
Hence the iterated pushdown storage types $\rP^n$ are even MSO-definable.

\section{Conclusion} 

We have considered a specific kind of (finitely encoded) storage types of automata, 
the MSO graph storage types.
Essentially, they are storage types of which each storage configuration is a graph, and 
each instruction executes a storage transformation that is an MSO-expressible graph relation, 
as defined in Section~\ref{sec:pairgraphs}. For every MSO graph storage type $S$ 
(and every alphabet $A$) we have designed an appropriate logical language $\MSOL(S,A)$ 
on string-like graphs, and we have proved a BET-theorem relating the languages over $A$ 
that are recognized by $S$-automata to those that can be expressed by a formula of $\MSOL(S,A)$.
We observe here that it is straightforward to extend the results of this paper to MSO graph storage types 
that, additionally, have an MSO-definable set of final configurations.

The notion of an MSO-expressible graph relation seems to be new, and needs further investigation.
The class of MSO-expressible graph relations seems to be quite large.
It is easy to see that it is closed under inverse (as mentioned already in Section~\ref{sec:msogratra}),
and under union and intersection (i.e., if $R_1$ and $R_2$ are MSO-expressible, 
then so are $R_1\cup R_2$ and $R_1\cap R_2$).
It is also straightforward to show that it contains (string) relations such as $\{(\edgr(a^n),\edgr(a^{2^n}))\mid n\in\nat\}$, 
$\{(\edgr(a^n),\edgr(a^{kn}))\mid n,k\in\nat\}$, and 
$\{(\edgr(a^nb^n),\edgr(a^nb^n))\mid n\in\nat\}$.\footnote{For instance, 
the pair graphs that define the relation 
$\{(\edgr(a^n),\edgr(a^{2^{n+1}-1}))\mid n\in\nat\}$ have intermediate $d$-edges from component $V_1$ to component $V_2$ satisfying the following two conditions: 
(1) for every $X\subseteq V_1$ there exists $v\in V_2$ such that $\inn_d(v)=X$, and 
(2) for every $v,v'\in V_2$, if $v\neq v'$ then $\inn_d(v)\neq\inn_d(v')$.
The relation $\{(\edgr(a^n),\edgr(a^{2^n}))\mid n\in\nat\}$ can be defined similarly, ignoring the last node of each of the two string graphs. 
The pair graphs that define the relation 
$\{(\edgr(a^n),\edgr(a^{k(n+1)-1}))\mid n,k\in\nat\}$ have $d$-edges 
from $V_2=\{v_1,\dots,v_m\}$ to $V_1=\{u_1,\dots,u_{n+1}\}$ such that
(1)~$\out_d(v_1)=\{u_1\}$,
(2)~for every $i\in[m]$ and $j\in[n+1]$, if $\out_d(v_i)=\{u_j\}$, then $\out_d(v_{i+1})=\{u_{f(j)}\}$, 
where $f(j)=j+1$ for $j\in[n]$ and $f(n+1)=1$, and 
(3)~$\out_d(v_m)=\{u_{n+1}\}$.
The pair graphs that define the relation 
$\{(\edgr(a^mb^n),\edgr(a^mb^n))\mid m,n\in\nat\}$
have intermediate $d$-edges establishing a bijection between $V_1$ and $V_2$ 
(as in Example~\ref{ex:pairgraphs})
such that corresponding nodes have the same label, where we define the label of a node to be the label of its outgoing $\{a,b\}$-edge. The pair graphs for $\{(\edgr(a^nb^n),\edgr(a^nb^n))\mid n\in\nat\}$ additionally have intermediate $d'$-edges establishing a bijection between the $a$-labeled nodes of $V_1$ 
and the $b$-labeled nodes of $V_2$, which implies $m=n$.}
Generalizing the last example, 
we can rather easily conclude from the BET-theorem of~\cite{lauschthe94} 
(and from Example~\ref{ex:pairgraphs}) that if~$L$ is a context-free language, 
then the identity on $\edgr(L)$ is MSO-expressible:  
the matching edges between the positions of the string can be simulated 
by intermediate edges in the corresponding pair graph. Hence the same holds if $L$ is 
an intersection of finitely many context-free languages, such as $L=\{a^nb^nc^n\mid n\in\nat\}$.  
On the other hand there are simple graph relations that are not MSO-expressible, such as 
$\{(\edgr(a^nb^n),\edgr(\epsilon))\mid n\in\nat\}$. In fact, it is rather easy to show that if $L$ is a 
language such that $\{(\edgr(w),\edgr(\epsilon))\mid w\in L\}$ is MSO-expressible, then $L$ is regular
(cf. Claim~18 in Appendix~A of~\cite{bojkielho19}): the intermediate edges of a pair graph $h$ such that 
$\pair(h)= (\edgr(w),\edgr(\epsilon))$ can be coded as additional labels of the nodes of $\edgr(w)$.  
This also shows that the class of MSO-expressible graph relations is not closed under composition
(cf.~\cite[Theorem~4.1]{bojkielho19} and its proof in Appendix~A of~\cite{bojkielho19}): 
composing the identity on $\edgr(\{a^nb^n\mid n\in\nat\})$
with the MSO-expressible relation $\{(\edgr(w),\edgr(\epsilon))\mid w\in \{a,b\}^*\}$ produces the above
non-MSO-expressible relation.

We have not been able to find an example of a (finitely encoded) storage type
that is \emph{not} isomorphic to an MSO graph storage type, i.e., that is not MSO-expressible
(as defined in Section~\ref{sec:MSO-expres}).\footnote{Of course we want such an example to be 
a storage type $S=(C,c_\init,\Theta,m)$ such that $C$ is a countable set and $m(\theta)$ is 
a partially computable binary relation for every $\theta\in\Theta$.
}
In the literature (e.g., \cite{engvog86,eng91c}),
the \emph{equivalence} of storage types is defined in such a way that 
(1)~isomorphic storage types are equivalent, 
(2)~equivalent storage types are language equivalent, and 
(3)~the pushdown operator $\rP$ preserves equivalence. 
Suppose that we would redefine a storage type to be MSO-expressible if it is equivalent 
(rather than isomorphic) to an MSO graph storage type. 
Then Theorem~\ref{thm:main-string} can still be viewed as a BET-theorem for MSO-expressible storage types,
and Theorem~\ref{thm:P} still holds. So now the even harder question is: 
are there examples of storage types that are \emph{not} equivalent to an MSO graph storage type?

A similar question is whether there exist MSO graph storage types 
(or even MSO-definable storage types, as defined in Section~\ref{sec:msogratra})
that are not automatic-representable (as defined in Section~\ref{sec:auto}).

It follows from Lemma~\ref{lm:twolang} that for every $S$-automaton $\cA$, 
$\LL(\cA)=\emptyset$ if and only if $\GLL(\cA)=\emptyset$. 
Since $\GLL(\cA)$ is MSO-definable,
and since for every closed MSO-logic formula $\phi$ and every $k\in\nat$ it is decidable whether 
$\LL(\phi)$ contains a graph of tree-width at most $k$ (cf.~\cite[Corollary~5.81]{coueng12}), 
it follows from this equivalence that the emptiness problem 
for $S$-automata over $A$ is decidable if the set $\{g\in \cG[S,A] \mid g\models \beh\}$ 
is of bounded tree-width. Unfortunately, this result is not very helpful: even for the pushdown storage type  this set contains all rectangular grids as subgraphs (disregarding labels) and hence is not of bounded tree-width, cf.~\cite[Corollary~2.60(1) and Example~2.56(3)]{coueng12}.\footnote{For the $2\times 5$ grid see Figure~\ref{fig:comput-intro}. See also Figure~\ref{fig:comput} for the $3\times 5$ grid for $S=\mathrm{Stack}$.} 
In~\cite{madpar11} this idea is successfully applied to an alternative definition of $\GLL(\cA)$ for
various restrictions of multi-pushdown and multi-queue automata. For an \mbox{$n$-pushdown} automaton~$\cA$, i.e., an automaton with $n$ independent pushdowns, the alternative $\GLL(\cA)$ consists of string graphs augmented with edges that model $n$ matchings, such that each matching corresponds to the pushes and pops of $\cA$ on one of its pushdowns. This generalizes the case $n=1$ in~\cite{lauschthe94}. 
It is shown in~\cite{madpar11} that the graph language $\GLL(\cA)$ is MSO-definable 
and that $\LL(\cA)=\emptyset$ if and only if $\GLL(\cA)=\emptyset$ (but no BET-theorem is proved). 
Moreover, for a number of restrictions on the behaviour of $\cA$ it is shown that $\GLL(\cA)$ is 
of bounded tree-width, and hence its emptiness is decidable.
Since for $n=1$ the set of string graphs with one matching has tree-width~2, 
this includes the case of (unrestricted) \mbox{1-pushdown} automata.  
Similar results are shown for multi-queue automata.
We leave it as an open problem to find an alternative definition of $\GLL(\cA)$ for which the above strategy 
is applicable to $S$-automata $\cA$ under certain natural conditions on $S$ that would be satisfied by most 
of the known storage types with a decidable emptiness problem 
(such as the $n$-iterated pushdown, see~\cite[Theorem~7.8]{dam82} and~\cite[Theorem~7.12]{eng91c}).
We finally note that a completely different strategy is investigated in~\cite{zet17}.

\acknowledgements We wish to thank Helmut Seidl for pointing out 
the relationship to automatic structures.
We are grateful to the reviewers for their constructive comments.


\end{document}